\documentclass{article}
\usepackage[utf8]{inputenc}
\usepackage{amsmath}
\usepackage{amssymb}
\usepackage{physics}
\usepackage{dsfont}
\usepackage{amsthm}
\usepackage{amsmath}
\usepackage{algorithm}
\usepackage{algpseudocode}
\usepackage{booktabs}
\usepackage{hhline}
\usepackage{enumerate}
\usepackage{subcaption}
\usepackage{graphicx}
\usepackage{makecell}
\usepackage{bold-extra} 
\usepackage{authblk}
\usepackage[normalem]{ulem}
\usepackage[nottoc]{tocbibind}
\usepackage{cite}
\renewcommand{\v}[1] 

\usepackage[hidelinks]{hyperref}
\hypersetup{colorlinks=false}

\makeatletter
\newenvironment{breakablealgorithm}
  {
   \begin{center}
     \refstepcounter{algorithm}
     \hrule height.8pt depth0pt \kern2pt
     \renewcommand{\caption}[2][\relax]{
       {\raggedright\textbf{\fname@algorithm~\thealgorithm} ##2\par}%
       \ifx\relax##1\relax 
         \addcontentsline{loa}{algorithm}{\protect\numberline{\thealgorithm}##2}%
       \else 
         \addcontentsline{loa}{algorithm}{\protect\numberline{\thealgorithm}##1}%
       \fi
       \kern2pt\hrule\kern2pt
     }
  }{
     \kern2pt\hrule\relax
   \end{center}
  }
\makeatother

\floatname{algorithm}{Algorithm}

\newtheorem{theorem}{Theorem}[section]
\newtheorem{problem}[theorem]{Problem}
\newtheorem{corollary}[theorem]{Corollary}
\newtheorem{lemma}[theorem]{Lemma}
\newtheorem{remark}[theorem]{Remark}
\newtheorem{definition}[theorem]{Definition}
\newtheorem{example}[theorem]{Example}

\usepackage[margin=1in]{geometry}
\usepackage{xcolor}
\usepackage{comment}
\newcommand{\USP}{\textrm{USP}}
\newcommand{\Comp}{\textrm{Comp}}
\newcommand{\mcalN}{\mathcal{N}}
\newcommand{\SWAP}{\textrm{SWAP}}
\newcommand{\UNCOMP}{\textrm{UNCOMP}}
\newcommand{\PHASE}{\textrm{PHASE}}
\newcommand{\Phase}{\textrm{Phase}}
\newcommand{\PREP}{\textrm{PREP}}
\newcommand{\SEL}{\textrm{SEL}}

\renewcommand{\v}[1]{\boldsymbol{#1}}

\title{Quantum sampling algorithms for quantum state preparation and matrix block-encoding}

\author{Jessica Lemieux} 
\author{Matteo Lostaglio}
\author{Sam Pallister}
\author{William Pol}
\author{Karthik Seetharam}
\author{Sukin Sim}
\author{Burak \c{S}ahino\u{g}lu\footnote{Correspondence to bsahinoglu@psiquantum.com.}}

\affil{\emph{PsiQuantum, 700 Hansen Way, Palo Alto, CA 94304}}

\date{\today}

\begin{document}

\maketitle

\begin{abstract}
The problems of quantum state preparation and matrix block-encoding are ubiquitous in quantum computing:
they are crucial parts of various quantum algorithms for the purpose for initial state preparation as well as loading problem relevant data.
We first present an algorithm based on quantum rejection sampling (QRS) that prepares a quantum state $|\psi_f\rangle \propto \sum^N_{x=1} f(x)|x\rangle$. When combined with efficient reference states the algorithm reduces the cost of quantum state preparation substantially, if certain criteria on $f$ are met.
When the preparation of the reference state is not the dominant cost, and the function $f$ and relevant properties are efficiently computable or provided otherwise with cost $o(N)$, the QRS-based method outperforms the generic state preparation algorithm, which has cost $\mathcal{O}(N)$.
We demonstrate the detailed performance (in terms of the number of Toffoli gates) of the QRS-based algorithm for quantum states commonly appearing in various quantum applications, e.g., those with coefficients that obey power law decay, Gaussian, and hyperbolic tangent, and compare it with other methods.
Then, we adapt QRS techniques to the matrix block-encoding problem and introduce a QRS-based algorithm for block-encoding a given matrix $A = \sum_{ij} A_{ij} \ketbra{i}{j}$. 
We work out rescaling factors for different access models, which encode how the information about the matrix is provided to the quantum computer.
We exemplify these results for a particular Toeplitz matrix with elements $A_{\v{ij}}= 1/\|\v{i}-\v{j}\|^2$, which appears in quantum chemistry, and PDE applications, e.g., when the Coulomb interaction is involved.
Our work unifies, and in certain ways goes beyond, various quantum state preparation and matrix block-encoding methods in the literature, and gives detailed performance analysis of important examples that appear in quantum applications.
\end{abstract}

\newpage

\tableofcontents

\newpage

\newpage

\section{Introduction}

The idea of sampling probability distributions from other readily available ones dates back at least to von Neumann, e.g., Ref.~\cite{von195113}.
In particular, the accept-reject method, also known as the rejection sampling method, has been widely used either by itself or as a part of many algorithms in statistics to sample from target random variables given access to a `reference random variable' that is simpler to generate or sample from.
Here we consider quantum versions of this algorithm for the task of preparing target quantum states and matrix block-encoding, two of the most widespread subroutines of quantum algorithms for a fault-tolerant quantum computer. 
In a nutshell, the idea is to identify `reference quantum states' that are simpler to prepare and use them to prepare a target quantum state. We also show that a similar idea applies to matrix block-encodings.\\

The core ideas behind quantum rejection sampling (QRS) have been implicitly incorporated in quantum computing since the early days, notably starting with Grover's work on unstructured search~\cite{grover1996fast, grover1997quantum}.
Given an oracle that marks a target bit string $\tilde{x} \in \{0,1\}^n$, Grover's algorithm evolves a quantum state from an equal weight superposition over all strings, $\frac{1}{\sqrt{N}} \sum_{x=1}^{N}\ket{x}$, towards the target string $\ket{\tilde{x}}$.
This can be seen as a quantum version of the problem of generating a probability distribution peaked at the target string $\tilde{x}$ from the uniform probability distribution over all strings.
The concept extends beyond simple search problems to broader problems where a quantum algorithm's output can be amplified to give the correct result, using amplitude amplification~\cite{brassard2002quantum}.
Many quantum algorithms in the literature, e.g., to solve linear systems or to prepare thermal states etc., make use of amplitude amplification to amplify the `good subspace' where the solution lives in, and discard the rest~\cite{harrow2009quantum, chowdhury2016quantum, holmes2022quantum}. This connection was further emphasized in Ref.~\cite{ozols2013quantum}.
As we shall see, similar considerations apply to matrix block-encoding, where several prior works~\cite{babbush2019quantum, berry2009black, childs2010relationship, wossnig2018quantum, gilyen2019quantum, su2021fault, nguyen2022block} involve algorithms that can be seen as special instances of a QRS-based method. 
While one can interpret several prior protocols as special instances of a QRS-based method, they have not been developed as such and hence rejection sampling ideas have not been systematically applied in practice to optimize these tasks. 
In this work, we put forward a general framework to systematically apply a QRS-based approach to state preparation and matrix block-encoding, and illustrate the advantages of doing so via concrete examples. \\

We start by introducing a general QRS protocol for the quantum state preparation problem. 
Given a family of sets of complex numbers $\{f(x) \in \mathbb{C}\}^{N}_{x=1}$ with increasing $N \in \mathbb{N}^+$, our goal is to  prepare the quantum states
\begin{align}
\label{eq:quantumstate}
\frac{1}{\mathcal{N}_f} \sum^N_{x=1} f(x)\ket{x},
\end{align}
\noindent where $\mathcal{N}_f:= \sqrt{\sum^N_{x=1} |f(x)|^2}$ is the normalization factor, as efficiently as possible.
Note that here and in the rest of the paper we specify the domain to be $[1,N]:= \{1, \ldots, N\}$. This can be extended to different choices of
domain via a relabeling.

Solving this problem requires, in the worst case, $\mathcal{O}(N)$ arbitrary angle single-qubit rotations, using methods such as those in Refs.~\cite{grover2002creating, mottonen2004transformation}. Note that these methods also require a classical pre-processing component, as we need to find the rotation angles $\{\theta_i\}_i$ from the coefficients $f(x)$ 
via efficient computation of normalization factors such as $\sum_{x \in D} |f(x)|^2$ over domains $D$ equal to $[1, 2], [3,4], \ldots, [N-1,N], [1,\dots,4], [5,\dots, 8], \ldots$, $[N - 3, \dots, N], \ldots, [1, \dots, N/2], [N/2 +1, \dots, N]$, the ratios of those, and the computation of $\arcsin$.
This can be achieved either via exact computation, an analytical formula (if available), or various sampling or Monte Carlo methods, depending on the function $f$.
One then computes $N/2$ angles from these ratios.
These are the angles that the quantum computer uses in order to perform state preparation.
For example, Ref.~\cite{kitaev2008wavefunction} applied this method to the case where $f$ is a Gaussian with a given mean and standard deviation, in which case $|f|^2$ is efficiently integrable/summable.
When arbitrary state preparation is considered, and only Toffoli or T gates are counted, the worst case scaling has been further improved to $\tilde{\mathcal{O}}(\sqrt{N})$ with more recent methods that introduce $\tilde{\mathcal{O}}(\sqrt{N})$ extra ancilla qubits, as in Ref.~\cite{low2018trading}.
In fact, Ref.~\cite{low2018trading} further shows that this scaling is the best one can hope for asymptotically in terms of number of $T$-gates.
More precisely it has been shown to take $\Omega(\sqrt{N \log N \log(1/\epsilon)})$ $T$-gates to prepare an arbitrary state up to error $\epsilon$.

However, the problem simplifies drastically for certain interesting cases, such as when $f$ is a constant function, or more generally when there is structure in $f$ that can be exploited. 
One approach to the problem is to use matrix product states, which appeared in Refs.~\cite{garcia2021quantum, holmes2020efficient}.
These methods use no amplitude amplification, and construct a direct unitary circuit that outputs the target state. 
While they may be considered efficient, they require solving classical optimization problems, use involved coherent quantum arithmetic operations, or require high-precision gate synthesis that may eventually make the methods impractical.
This is a promising direction and it can work particularly well for some functions, but we do not pursue it further in this work.

Then, there are other types of methods that use QRS ideas.   
An early work by Grover~\cite{grover2000synthesis}, in fact, already uses a uniform superposition state and flags components to be coherently `accepted/rejected' via amplitude amplification.
In this algorithm, a coherent rotation depending on the function values $f(x)$ are used, and in particular a subroutine for calculating $\arcsin$ function is employed.
Built on this work, Ref.~\cite{sanders2019black} replaces the coherent rotation with a sampling state and comparator, and hence reduced the cost of this protocol.
This was further improved in Ref.~\cite{bausch2022fast}, where the uniform superposition state is replaced by an `amplitude-gradient state'.
This improves the query complexity (and total gate complexity) for certain target states, but falls short for other interesting target states such as when $f$ is Gaussian and high precision is required to capture the tails.
More recently, an approach that uses the quantum singular value transform to modify function profiles has been introduced to decrease the qubit count, and it particularly works well in terms of gate count if the function is known to be approximated by a low-degree polynomial, e.g., see Ref.~\cite{mcardle2022quantum}.\\

In this work, instead of preparing an initial uniform superposition state, as in Ref.~\cite{grover2000synthesis, sanders2019black}, or replacing it with an amplitude gradient state, as in Ref.~\cite{bausch2022fast}, we follow a more systematic  QRS approach, by designing an initial `reference state' for each given target state.
This reduces the number of amplitude amplification steps needed, as we shall see.
When presenting the general method, we assume an efficient unitary that generates the reference state
\begin{align}
\frac{1}{\mathcal{N}_g}\sum^N_{x=1} g(x) \ket{x},
\end{align}
where, for simplicity, we usually opt for a domain $x \in [1,N]$, and otherwise use a more general domain $x\in D$.
Furthermore, we assume that the unitaries $U_{\tilde{g}}: \ket{x} \otimes \ket{0} \mapsto \ket{x} \otimes \ket{\tilde{g}(x)}$ (and similarly $U_{\tilde{f}}$) are efficient to implement, where
$\tilde{f}$ and $\tilde{g}$ are functions related to the original functions $f$ and $g$, respectively.\footnote{For simplicity, we work with $\tilde{g}= g$, and $\tilde{f}= f$ when we first introduce the algorithm, and we give the general version in Appendix~\ref{app:GeneralClause}.}
Reference states can be chosen in various ways and can be designed such that we only require $\mathcal{O}(1)$ rounds of amplitude amplification in many interesting cases.
The QRS-based algorithm also uses an ancilla register of dimension $M$ that depends on the choice of $\tilde{f}$, $\tilde{g}$ and how precisely we wish to obtain the target state.
In a nutshell, this state preparation method closely follows a quantum implementation of the rejection sampling method, and ensures good performance by selecting an efficiently preparable initial reference state - when combined with further constraints on the functions for gate efficiency of the subroutines.\\

A well-known related problem to quantum state preparation is the block-encoding problem, i.e., the encoding of a given matrix $A$ as a submatrix of a unitary, up to a rescaling factor $\alpha$

\begin{align}
\label{eq:blockencodingintro}
    U_{A/\alpha} = \begin{bmatrix} A/\alpha & * \\ * & *
    \end{bmatrix},
\end{align}
where $*$ indicates any choice ensuring that $U_{A/\alpha}$ is unitary.
Note that there are various ways to generate the unitary $U_{A/\alpha}$ which are intimately related to the access model for the matrix $A$, or the access model is chosen such that the rescaling factor $\alpha$ is smallest and the gate complexity of $U_{A/\alpha}$ is minimal. 
Many quantum algorithms for problems ranging from Hamiltonian simulation, solving linear systems of equations, solving partial differential equations, filtering, estimating average values of observables, etc., are constructed under the assumption that some matrix $A$ can be embedded in a unitary as in Eq.~\eqref{eq:blockencodingintro}. 
For a sample of the vast literature, see e.g.~\cite{gilyen2019quantum, low2019hamiltonian, costa2022optimal, jennings2023efficient, berry2017quantum, an2022theory, jennings2023cost, lin2020optimal, rall2020quantum, martyn2021grand} and references therein. 
Correspondingly, numerous constructions of matrix block-encodings have appeared in the literature, see, e.g., Refs.~\cite{berry2009black, childs2012hamiltonian, berry2014exponential, nguyen2022block, camps2022explicit, sunderhauf2023block}.

The common underlying property in all of these approaches, similar to the quantum state preparation problem, is that efficient block-encoding of matrices is possible when certain structures can be exploited, e.g., low sparsity, low rank, hierarchical decomposition, structurally repeated elements (such as Toeplitz matrices), etc.
Importantly, this ``structure'' should be exploitable via some underlying efficient quantum state preparation subroutine.\\

In this work, we introduce a framework to systematically apply a QRS approach to this problem. 
We shall assume that the block-encoding of a reference matrix $G$ can be efficiently constructed, and use this ability to construct a block-encoding of $A$. 
For example, from this perspective the standard sparse access block-encoding construction is a special case in which the reference matrix $G$ is uniform over the nonzero elements of $A$. 
Our general circuit covers the block-encoding circuits in the literature, gives a unifying understanding of these solutions, and goes beyond them as it introduces the idea of tailoring the reference matrix to each specific target block-encoding. \\

The article is organized as follows: In Section~\ref{sec:AlgoState}, we formally define the state preparation problem, the elementary subroutines to be used and assumptions that renders our algorithm efficient. 
We introduce the quantum state preparation algorithm in detail, prove that it works, and calculate the cost in terms of how many times the subroutines are called.
We then lay out out general principles for designing efficient reference functions~$g$ from which we sample the target.
In Section~\ref{sec:Examples}, we work out a few example functions that are commonly used, such as $f(x)= 1/x^{\beta}$ for $\beta \in \mathbb{R}$, the Gaussian $f(x)= \exp{-(x-\mu)^2/(2 \sigma^2)}$, and the hyperbolic tangent $f(x)= \tanh(x)$, together with reference states to be used, such as those associated with piecewise constant, and exponential functions.
This includes designing the reference functions $g$, calculating the resulting number of amplitude amplification steps, and computing the exact gate count for a fully specified function within a range of parameters for the target precision $\epsilon$ and the size $N$ of the domain.
In Section~\ref{sec:Comparison}, we compare the QRS-based method with other methods, which we introduced above in general terms, discussing what type of functions are expected to be more suitable for each method.
Ultimately, for end users of our algorithm and others, the choice of the most suitable method requires accounting for the constant factors involved on a case-by-case basis. 
We present a comparison of our QRS-based method and the LKS method (Ref.~\cite{low2018trading}) for the functions analyzed in Section~\ref{sec:Examples}.
Making fair comparisons is a subtle topic that needs to be addressed in detail, including accounting for the number of ancilla qubits required. 
For example, the LKS method optimizes $T$-count at the expense of increasing the qubit count, so a comparison involving only the former does not give a full picture.
While we perform the resource counts in terms of Toffoli gates, we comment on these and other issues for comparing different methods.
In Section~\ref{sec:QSamplingUnitaries}, we first introduce our general circuit for matrix block-encoding, and compute the rescaling factor $\alpha$ and the circuit subroutines for carrying out coherent rejection sampling in terms of the properties of the quantum states that are prepared to encode the reference matrix. 
We then specialize the general result to different access models that exploit sparsity, Toeplitz structure, submatrix structures, row-column, and column data, respectively, and show how the general circuit is modified so that the block-encoding specializes to these cases.
In Section~\ref{sec:ExampleBE}, we illustrate the performance of these access models for the Coulomb kernel in three dimensions, for which several approaches give comparable scalings.
Finally, in Section~\ref{sec:Conclusions}, we give a summary of our results and offer an outlook for potential future work.

\section{A quantum sampling method with designed reference states}\label{sec:AlgoState}

In this section, we first state the definitions of the problem and subroutines used, with their efficiency requirements (Section~\ref{subsec:DefinitionsAndAssumptionsStatePrep}).
Then we introduce the quantum circuit that is built using these subroutines, prove that it works, and give the cost estimate in terms the number of calls to subroutines (Section~\ref{subsec:Algorithm}).
The gate complexity of the circuit depends highly on the design choices made for the reference state/function, so at the end of this section, we lay out some guidelines that will be useful to keep the gate complexity as low as possible (Section~\ref{subsec:DesignGuidelines}).

\subsection{Definitions and assumptions}\label{subsec:DefinitionsAndAssumptionsStatePrep}

\begin{problem}[Quantum state preparation]\label{problem:QuantumStatePreparation}
Let $\{f(x) \in \mathbb C\}^N_{x=1}$, with $f(x)= |f(x)| e^{i\varphi(x)}$ and $\varphi(x) \in [0,2\pi)$, be a family of sets of complex numbers with increasing $N \in \mathbb{N}^+$, and let $\epsilon \in \mathbb{R}^+$ be a given error parameter such that $0 < \epsilon < 1$.
The quantum state preparation problem is the problem of preparing the family of quantum states $\{\ket{\psi}\}_N$ that are an $\epsilon$-approximation to the target states $\ket{\psi_f}:= \frac{1}{\mcalN_f} \sum^N_{x=1} f(x) |x\rangle$, i.e., for every $N$,
\begin{align}
\|\ket{\psi} - \ket{\psi_f}\| \leq \epsilon
\end{align}
where $\|\cdot\|$ is the Euclidean norm.
\end{problem}
Ideally, for each fixed $N$ the goal is to find the circuit with minimal gates.
In general, though, we design the algorithm with the family of problems in mind.
Note that the problem is mainly interesting for large enough $N$.

We start with defining the subroutines, which we assume are available and efficient to implement.
We make use of subroutines $\PREP_g$ given a function $g$, $U_{f}$ given a function $f$, Comparator  ($\Comp$), and Uniform State Preparation ($\USP$), which are defined below. 
In terms of notation, recall that $[1,N]:=\{1,\dots,N\}$. We also denote by $\ket{0}$ the product state of a number of qubits in the zero state. 
For example, for an $N$-dimensional system, $\ket{0}$ is a shorthand for $\ket{0}^{ \otimes \lceil \log_2 N \rceil}$.

\begin{definition}[$\PREP_g$: The unitary that prepares the reference state $\ket{\psi_g}$]\label{def:PREPg}
Let $g:[1,N] \mapsto \mathbb{R}^+$ be a nonnegative function.
Then ${\PREP}_g$ is defined to be a unitary that generates the state $\ket{\psi_g}:= \frac{1}{\mcalN_g} \sum^{N}_{x=1} g(x)\ket{x}$ from an initial product state $\ket{0}$, i.e., 
\begin{align}
\PREP_g \ket{0}= \ket{\psi_g}.
\end{align}
\end{definition}
\noindent Note that using such a subroutine could seem paradoxical, because at the end of the day we want to realize $\PREP_f$ where~$f$ is the target function. 
The assumption here is that we pick our reference function $g$ from those that are easy to prepare, hence $\PREP_g$ is considerably easier as compared to $\PREP_f$ for the target function $f$. This is exactly the idea of rejection sampling.
Next, we define the standard subroutine $U_f$, a unitary that coherently computes the value of a given function $f$ to an ancilla.

\begin{definition}[$U_f$: Computing function values to ancilla]\label{def:Uf}
Let $f:[1,N] \mapsto \mathbb{R}$ be a given function.
Then $U_f$ is defined to be a unitary with action $U_f \ket{x} \ket{0}= \ket{x} \ket{f(x)}$.
We remark that calculating the function values $f$ with $b$-bit precision will lead to an approximation of the function values up to an additive error $2^{-b}$.
Hence the number of ancilla qubits that are required for storing the value of the function with additive error $\epsilon$ is equal to $b_{\textrm{anc}}= \max\{\lceil \log_2 (\max_x f(x)  - \min_x f(x) +1) \rceil, 1\} + \lceil \log 1/\epsilon \rceil$.
\end{definition}

\noindent Next we introduce the subroutine $\Comp$ that compares the values (or a function of them) recorded in two of its input registers, and coherently flags the result of the comparison with $0$ or $1$ to another qubit ancilla.

\begin{definition}[$\Comp$: Comparator]\label{def:Comp}
Let $f:[1,N] \mapsto \mathbb{R}$ and $g: [1,N] \mapsto \mathbb{R}$ be given functions.
Let $C: (f(x), g(x), m) \mapsto \{0,1\}$ be a binary-valued function that checks a yes/no condition on $f(x)$, $g(x)$, and $m$ at the point $x$ and outputs $0$/$1$, respectively for yes/no. 
Then $\Comp_C$ that operates with clause $C$ is defined to be a unitary with action $\Comp_C |f(x)\rangle |g(x)\rangle |m\rangle |0\rangle \mapsto  |f(x)\rangle |g(x)\rangle |m\rangle |C(f(x), g(x), m)\rangle $.
The complexity of this subroutine largely depends on the complexity of implementing the clause $C$.
Note that we sometimes omit the subscript $C$ when the particular clause is clear or explicitly defined beforehand.
\end{definition}

\noindent Finally, the subroutine $\USP$ generates a uniform state preparation over a given set of objects, and $\Phase$ realizes the complex phases of the state.
They are defined as follows.

\begin{definition}[$\USP$: Uniform state preparation, and $\Phase$]\label{def:USPandPhase}
Let $M \in \mathbb{N}^+$.
$\USP_M$ is defined to be a unitary whose action is defined as $\USP_M \ket{0}= \frac{1}{\sqrt{M}} \sum^M_{m=1} \ket{m}$.
Given a phase function $\varphi: x \mapsto \varphi(x) \in [0,2\pi)$, the unitary $\Phase$ acts as $\Phase: \ket{x} \mapsto e^{i\varphi(x)}\ket{x}$. 
See Fig.~\ref{fig:PhaseState} for an implementation.
\end{definition}

\begin{figure}[t]
    \centering
    \includegraphics[scale=0.50]{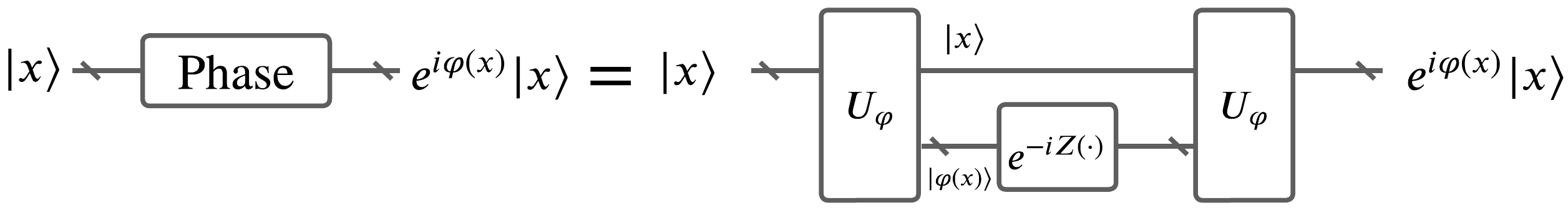}
    \caption{An implementation of $\Phase$, which uses the subroutine $U_{\varphi}$, where $\varphi: x \mapsto \varphi(x)$ are the phase angles, which are coherently computed to a reference.
    For $\varphi \in [0, 2\pi)$ expressed in binary as $\varphi_1 \varphi_2 \ldots$,
    the gate $e^{iZ (\cdot)}$ acts on $\ket{\varphi_1} \ket{\varphi_2} \ldots$ to output $e^{iZ\varphi_1}\ket{\varphi_1} e^{iZ\varphi_2}\ket{\varphi_2} \ldots$.}
    \label{fig:PhaseState}
\end{figure}

\noindent The complexity of the subroutine $\PREP_g$ depends on the function $g$, and the design principle is such that this subroutine is of complexity substantially less than the worst-case bound $\mathcal{O}(N)$.
The choice of $g$ together with the target function $f$ determines the complexity of the other subroutines.
Furthermore, note that the functions we compute and the comparator we perform do not need to be identical to the original functions $f$ and $g$.
In fact, often times, it is computationally more efficient to use other functions $\tilde{f}$ and $\tilde{g}$ in the subroutines $U_{\tilde{f}}$, $U_{\tilde{g}}$, and $\Comp$ (see the power law and Gaussian examples in Section~\ref{sec:Examples}).
While it is unlikely that any of these subroutines will take more than $\mathcal{O}(\log N)$ for efficiently computable functions, 
one should choose the clause $C$ optimally since it defines the comparator that also affects the choice of the modified functions $\tilde{f}$ and $\tilde{g}$ to be computed, and the dimension of the sampling space $M$.
The last subroutine, $\USP$, is known to be efficient, in particular for the case when $M$ is a power of $2$. 
While there is no unique option, and the possible choices proliferate quickly, we give some design guidelines of how to make these choices in Section~\ref{subsec:DesignGuidelines}.

\subsection{The algorithm and its cost in terms of subroutines}\label{subsec:Algorithm}

We start with the QRS-based state preparation algorithm, which is a general purpose solution to Problem~\ref{problem:QuantumStatePreparation}.
The quantum circuit description of Algorithm~\ref{algo:QSPA} is given in Fig.~\ref{fig:GeneralPurposeCircuit}.
We remark that a more general version of this quantum circuit also exists, where $U_f$, $U_g$, and $\Comp$ are modified to $U_{\tilde{f}}, U_{\tilde{g}}$, and $\Comp_{\tilde{C}}$ with modified functions $\tilde{f}$, $\tilde{g}$ and potentially a modified clause $\tilde{C}$.
This generalization is given in Appendix~\ref{app:GeneralClause} and in several cases it allows for further optimizations of the subroutines.

\begin{figure}
    \includegraphics[scale=0.53]{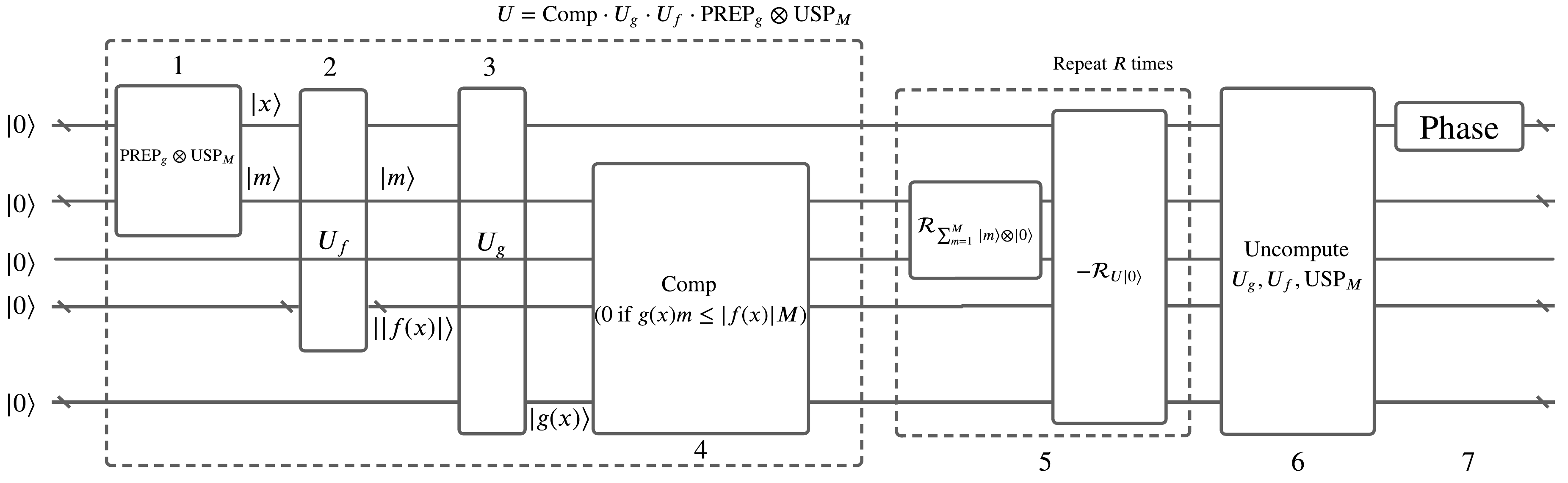}
    \caption{The quantum circuit for the general purpose state preparation based on quantum rejection sampling. $\PREP_g$, $U_{f}$, $U_{g}$, $\Comp$, $\USP$ and $\PHASE$ are subroutines that can be found in Defs.~\ref{def:PREPg}, \ref{def:Uf}, \ref{def:Comp}, \ref{def:USPandPhase}. 
    $M \in {\mathbb{N}}^+$ is chosen such that the result of the comparator lead to an approximation to the target function up to a given desired accuracy.} 
    \label{fig:GeneralPurposeCircuit}
\end{figure}

\begin{breakablealgorithm}
\caption{Quantum state preparation algorithm}\label{algo:QSPA}
\begin{algorithmic}[1]
\Require A function $f(x)= |f(x)|e^{i\varphi(x)}$ defining a target state to be prepared according to Problem~\ref{problem:QuantumStatePreparation}.
A reference function $g$ such that $g(x) \geq |f(x)|$ for all $x \in [1,N]$.
$\PREP_g$, $U_f$, $U_g$, $\Comp$, $\USP_M$, and $\Phase$ given as in Defs.~\ref{def:PREPg},~\ref{def:Uf},~\ref{def:Comp},~\ref{def:USPandPhase}.
Dimension of the sampling space $M \in \mathbb{N}^+$, chosen for a desired $\epsilon$-approximation to the target state, and the number of amplitude amplification steps $R$ as in Eq.~\eqref{eq:R}.

\Ensure An $\epsilon$-approximation to the target state $\ket{\psi_f}$ in Problem~\ref{problem:QuantumStatePreparation}.
\vspace{0.2 cm}
\State Prepare the reference state $\ket{\psi_g}$ and a uniform state of dimension $M$.

\State Compute the value $|f(x)|$ to an ancilla.

\State Compute the value $g(x)$ to an ancilla.

\State Coherently flag those $m$ for which $|f(x)| M \geq g(x) m$ with $|0\rangle$ and the rest (those for which $M |f(x)| < m g(x)$) with $|1\rangle$.

\State Amplify the $\left(\frac{1}{\sqrt{M}}\sum^{M}_{m=1} \ket{m} \otimes \ket{0} \right)$ branch.

\State Uncompute the ancilla registers holding the values $|f(x)|$ and $g(x)$, and uniform state preparation.

\State Put phases if necessary (see Figure~\ref{fig:PhaseState}).
\end{algorithmic}
\end{breakablealgorithm}

\noindent Below, our main theorem shows that Algorithm~\ref{algo:QSPA} works.
We furthermore compute the cost of the algorithm in terms of the number of calls to the subroutines, see Table~\ref{table:CallsToSubroutines}.

\begin{theorem}[Quantum sampling algorithm for quantum state preparation]\label{thm:mainState}
Let $f: [1,N] \mapsto \mathbb{C}$ be a given complex function that can be written as $f(x)= |f(x)|e^{i\varphi(x)}$.
Let $f$ and $g$ be the functions related to $f$ and $g$ such that the clause $C$ is defined as 

\begin{align}
C(|f(x)|, g(x), m)= \begin{cases}
0 ,\; m g(x) \leq M |f(x)|\\
1, \; m g(x) > M |f(x)|.
\end{cases}
\end{align}
Then Algorithm 1 prepares the state $\ket{\psi_f}:= \frac{1}{\mcalN_f} \sum^N_{x=1}|f(x)| e^{i\varphi(x)} \ket{x}$, and makes a number of calls to the core  subroutines from Defs.~\ref{def:PREPg},~\ref{def:Uf},~\ref{def:Comp},~\ref{def:USPandPhase} as given in Table~\ref{table:CallsToSubroutines}, where $R$ is the number of amplitude amplification steps
\begin{align}
\label{eq:R}
R= \left\lceil \dfrac{\pi}{4\arcsin \left(\frac{\sqrt{\sum_x |f(x)|^2}}{\sqrt{\sum_x g(x)^2}}\right)} -\frac{1}{2} \right\rceil.
\end{align}

\begin{table}[b]
\begin{center}
\begin{tabular}{ |l|c|c| } 
\hline
\# of calls to $\PREP_g$ & $1 + 2R$ \\
\hline
\# of calls to $\USP_M$ & $2 + 4R$  \\ 
\hline
\# of calls to $U_{f}$ and $U^\dagger_{f}$ & $2 + 2R$  \\
\hline
\# of calls to $U_{g}$ and $U^\dagger_{g}$ & $2 + 2R$  \\ 
\hline
\# of calls to $\Comp_C$ & $1 + 2R$  \\
\hline
\# of calls to $\PHASE$ & 1  \\
\hline
\end{tabular}
\end{center}
\caption{Number of calls to the given subroutines, where $R$ is the number of amplitude amplification steps.}
\label{table:CallsToSubroutines}
\end{table}
\end{theorem}

\begin{proof}
First we prove that the algorithm works, and then we derive the cost of the algorithm in terms of the number of calls to the subroutines. 
Note that each subroutine is numbered in Fig.~\ref{fig:GeneralPurposeCircuit} from $1$ to $7$, and in the following we label the state obtained as a result of the subroutine with the subscript denoted by the numbering of the subroutine.
\begin{enumerate}
\item We first create the reference state together with an additional sampling state which is a uniform superposition over $M$ indices, and obtain:
\begin{align}
|\phi_1\rangle = \frac{1}{\mcalN_g}  \sum^N_{x=1}  g(x)|x\rangle \frac{1}{\sqrt{M}}\sum^{M}_{m=1} \ket{m}.  
\end{align}

\item Then, we apply the unitary $U_{|f|}$ and then $U_{g}$ on the state and two extra registers, obtaining:
\begin{align}|\phi_3\rangle = \frac{1}{\mcalN_g} \sum^N_{x=1}  g(x)|x\rangle \frac{1}{\sqrt{M}}\sum^{M}_{m=1} ||f(x)|\rangle |g(x)\rangle \ket{m}.
\end{align}

\item We coherently flag those $m$ for which $|f(x)| M \geq g(x) m$ with $|0\rangle$ and the rest (those for which $M |f(x)| < m g(x)$) with $|1\rangle$, hence obtaining the state:
\begin{align}
|\phi_4\rangle = \frac{1}{\mcalN_g} \sum^N_{x=1} g(x) \left[ |x\rangle ||f(x)|\rangle |g(x)\rangle \frac{1}{\sqrt{M}}\left( \sum^{ \left\lfloor M \frac{|f(x)|}{g(x)} \right\rfloor}_{m=1} |m\rangle |0\rangle + \sum^{M}_{m= \left\lfloor M\frac{|f(x)|}{g(x)} \right\rfloor +1} |m\rangle |1\rangle \right) \right],
\end{align}
The projection of this state onto the subspace defined as the image of the projector $P:= \mathds{1} \otimes \mathds{1} \otimes \mathds{1} \otimes \ketbra{\phi}$, where $\ket{\phi}:= \left(\frac{1}{\sqrt{M}}\sum^{M}_{m=1} |m\rangle \otimes |0\rangle \right)$, is 
\begin{align}|\phi_{4'}\rangle := P|\phi_4 \rangle = \frac{1}{\mcalN_g} \sum^N_{x=1} |f(x)|  \ket{x} ||f(x)|\rangle \ket{g(x)} \frac{1}{\sqrt{M}}\sum^{M}_{m=1} |m\rangle |0\rangle,
\end{align}
where $\| |\phi_{4'}\rangle\| = \frac{\mathcal{N}_f}{\mathcal{N}_g}$. 

\item Then we amplify the $\left(\frac{1}{\sqrt{M}}\sum^{M}_{m=1} \ket{m} \otimes \ket{0} \right)$ branch, which takes $R= \left\lceil \frac{\pi}{4 \arcsin \left(\frac{\mathcal{N}_f}{\mathcal{N}_g}\right)}-\frac{1}{2} \right\rceil $ many amplitude amplification steps, and obtain
\begin{align}
|\phi_5\rangle  = \frac{1}{\mcalN_f}\sum^N_{x=1} |f(x)| \ket{x} ||f(x)|\rangle \ket{g(x)} \frac{1}{\sqrt{M}}\sum^{M}_{m=1} \ket{m} \ket{0}.
\end{align}

\item We then uncompute $U_{f}$, $U_{g}$  and the $\USP_M$ to arrive at
\begin{align} 
|\phi_6\rangle  = \frac{1}{\mcalN_f}\sum^N_{x=1} |f(x)|  |x\rangle.
\end{align}

\item We finally add the phases by applying $\PHASE$ if necessary, and obtain
\begin{align} 
|\phi_7\rangle  = \frac{1}{\mcalN_f}\sum^N_{x=1} f(x)  |x\rangle.
\end{align}
\end{enumerate}

Following the quantum circuit given in Fig.~\ref{fig:GeneralPurposeCircuit}, note that the unitary $U= \Comp \cdot U_g \cdot U_{|f|} \cdot \PREP_g \otimes \USP_M$ (together with its inverse $U^\dagger$) is called in total $1+2R$ times.
Additionally, $\USP_M$ (together with its dagger) are further called $1+2R$ times, and $U^\dagger_g$ and $U^\dagger_{|f|}$ are each called one more time.  
After gathering these, we present the number of calls to each subroutine as in Table~\ref{table:CallsToSubroutines}.
$R= \lceil \pi/(4 \arcsin(\mcalN_f/\mcalN_g)) - \frac{1}{2} \rceil$ is the number of amplitude amplification steps, where $\mcalN_f:= \sqrt{\sum_{x} |f(x)|^2}$.
\end{proof}

\subsection{Design guidelines for efficient implementation}\label{subsec:DesignGuidelines}

The asymptotic performance of the algorithm depends on the asymptotic efficiency of the subroutines. 
The cost of the algorithm can often be further improved by making certain design choices that we summarize in this section. 

Let us start with $U_g$. 
We wish to choose the reference function $g$ such that $g(x) \geq f(x)$ for all $x$ in the domain of interest $D$ such that the reference state $\ket{\psi_g}$ is efficient to prepare.
For example, a constant function $g(x)= c$ where $c=\max_{x \in D} f(x)$ is a choice.
The reference state implied by this function is particularly easy to prepare, e.g., it is the uniform state over the values of $x$ that lies in the domain $D$.
Depending on how the values of $f$ fluctuate in the domain, we may want to modify the function to a piecewise constant function, such as $g(x)= c_\mu$ for $x \in D_\mu$, where $c_\mu= \max_{x \in D_\mu} f(x)$, and the subdomains $D_\mu$ are nonoverlapping and $\cup_\mu D_\mu = D$.
In this case, we need to generate the reference state $\ket{\psi_g}$ which is a certain linear combination of some states $\ket{\psi_\mu}$ which are uniform superposition over elements on corresponding subdomains.
As long as $K$, the number of subdomains, scales like $o(N)$, the algorithm has the potential of having an asymptotic advantage compared to the methods that assumes arbitrary state preparation, such as Refs.~\cite{grover2002creating, mottonen2004transformation, low2018trading}.
Ideally, we prefer $K= \mathcal{O}(\log N)$.
For functions that are rapidly-decaying in some subdomains, we may employ an exponential rather than a constant function, so that the number of amplitude amplification steps is reduced. 
Quantum circuits that prepare these kind of states are given in Lemma~\ref{lem:AQuantumCircuitForPrepG} of Section~\ref{subsubsec:ChoiceOfG}, with examples of piece-wise constant functions (see Example~\ref{ex:PWConstantZiggurat}), and piece-wise functions that are a mixture of constant functions and exponentials (see Example~\ref{ex:ConstantAndExp}). Also, these quantum circuits are less costly when the domains and functions values are chosen such that the relevant numbers are integer powers of~$2$.
Fig.~\ref{fig:ReferenceComparison} presents illustrative reference functions $g$ for a specific inverse function target (Fig.~\ref{fig:ReferenceComparisonInverse}) and a specific Gaussian function target (Fig.~\ref{fig:ReferenceComparisonGaussian}).
\begin{figure}[t]
\centering
\begin{subfigure}[t]{.4\textwidth}
  \includegraphics[width=\linewidth]{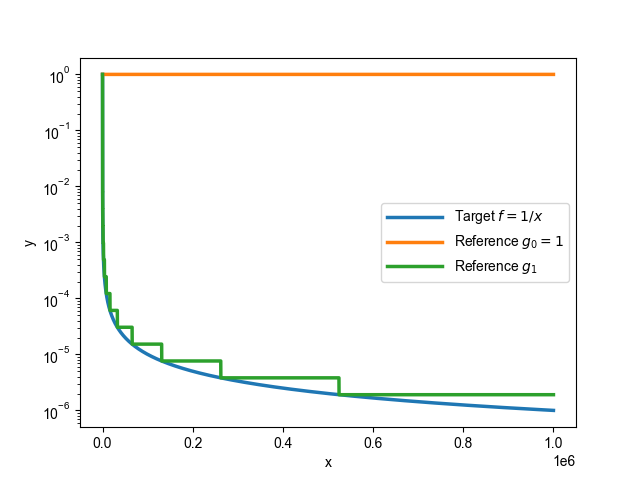}
  \caption{Log-linear-scale plots of the target function $f(x)=1/x$, and two potential reference functions $g_0(x)= 1$, and $g_1(x)= 2^{-\lfloor \log x \rfloor}$ for $x \in [1, 10^6]$.}
  \label{fig:ReferenceComparisonInverse}
\end{subfigure}\hspace{4em}%
\begin{subfigure}[t]{.4\textwidth}
  \includegraphics[width=\linewidth]{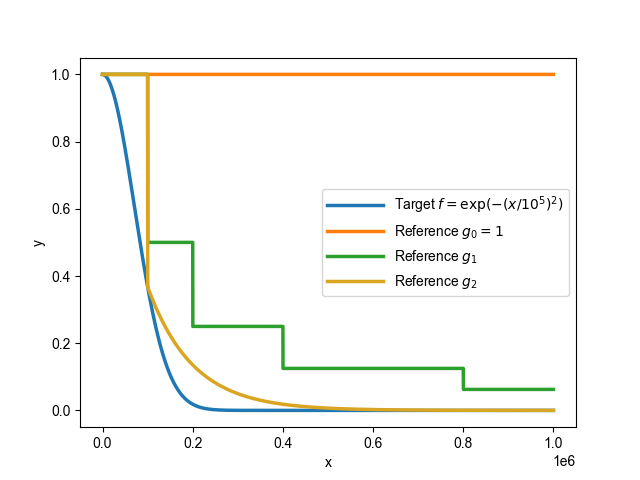}
  \caption{Linear-linear-scale plots of the target function $f(x)=\exp(-(x/10^5)^2)$, and three potential reference functions $g(x)=1$ for $x \in [1, 10^5]$.
  In the tail region ($x \in [10^5, 10^6]$), $g_0(x)= 1$, $g_1(x)= 2^{-\lfloor \log x - \log 10^5 +1  \rfloor}$, and $g_2(x)= \exp(-x/10^5)$.}
  \label{fig:ReferenceComparisonGaussian}
\end{subfigure}
\caption{Illustration of various choices of reference functions for the target functions (a) $f(x)= 1/x$, and (b) $f(x)= \exp(-(x/10^5)^2)$.}
\label{fig:ReferenceComparison}
\end{figure}

Next, we consider how to choose the dimension of the sampling space $M$, how precisely the functions $f$ and $g$ need to be computed to an ancilla, and how to utilize $f$ and $g$ in the comparator in the most efficient way.
The value of $M$ is one of the factors that affects the efficiency and accuracy of the algorithm, and it has to be chosen high enough so that the comparator that checks for the largest $m \in \mathbb{N}^+$ such that $m g(x) \leq M f(x)$ does not make too large of a rounding error.
The final desired accuracy $\epsilon$ in preparing the target state $\ket{\psi_f}$ implies a lower bound on $M$, as well as an upper bound on the precision of computing the function values $f$ and $g$ to the ancilla.
These results are given in Section~\ref{subsubsec:ChoiceOfM} in Lemma~\ref{lem:ChoosingM}.
Note that for the utmost efficiency in terms of gates, we prefer to choose $M$ to be an integer power of $2$.

\subsubsection{The choice of the reference function \texorpdfstring{$g$}{Lg} and the corresponding quantum circuit}\label{subsubsec:ChoiceOfG}

One of the main points that ensures the efficiency of our QRS-based method is an efficient preparation of the reference state $\ket{\psi_g}$, i.e., the quantum circuit $\PREP_g$ defined as in Definition~\ref{def:PREPg} is efficient.
In particular, we make the choice that $g$ leads to a quantum state $|\psi_g\rangle$ that can be constructed as a product state or a linear combination of a handful of product states over separate regions in the domain.
In particular, we desire to pick $g: D \rightarrow \mathbb{R}$ as follows

\begin{align}
g= \sum^K_{\mu=1} {g_\mu} \; \textrm{where}\; g_\mu: D_\mu \rightarrow \mathbb{R},\; D_\mu \cap D_{\mu'}= \emptyset \; \textrm{for all} \; \mu \neq \mu',\; \textrm{and} \; D= \bigcup_{\mu} D_\mu,
\end{align}
where each $g_\mu$ is a factorizable (among the bit representation of the domain) function that leads to a product state, such as a state with uniform amplitudes for all values in its domain.
The crucial point is that we need a quantum circuit that patches the pieces $\{|\psi_{g_\mu}\}_\mu$ to the full state $|\psi_g\rangle$.
This is achieved by linear combination of unitaries and an appropriate uncompute subroutine.
The full statement and the details are given as follows.

\begin{lemma}[A quantum circuit for $\PREP_g$]\label{lem:AQuantumCircuitForPrepG}
Let $g= \sum^K_{\mu=1} g_\mu : D \rightarrow \mathbb{R}$ be a function that is given as a sum of $K$ functions $g_\mu: D_\mu \rightarrow \mathbb{R}$ where $D= \bigcup_\mu D_\mu$, $D_\mu \cap D_{\mu'}= \emptyset$ for all $\mu \neq \mu'$. 
Let $\PREP_{g_\mu}$ be a given quantum circuit such that $\PREP_{g_{\mu}}\ket{0} = \ket{\psi_{g_\mu}}= \frac{1}{\mathcal{N}_{g_\mu}}  \sum_{x \in D_\mu} g_\mu(x) \ket{x}$.
Then there exists a quantum circuit that uses an additional $\log K$ qubits, an additional quantum circuit of $\PREP_{\v{c}}$ with $\v{c}= (\mcalN_{g_1}/\mcalN_{g}, \mcalN_{g_2}/\mcalN_{g}, \ldots, \mcalN_{g_{K}}/\mcalN_g)$ and calls to controlled $\PREP_{g_\mu}$ to create a state 
\begin{align}
    \ket{\psi'_g} = \frac{1}{\mcalN_g}\sum_{x \in D} g_{\mu}(x) \ket{\mu(x)} \ket{x}.
\end{align} 
Optionally, an uncomputing circuit $\UNCOMP$ allows one to create a state $\ket{\psi_g}= \frac{1}{\mcalN_g}\sum_{x \in D} g_{\mu}(x) \ket{x}$.
Specifically,
\begin{align}\label{eq:LemmaCircuitForG}
\ket{\psi'_g}= \left(\sum^K_{\mu=1} \ketbra{\mu}{\mu}_A \otimes (\PREP_{g_\mu})_S \right) \cdot (\PREP_{\v{c}} \otimes \mathds{1})\ket{0}_{A} \otimes \ket{0}_{S}= \frac{1}{\mcalN_g}\sum_{x \in D} g_{\mu}(x) \ket{\mu: x \in D_{\mu}} \ket{x} 
\end{align}
and 
\begin{align}\label{eq:LemmaCircuitForG'}
\ket{\psi_{g}}= \UNCOMP \ket{\psi'_g}= \frac{1}{\mcalN_g}\sum_{x \in D} g_{\mu}(x) \ket{x} 
\end{align}
where the unitary $\UNCOMP$ is given as 
\small
\begin{align}
\UNCOMP= ( \mathds{1}_A \otimes \sum_x \ketbra{0}{\mu(x)}_{A'} \otimes \ketbra{x}{x}_S) \left(\sum^K_{\mu', \mu=1} \ketbra{\mu'-\mu}{\mu'}_A \otimes \ketbra{\mu}{\mu}_{A'}\otimes \mathds{1}_S \right) ( \mathds{1}_A \otimes \sum_x \ketbra{\mu(x)}{0}_{A'} \otimes \ketbra{x}{x}_S ).
\end{align}
\normalsize
The circuit is given in Figure~\ref{fig:CircuitForG}.
\end{lemma}

\begin{proof}
We follow the state from the circuit given in Figure~\ref{fig:CircuitForG}, which is broken into three stages, as in Eq.~\eqref{eq:LemmaCircuitForG}.
\begin{figure}
    \centering
    \includegraphics[scale=0.65]{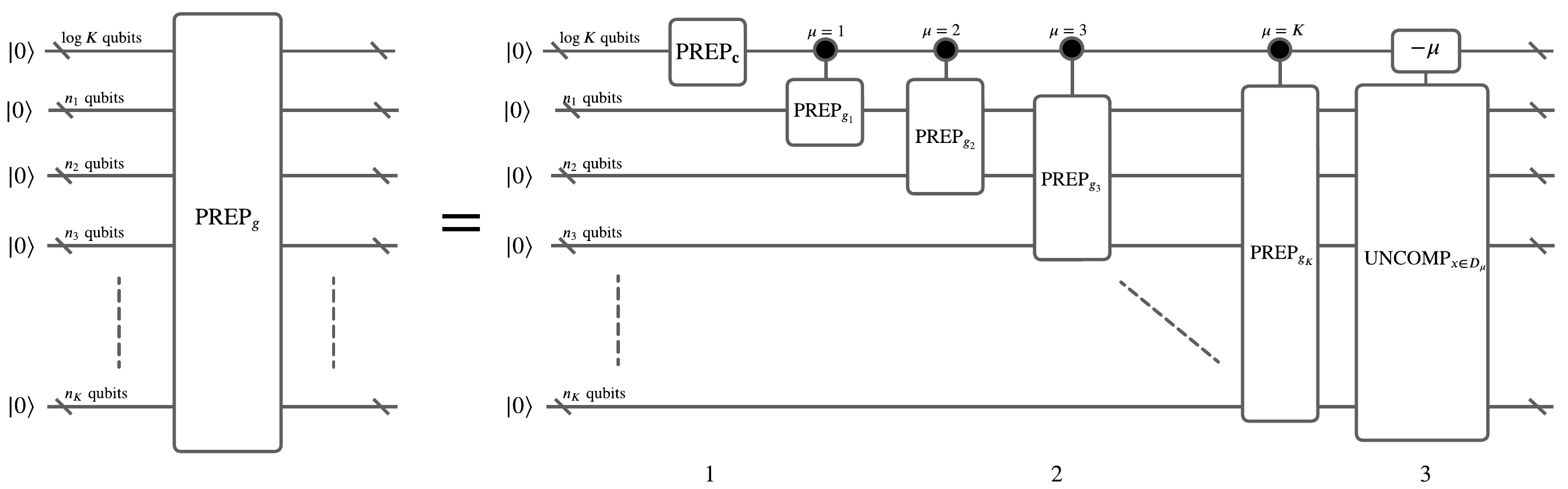}
    \caption{A general purpose quantum circuit that prepares the state $\ket{\psi_g}$ where $g$ is given as a sum of functions $g_\mu$ defined over non-overlapping domains $D_\mu$, respectively.
    This circuit assumes that the description of region $D_\mu$ fits into $n_\mu$ registers, in particular $|D_\mu| \leq 2^{n_\mu}$.
    }
    \label{fig:CircuitForG}
\end{figure}
At the end of the first stage, we have 
\begin{align}
\ket{\phi_1}&= (\PREP_{\v{c}} \otimes \mathds{1}) \ket{0}_A \otimes \ket{0}_S\\
&= \sum^{K}_{\mu=1} \frac{\mcalN_{g_\mu}}{\mcalN_{g}} |\mu\rangle_A \ket{0}_S
\end{align}
where $\mcalN_g$ and $\mcalN_{g_\mu}$ are the normalization factors.
At the end of the second stage, we have 
\begin{align}
\ket{\phi_2}&= \left(\sum^K_{\mu=1} \ketbra{\mu}{\mu}_A \otimes (\PREP_{g_\mu})_S \right) \ket{\phi_1}_{AS}\\
&= \sum_{\mu} \frac{\mcalN_{g_\mu}}{\mcalN_{g}} |\mu\rangle_A \otimes \frac{1}{\mcalN_{g_\mu}} \sum_{x \in D_\mu} g_{\mu}(x) \ket{x}_S.
\end{align}
Note that $\mcalN_g= \sqrt{\sum_{\mu} \mcalN^2_{g_\mu}}$, and $\mcalN_{g_\mu}= \sqrt{\sum_x |g_{\mu}|^2}$, by definition.
This completes the first part of the proof, as $\ket{\psi'_g}= \ket{\phi_{2}}$.

For the second part, notice that the system is still entangled with the ancilla, i.e., if $x \in D_\mu$ the ancilla is in state $\mu$ labeling the domain that $x$ belongs to.
The final stage to get $\ket{\psi_{g}}$ uncomputes the label $\mu(x)$, by first coherently computing the label of the domain $\mu$ that $x$ belongs to into an additional quantum register labeled by $A'$.
Then we subtract that value from the first register, and finally uncompute the domain label $\mu$ that was computed at the beginning of the $\UNCOMP$ circuit.
Hence, we complete the second part of the proof by showing
\begin{align}
\ket{\phi_3}&= \UNCOMP \ket{\phi_2}\\
&= \UNCOMP \left(\frac{1}{\mcalN_g} \sum^K_{\mu=1} \sum_{x \in D_\mu} g_\mu(x) \ket{\mu}_A \ket{x}_S\right)\\
&= \ket{0}_A \ket{\psi_{g}}_S.
\end{align}
$\UNCOMP$ first computes $\mu(x)$ to an ancilla labeled by $A'$, and then subtracts it from the value in register $A$, and finally uncomputes the value $\mu(x)$ in register $A'$ to $0$.
\end{proof}
Note that, in the quantum state preparation algorithms that use $\PREP_g$ as a subroutine given a reference function $g$, we usually make use of the state $\ket{\psi'_g}$ in Eq.~\eqref{eq:LemmaCircuitForG} rather than $\ket{\psi_{g}}$ given in Eq.~\eqref{eq:LemmaCircuitForG'}.
The comparator in the algorithm then uses the index $\mu$ as well, because the inequality that is used in the comparator usually depends on $\mu$.
After that, the index $\mu$ can be uncomputed or, if appropriate, be kept.

Now, we give an explicit construction for two examples of reference states that we make use of frequently.
First, in Example~\ref{ex:PWConstantZiggurat}, we pick the reference function $g$ to be piecewise constant, which we also term a ``ziggurat'' due to its shape.
This function serves as a good reference function for target functions $f$ which do not have long and rapidly decaying tails/regions.
Second, in Example~\ref{ex:ConstantAndExp}, we pick a reference function $g$ to be a sum of a constant function and an exponentially decaying function.
This, and similar functions, serve as good reference states for target functions $f$ that have rapidly decaying tails or regions, such as a Gaussian.
These examples are both given for functions whose domains are $1$-dimensional, i.e., $x \in [1,N]$ hence lies on a line, and the construction generalizes to similar functions with higher dimensional domains, with some added complexity of loading the necessary domain information to the quantum computer.

\begin{example}[Piece-wise constant, aka Ziggurat]\label{ex:PWConstantZiggurat}
Let $g= g_\mu$ for $x\in [N_{\mu-1}, N_\mu - 1]$ where $N_0= 1$ and $1\leq N_1 \leq N_2 \leq \ldots \leq N_K= N-1$, and $\mu \in [1,K]$. 
Namely, we wish to prepare the state $\ket{\psi_g}\propto \sum^K_{\mu=1}\sum^{N_\mu -1}_{x=N_{\mu-1}} g_\mu \ket{x}$. 
One initially prepares a smaller reference state 
\begin{equation}
\label{eq:psicoarse}
  \ket{\psi^{\textrm{coarse}}_g} = \PREP_{\v{c}} \ket{0}  \propto \sum_{\mu=1}^K \sqrt{N_\mu - N_{\mu-1}} g_\mu \ket{\mu}.
\end{equation}
Then $\ket{\psi'_{g}}$ and $\ket{\psi_{g}}$ can be obtained from $\ket{\psi^{\textrm{coarse}}_g}$ by performing, conditioned on $\mu$, a uniform state preparation over $\log_2 (N_\mu - N_{\mu-1})$ qubit registers and then, if necessary, uncomputing the register $\ket{\mu}$ (as per Fig.~\ref{fig:CircuitForG}). 
\begin{figure}
    \centering
    \includegraphics[scale=0.60]{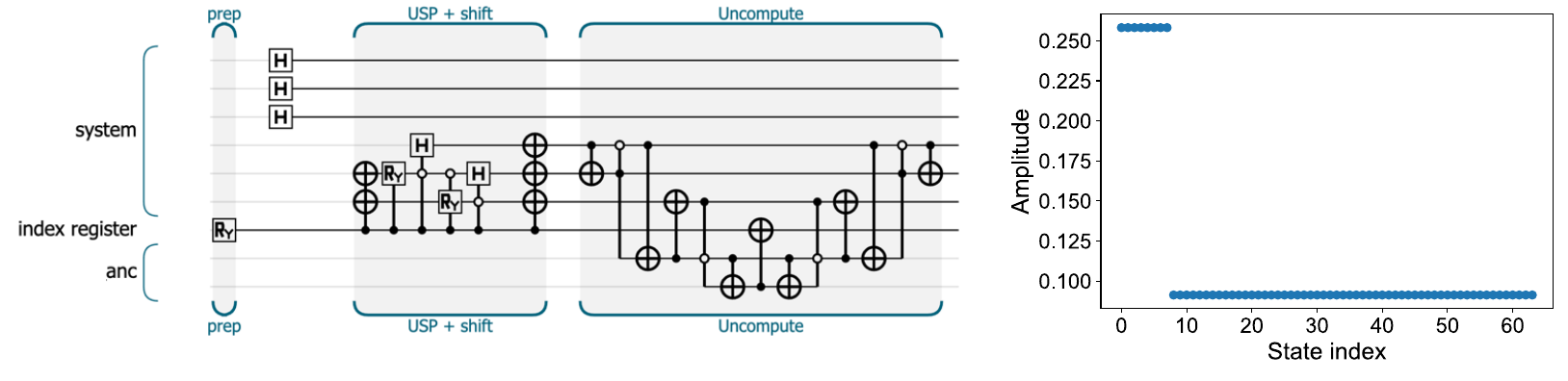}
    \caption{
    An example circuit for implementing a piece-wise constant function on six qubits, where the amplitude on the first seven basis states is $\sqrt{\frac{1}{15}}$ and the amplitude on the rest is $\sqrt{\frac{1}{120}}$ (as shown in the plot on the right). The prepare operation on the index register sets the proportion of the amplitudes to be distributed among the two regions of the piece-wise constant function. Then, conditioned on the index register, for this example, a uniform state preparation (USP) circuit is applied followed by a shifting operation to apply the USP on a desired portion of the wavefunction. The ``Uncompute'' part of the circuit disentangles the index register from the system register and unprepares it back to $\ket{0}$.}
    \label{fig:piecewise_constant_sim}
\end{figure}
We show an example circuit for preparing piece-wise constant function in Fig.~\ref{fig:piecewise_constant_sim}.
\end{example}

\begin{example}[Exponentially decaying tails]\label{ex:ConstantAndExp}
Let $g= 1$ for $x \in [0, N_1-1]$ and $g= \exp(-\beta x)$ for $x \in [N_1 , N-1]$.
Namely, we wish to prepare the state $\ket{\psi_g}\propto \sum^{N_1-1}_{x=0}\ket{x} + \sum^{N-1}_{x= N_1} \exp(-\beta x)$.
\begin{figure}
    \centering
    \includegraphics[scale=0.60]{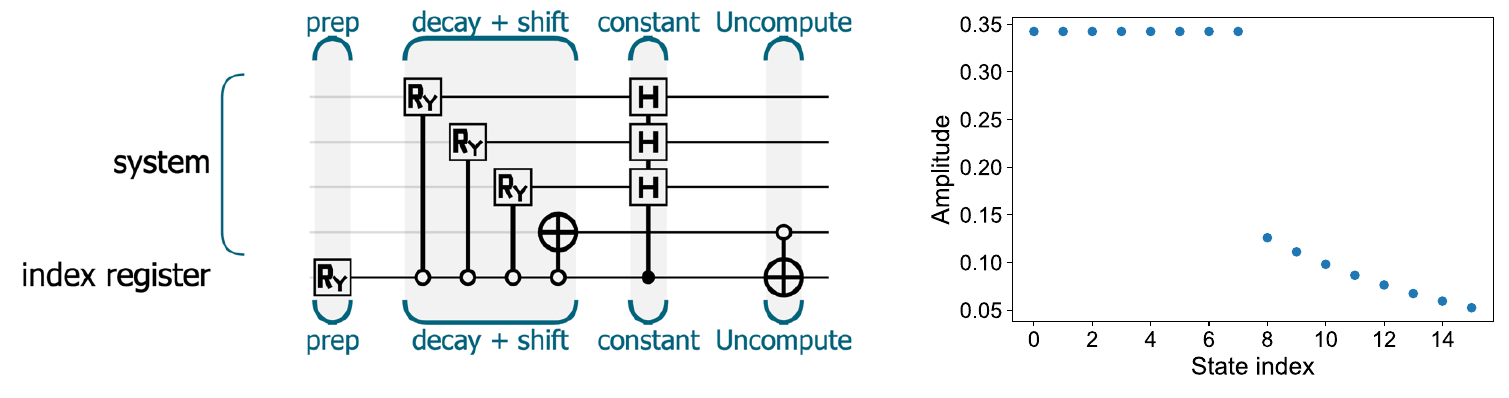}
    \caption{An example circuit for implementing a function that is a sum of a constant function and an exponentially decaying function. In this four-qubit example, the constant function is applied to the first eight basis states, and an exponential decay function is applied to the latter eight. The circuit is labeled with operations used to prepare the two sectors of the function (``prep''), prepare the decay function followed by a shift (``decay + shift''), prepare the constant function (``constant''), and uncompute the ancillary (index) register (``Uncompute'').
    }
    \label{fig:exp_decaying_tails_sim}
\end{figure}
We show an example circuit for preparing a function with an exponentially decaying tail in Fig.~\ref{fig:exp_decaying_tails_sim}.
\end{example}

\noindent Special cases of Examples~\ref{ex:PWConstantZiggurat}-\ref{ex:ConstantAndExp} are worked out in detail in Section~\ref{subsec:ExamplesZiggurat} and Section~\ref{subsec:ExamplesExponential}, respectively. 
Next, we give a general reference state construction obtained by coarsening the target function.
This reference function construction can work well for many target functions, including those which are smooth or beyond.

\begin{example}[Reference states for smooth target functions]
\label{ex:meshing}
Note that in several cases -- including some examples in this manuscript -- the target function $f: [1,N] \rightarrow \mathbb{R}^+$ is the $N$-point discretization of a smooth function $f_{\textrm{smooth}}: \mathbb{R}^l \rightarrow \mathbb{R}^+$, $l=1,2,3,\dots$. For example, for $l=1$ and $f_{\textrm{smooth}}$ supported in the interval $[0,1]$, we may have $f(x) := f_{\textrm{smooth}}(x/N)$. 
A natural way of defining a reference is then:
\begin{itemize}
    \item Mesh the domain of $f_{\textrm{smooth}}$ via a (regular) grid with $K$ cells $\{C_\mu\}_{\mu=1}^K$;
    \item Let $g_\mu \geq \sup_{x \in C_\mu} f_{\textrm{smooth}}(x)$.
    \item Let $N_\mu$ be the number of points in the $N$-point discretization that are inside $C_\mu$. Typically $N_\mu$ is just a constant, $N_\mu = N/K$, but taking a varying mesh can be advantageous.
    \item This defines a small reference state $\ket{\psi^{\textrm{coarse}}_g}$ as in Eq.~\eqref{eq:psicoarse}.
    \item Given $\ket{\psi^{\textrm{coarse}}_g}$ we prepare the reference state $\ket{\psi_g}$. 
    If $l=1$, $N_\mu= N/K$ and $N/K$ is a power of $2$,  the uniform state preparations required in Example~\ref{ex:PWConstantZiggurat} are simply uncontrolled products of Hadamards on the ancilla registers.  
\end{itemize}

For smooth $f_{\textrm{smooth}}$ on a compact domain, the convergence of the upper Riemann sums ensures that taking $K=\mathcal{O}(1)$ cells is enough to ensure $R= \mathcal{O}(1)$ rounds of amplitude amplification steps are sufficient in the quantum sampling algorithm. 
This gives a wide range of scenarios where quantum state preparation is asymptotically efficient if the target function can be efficiently computed to a register. 
While this approach is a useful baseline, in practice we want to try to use as much information about the function $f$ as possible beyond its smoothness (e.g., symmetry), to further reduce the cost.
\end{example}

A two-dimensional instance of Example~\ref{ex:meshing} is shown in Appendix~\ref{app:meshExample} where we analyze a complicated function without apparent structure that is exploitable by existing state preparation methods. Using a meshed ziggurat reference state, the target state can be efficiently prepared with high probability, i.e. requiring only one round of amplitude amplification. 
This example illustrates the versatility and power of the rejection sampling protocol, which can handle seemingly difficult functions (including those with highly irregular oscillations, discontinuities, etc.) with relative ease.
In these cases, the majority of the cost would likely arise from the arithmetic computation of the function itself.
Note again that the same constuction given in Example~\ref{ex:meshing} can also work well for general functions beyond only smooth functions. 
In fact, this is one of the class of target functions that the QRS-based method has better scaling compared to other quantum state preparation approaches based on polynomial approximations to the target functions, such as the one based on quantum signal processing given in Ref.~\cite{mcardle2022quantum}.
A more comprehensive comparison of methods and their expected performances are given in Sec.~\ref{subsec:GeneralComparison}.

\subsubsection{The choice of \texorpdfstring{$M$}{Lg}
}\label{subsubsec:ChoiceOfM}
The choice of $M$ is related to the desired accuracy for the target state.
The higher $M$ is, the more accurate the state preparation is. 
However, this also increases the complexity of the algorithm.
Hence, we want to pick a minimal value for $M$ such that the resulting state is guaranteed to be approximated to an error $\epsilon$, as defined in Def.~\ref{problem:QuantumStatePreparation}.
The following lemma addresses this by working out the accuracy of the comparator and how it affects the accuracy of the state preparation.

\begin{lemma}[$M$: Dimension of the sampling space]\label{lem:ChoosingM}
Let $\ket{\psi_f}= \frac{1}{\mcalN_f} \sum^N_{x=1} f(x) \ket{x}$ be the target state, where $f(x) \in \mathbb{R}^+$, and let $\ket{\psi_g}$ be the reference state defined similarly to $\ket{\psi_f}$, where $g(x) \in \mathbb{R}^+$ and $g(x) \geq f(x)$ for all $x \in [1,N]$.
Then, for all values of $M$ such that
\begin{align}
M \geq \max\left\{ \frac{2\mcalN_g}{\epsilon \mcalN_f}, \max_x \frac{2 g(x)}{f(x)}\right\}
\end{align}

\noindent Algorithm~\ref{algo:QSPA} produces an $\epsilon$-approximation $|\psi_{\bar{f}}\rangle$ to the target state $|\psi_{f}\rangle$, i.e., $\||\psi_f\rangle - |\psi_{\bar{f}}\rangle\rangle\| \leq \epsilon$.
This assumes an at most $\tilde{\epsilon}$ additive error for computing the ratio $f(x)/g(x)$, computed as $\widetilde{f/g}(x)$ such that
\begin{align}
\frac{f(x)}{g(x)} - \tilde{\epsilon}\leq \widetilde{f/g}(x) \leq \frac{f(x)}{g(x)} + \tilde{\epsilon}
\end{align}
where
\begin{align}
\tilde{\epsilon}(x) \leq \min\left\{ \frac{\epsilon \mcalN_f}{2 \mcalN_g}, \frac{f(x)}{2 g(x)} \right\}.
\end{align}
\end{lemma}

\noindent Note that, when computations of functions are performed for either the ratio directly or the related functions, one opts for a final bit precision that is implied by the values $\tilde{\epsilon}$.
For example, if during the algorithm $f(x)/g(x)$ is directly being computed with precision independent of $x$, then the bit precision of the computation is chosen as $\tilde{\epsilon} \leq \min\left\{ \frac{\epsilon \mcalN_f}{2 \mcalN_g}, \min_x \frac{f(x)}{2 g(x)} \right\}$.

\begin{proof}
The proof consists of two parts.
First, we show that an at most additive $\epsilon_0 (x)= \epsilon g(x)/\mcalN_g$-approximation to the normalized amplitude values results in an at most $\epsilon$-approximation to the final state.
Assume that the normalized amplitude function $f(x)/\mcalN_f$ has been obtained with $\epsilon_0$ additive precision, i.e., the state $|\psi_{\bar{f}}\rangle$ is obtained in the quantum computer with $\bar{f}$ such that
\begin{align}\label{eq:Lemma211proof1}
\frac{f(x)}{\mcalN_f} - \epsilon_0(x) \leq \frac{\bar{f}(x)}{\mcalN_{\bar{f}}} \leq \frac{f(x)}{\mcalN_f} + \epsilon_0(x).
\end{align}
\noindent This implies that
\begin{align}
\left\| |\psi_f\rangle - |\psi_{\bar{f}}\rangle \right\| & = \left\| \frac{1}{\mcalN_f} \sum^N_{x=1} f(x) \ket{x} -  \frac{1}{\mcalN_{\bar{f}}} \sum^N_{x=1} \bar{f}(x) \ket{x} \right\| \\
& = \sqrt{\sum^N_{x=1} |\epsilon_0(x)|^2}.
\end{align}
\noindent Hence, choosing $\epsilon_0 (x) = \epsilon g(x)/\mcalN_g$ is sufficient to ensure that the target state is approximated with $\epsilon$ error.
Note that the algorithm works so that $f(x)/\mcalN_f - \epsilon_0(x) \geq 0$ is also satisfied, due to the first inequality of Eq.~\eqref{eq:Lemma211proof1}. 
This implies that an $\epsilon_0$ such that
\begin{align}\label{eq:Lemma211proof2}
\epsilon_0(x) = \min\left\{\frac{\epsilon g(x)}{\mcalN_g}, \frac{f(x)}{\mcalN_f} \right\}
\end{align}
is sufficient to obtain an $\epsilon$-approximation to the final state.\\

We then need to make sure our computations are realized within our algorithm with a desired sufficient additive error as well.
This brings us to the second part of the proof.
We want $M$ to be high enough so that the final value of the normalized amplitude is approximate with $\epsilon_0$ additive error.
To be more precise, assuming $\ket{\psi_g}$ is realized exactly, we wish to compute the ratio $f(x)/g(x)$, as $\widetilde{f/g}(x)$, with $\tilde{\epsilon}$ additive error, namely,
\begin{align}
f(x)/g(x) - \tilde{\epsilon}(x) \leq \widetilde{f/g}(x) \leq f(x)/g(x) + \tilde{\epsilon}(x),
\end{align}
for all $x$, where $\tilde{\epsilon}$ is chosen such that
\begin{align}
\frac{f(x)}{\mcalN_f} - \epsilon_0(x) \leq \frac{g(x)}{M \mcalN_f} \left\lfloor M \widetilde{f/g}(x) \right\rfloor \leq \frac{f(x)}{\mcalN_f} + \epsilon_0(x).
\end{align}

\noindent This implies the following two conditions
\begin{align}
\frac{g(x)}{M \mcalN_f} \left( M \frac{f(x)}{g(x)} + M \tilde{\epsilon}(x) \right) \leq \frac{f(x)}{\mcalN_f} + \epsilon_0(x), \; \textrm{and} \; \frac{f(x)}{\mcalN_f} - \epsilon_0(x) \leq \frac{g(x)}{M \mcalN_f} \left( M \frac{f(x)}{g(x)} - M \tilde{\epsilon}(x) -1 \right).
\end{align}
These then imply
\begin{align}
\tilde{\epsilon}(x) \leq \frac{\epsilon_0(x) \mcalN_f}{g(x)}, \; \textrm{and} \; M \geq \frac{1}{\left( \frac{\mcalN_f}{g(x)}\epsilon_0(x) - \tilde{\epsilon}(x) \right)},
\end{align}
respectively.
Employing Eq.~\eqref{eq:Lemma211proof2} for values of $\epsilon_0$, and choosing $\tilde{\epsilon}(x) \leq \frac{\epsilon_0(x) \mcalN_f}{2 g(x)}$, we obtain the result.
\end{proof}

\section{Examples and resource counts}\label{sec:Examples}

In this section, we examine example cases, including target functions such as piecewise-constant over extended domains (aka ziggurat), exponential, power-law, Gaussian, and hyperbolic tangent.
The first two cases, i.e., ziggurat and exponential, are used as reference states, and do not need to use a quantum sampling approach.
The rest of the examples are worked out in detail by designing reference functions, modifying the comparator and other circuit components and implementations for optimal Toffoli count.
Each subsection includes a resource count for a particular instance of the given class of functions.
Note that the choice of domain $D$, while generally chosen to be $[1,N]$, depends highly on the given problem.
One way to approach this is to map the elements $x$ in the domain $D$ to the domain $[1,N]$ bijectively.
This works particularly well, if this bijective map is an efficiently implementable reversible operation, e.g., an invertible linear map $A$ (for example, $x \mapsto x/2$).
However, in general, a potential problem with this approach is that one may change, complicate, or even lose the structure in the function $f$.
Often, however, the target function is efficiently computable with an efficient reference state without changing the domain.
In this case, one needs to take extra care about creating states in the span of the the states $\ket{x}$ for $x \in D$.
Below and in other examples, while it is common to have a regular domain with equally distanced points, such as $[1,N]$, the discussion and methods work for general domains (although an arbitrarily complex domain would be inefficient to encode, even for a constant function).
Furthermore, note that we have assumed a Fourier state~\cite{gidney2018halving, low2018trading} is ready to use, and omit the one time cost of preparing this state, together with minor $T$-state costs that assist in reducing the number of Toffoli gates for certain operations (such as in Fig.17 of Ref.~\cite{lee2021even} for controlled-Hadamards, where $T$-states are used catalytically).

\subsection{Ziggurat with extended domains: \texorpdfstring{$\ket{\psi_g} \propto \sum^n_{\mu=1} \frac{1}{\sqrt{2^{\mu-1}}} \sum^{2^{\mu}-1}_{x \in 2^{\mu-1}} \ket{x}$}{Lg}}\label{subsec:ExamplesZiggurat}

This example is a special case of Example~\ref{ex:PWConstantZiggurat}, where the amplitude of the computational basis state $\ket{x}$ is equal to $1/\sqrt{2^{\mu-1}}$ when $x \in [2^{\mu-1}, 2^{\mu}-1]$, and $N= 2^n$. 
This reference state can be constructed with a special case of the quantum circuit given in Fig.~\ref{fig:CircuitForG}.
The quantum circuit consists of only $n-1$ controlled Hadamard gates and appropriate shifts of the register values.
Note that the circuit is a bit more complicated (similar to the general version given in Figure~\ref{fig:CircuitForG}) if the boundaries of each region are not exactly an integer power of $2$.
More precisely, the circuit would have  $\mathcal{O}(n)$ Toffoli gates.

\subsection{Exponential: \texorpdfstring{$\ket{\psi_g}\propto \sum_x \exp(-\beta x) \ket{x}$}{Lg}}\label{subsec:ExamplesExponential}

\noindent The exponential function, $\exp(-\beta x)$, is a factorizable function. 
Hence the most direct way to generate these type of states is to use a product of single-qubit arbitrary $SU(2)$ rotations determined by the parameter $\beta \in \mathbb{C}$.
More precisely, given an $(n)$-bit representation of $x$ as $x= (x_1, \ldots, x_n)$, i.e., 

\begin{align}
x= \sum^n_{i=1}x_i 2^i
\end{align}
we express the exponential function as follows
\begin{align}
\exp(-\beta x)= \prod^n_{i=1} e^{-\beta 2^{i} x_i}.
\end{align}

\noindent Using this expression, the target state can be expressed as a product state as follows

\begin{align}
\frac{1}{\mcalN_g}\sum_x \exp(-\beta x) \ket{x} &= \frac{1}{\mcalN_g} \bigotimes^n_{i=1} \sum^1_{x_i=0} e^{-\beta 2^{i} x_i}\ket{x_i}\\
\label{eq:ExponentialProductForm} &= \bigotimes^n_{i=1} \left[ \sum^1_{x_i=0} \frac{1}{\sqrt{1 + e^{-\beta 2^{i+1}}}}e^{-\beta 2^{i} x_i}\ket{x_i} \right].
\end{align}

\noindent Note that the state given in Eq.~\eqref{eq:ExponentialProductForm} is a product state, each of which can be produced with a quantum circuit that implements single-qubit rotations of a given angle.
Particularly, the $i$th state to be produced is given by 

\begin{align}
\frac{1}{\sqrt{1 + e^{-\beta 2^{i+1}}}} \left(\ket{0} + e^{-\beta 2^i}\ket{1}\right).
\end{align}

\noindent The complexity of the circuit that produces this state in terms of the Toffoli count can be found assuming that a $b$-bit Fourier state is provided.
An arbitrary $b$-bit angle $Y$-rotation can then be realized with $b$ Toffoli gates.
The bit precision for angle is chosen such that an $\epsilon$ approximation to the final state is achieved.
It is sufficient to choose $b= \lceil \log (2\pi n/\epsilon) \rceil$.
The total Toffoli cost of generating the target state is $nb$.

\noindent For a numerical resource count, we make the following choices: the function is given by $f(x)= \exp(-|x|)$, i.e., for $\beta= 1$, in the domain $x \in \{-0.5, -0.5 +\delta , \ldots, -\delta/2, \delta/2 , \ldots, 0.5\}$ for a given $0 < \delta \ll 1$.
Note that this was done without loss of generality, since any rescaling of $x$ can be pushed to the value of the parameter $\beta$.
This corresponds to a quantum state of dimension of $N=1/\delta$. 
In this setting, $n=\lceil \log N \rceil$ and $b= \lceil \log(2 \pi n/\epsilon) \rceil$.
See Fig.~\ref{fig:ExpExample} for the Toffoli count estimates with varying $N \in [10^2, 10^9]$ and $1/\epsilon \in [10^2, 10^9]$.

\begin{figure}
    \centering
    \includegraphics[scale=0.44]{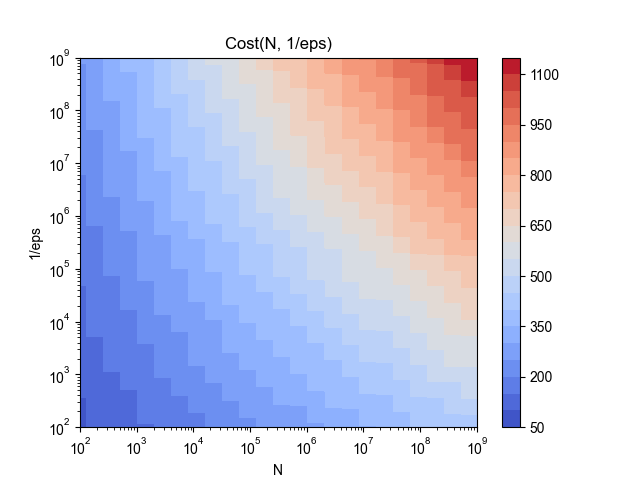}
    \caption{The color map diagram for the Toffoli cost of preparing $\ket{\psi_g}$ where $g(x)= \exp(-|x|)$ in the domain $x \in \{-0.5, -0.5 + \delta/2, \ldots, -\delta/2, \delta/2, \ldots, 0.5 - \delta/2, 0.5\}$.
    The horizontal and vertical axis are $N= 1/\delta \in [10^2, 10^9]$ and $1/\epsilon \in [10^2, 10^9]$.}
    \label{fig:ExpExample}
\end{figure}

\subsection{Power law: \texorpdfstring{$\ket{\psi_f} \propto \sum_{x} 1/x^\beta \ket{x}$}{Lg}}\label{subsec:ExamplesPowerLaw}

We consider the preparation of the state $\ket{\psi_f} \propto \sum_{x=1}^{2^n-1} 1/|x|^\beta \ket{x}$. 
We review two cases: (a) preparing the state in a one dimensional domain and (b) in a three dimensional domain.
For both cases, we provide explicit circuits for $\beta=1$, but the reference function $g$ is general for $\beta>0$.

\subsubsection{\texorpdfstring{$1/x^\beta$}{Lg} in one dimension}\label{subsubsec:ExamplesPowerLaw_1d}

We consider the prescription for the state in one dimension. 
Following the steps similar to Ref.~\cite{babbush2019quantum}, we partition the set of points $D = \{1, \ldots, 2^n-1\}$ into ``annuli'' indexed by $\mu$, where
\begin{equation}
    D_\mu = \{x: 2^{\mu - 1} \leq x < 2^{\mu}\}, \; \mu = 1, 2, \ldots n.
\end{equation}
The reference function we choose is a piecewise constant function (i.e. constant on each annulus $D_\mu$), i.e. a ziggurat. 
In particular, we choose $g_\mu(x)= \max_{x \in D_\mu} |f(x)|$ defined on $x \in D_\mu$, namely:
\begin{equation}\label{eq:powerlaw_ziggurat}
g_\mu(x) = \begin{cases}
2^{-\beta(\mu-1)},& \textrm{for}\;  x \in [2^{\mu-1}, 2^{\mu}),\\
0,& \textrm{otherwise}.
\end{cases}
\end{equation}
We plot the one-dimensional functions $f(x)= 1/x^\beta$ and $g= \sum_\mu g_\mu(x)$, for $\beta=1$ and $n=4$ in Fig.~\ref{fig:power_law}.

The circuit that prepares $\ket{\psi'_g}$ as in Eq.~\eqref{eq:LemmaCircuitForG} consists of two unitaries, $\PREP_{\v{c}}$ and $\sum_{\mu} \ketbra{\mu} \otimes \PREP_{g_\mu}$.
The unitary $\PREP_{\v{c}}$ prepares a state with amplitudes $c_\mu= g_\mu \sqrt{|D_\mu|}$, indexed by $\mu$. 
The action of $\PREP_{\v{c}}$ is given as

\begin{equation}
\PREP_{\v{c}}\ket{0} = \frac{1}{\mcalN} \sum_\mu c_\mu \ket{\mu},
\end{equation}
where normalization constant is $\mcalN = \sum_{\mu}|c_\mu|^2$. 

\begin{figure}[t]
\centering
\begin{subfigure}[t]{1\textwidth}
  \centering
  \includegraphics[width=0.6\linewidth]{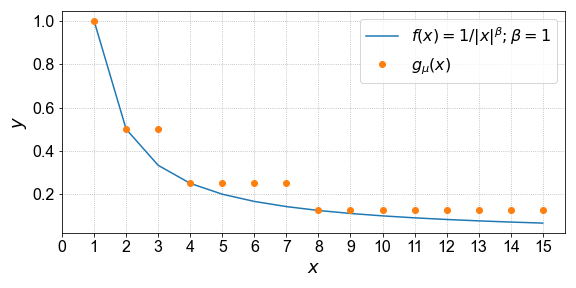}
  \caption{One-dimensional power law function $1/x^\beta$ for $n=4$ and $\beta=1$ is shown in blue. The reference function $g_\mu(x)$ as defined in Eq.~\eqref{eq:powerlaw_ziggurat} is shown in orange. There are $n$ $D_\mu$ regions indexed by $\mu \in [1, 4]$. One could think of these values (in blue) as the un-normalized amplitudes of the state we wish to prepare.}
  \label{fig:power_law}
\end{subfigure}\\ \vspace{1em}%
\begin{subfigure}[t]{.4\textwidth}
  \includegraphics[width=\linewidth]{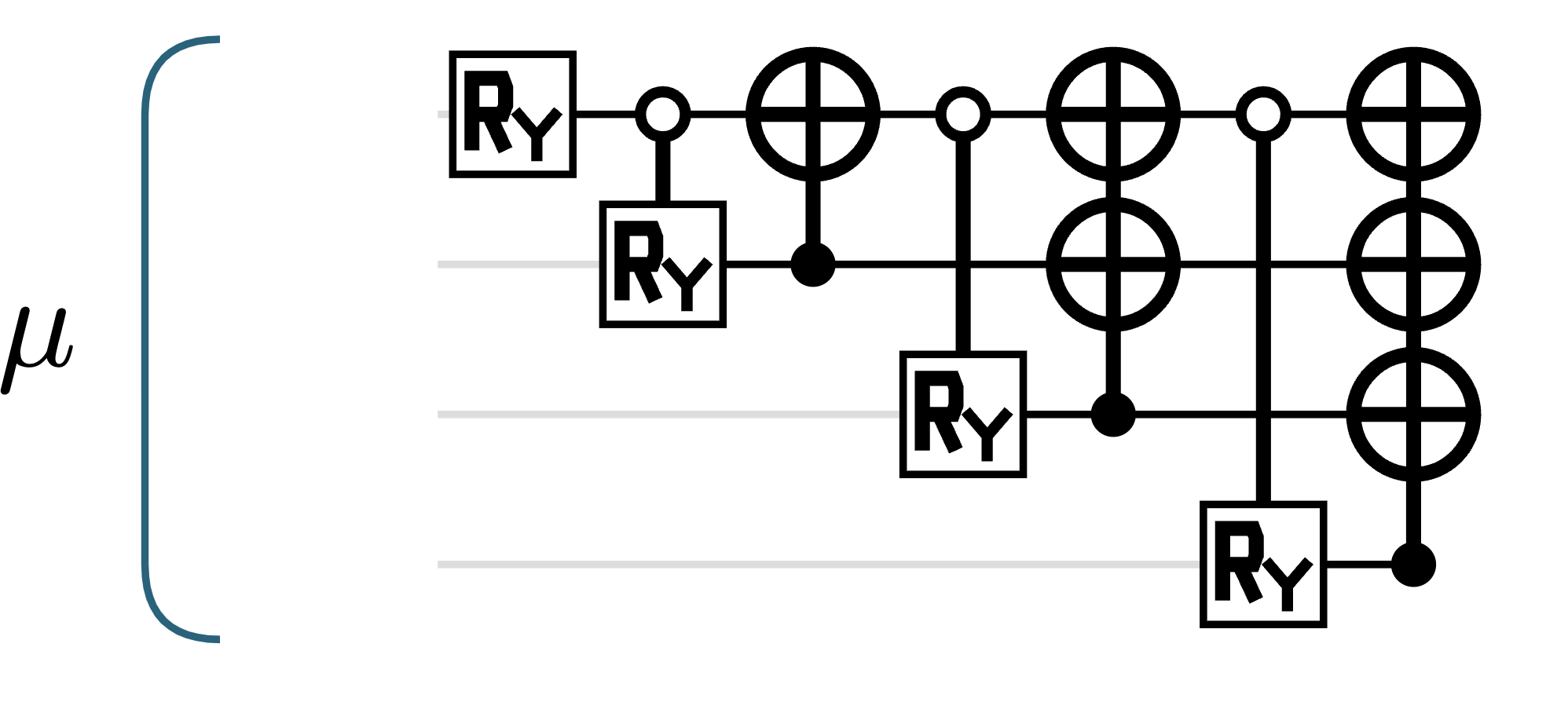}
  \caption{Circuit implementing $\PREP_{\v{c}}$ (see Eq.~\eqref{eq:powerlaw_ziggurat_state}) for when $n=4$ and $\beta = 1$. The $\mu$ register assumes a unary encoding.}
  \label{fig:power_law_prep_c}
\end{subfigure}\hspace{4em}
\begin{subfigure}[t]{.4\textwidth}
    \centering
    \includegraphics[width=\linewidth]{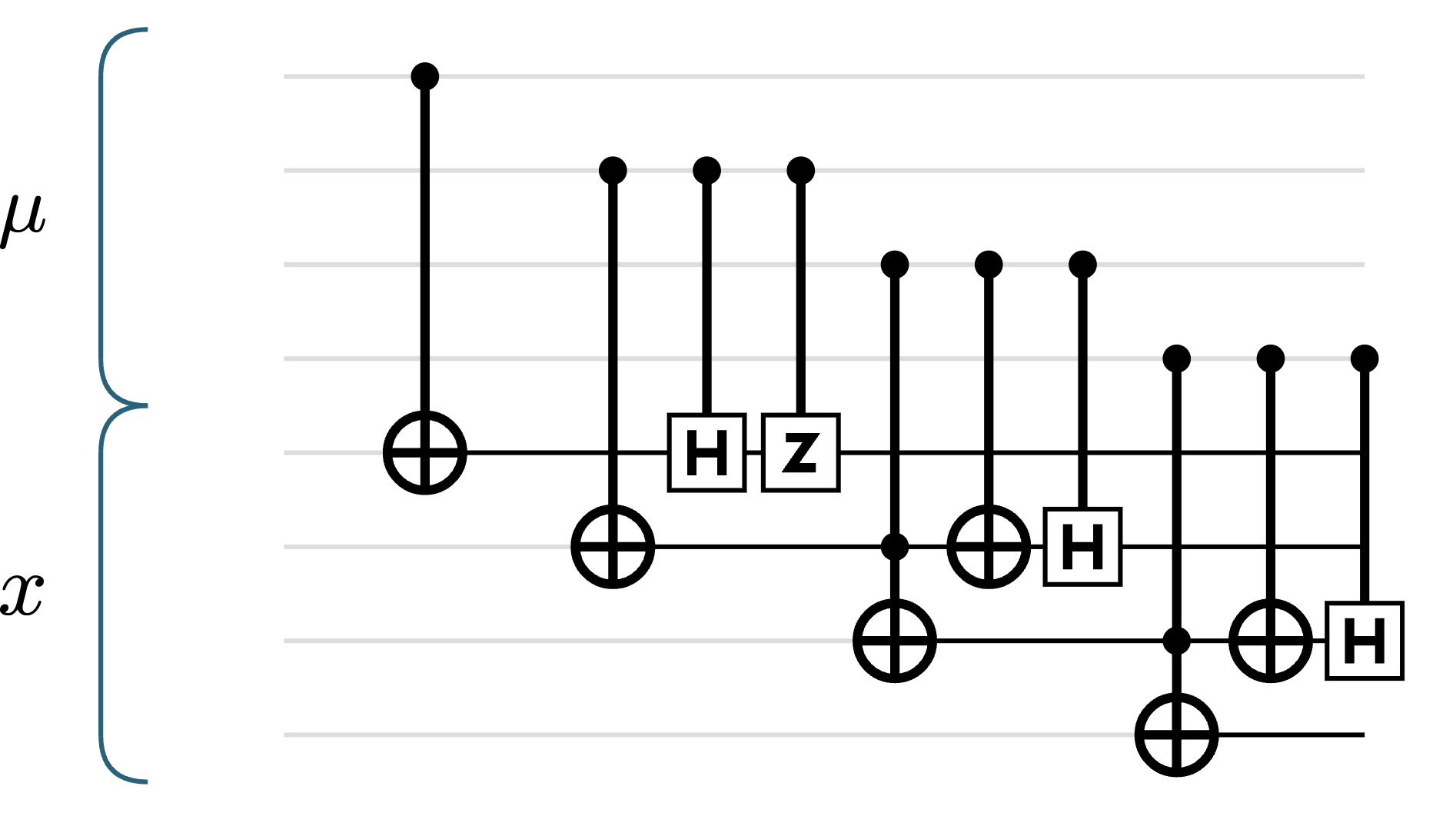}
    \caption{Circuit implementing $\sum_{\mu} \ketbra{\mu} \otimes \PREP_{g_\mu}$ (see Eq.~\eqref{eq:powerlaw_ziggurat_select}) for the one-dimensional case when $n=4$ and $\beta = 1$.
    }
    \label{fig:power_law_mu_select}
\end{subfigure}
\caption{(a) Example ziggurat reference function $g$ for the target function $f(x)= 1/x^\beta$, and the quantum circuits that prepares the reference state $\ket{\psi_g}$: (b) $\PREP_{\v{c}}$ and (c) $\sum_{\mu} \ketbra{\mu} \otimes \PREP_{g_{\mu}}$, for $\beta=1$, and $n=4$.}
\label{fig:ProfileAndCircuitforGInverse}
\end{figure}

\noindent One possible $\PREP_{\v{c}}$ that is consistent with the choice of the regions, such that $g_\mu= c_\mu/\sqrt{|D_\mu|}$ is given as
\begin{equation}\label{eq:powerlaw_ziggurat_state}
\PREP_{\v{c}}\ket{0} = \frac{1}{\mcalN} \sum_{\mu=1}^{n} 2^{\frac{(\mu-1)(1-2\beta)}{2}} \ket{\mu} = \frac{1}{\mcalN} \sum_{\mu=1}^{n} 2^{\frac{(\mu-1)(1-2\beta)}{2}} |0 \ldots 0 \underbrace{1 \ldots 1}_{\mu}\rangle,
\end{equation}

\noindent where the right-hand equality indicates that the register indices $\mu$ in unary. This unitary can be implemented by a ladder of controlled-$R_Y$ rotations with angle
$\alpha_k= 2 \arcsin\left[c_k/\sqrt{1- \sum^{k-1}_{i=1} |c_i|^2}\right]$, where $c_k$ have been normalized. The first $R_Y$ rotation with angle $\alpha_1= 2 \arcsin\left[c_1\right]$ is uncontrolled. 
A circuit implementation of $\PREP_{\v{c}}$, where $\beta=1$ and $n=4$, is shown in Fig.~\ref{fig:power_law_prep_c}.
Note that this is a deterministic implementation of $\PREP_{\v{c}}$, which differs from the circuit using controlled Hadamards and a flag qubit in Ref.~\cite{babbush2019quantum}, which is used within a LCU-type block-encoding.
Our implementation considers a state preparation task in a more general setting. 
When used as a subroutine in a block-encoding, one may prefer one implementation over the other depending on the trade-off between the resulting rescaling factor and gate count.
A more detailed discussion on this can be found in Sec.~\ref{sec:BEExamplePREPSELPREP}.

After the application of $\PREP_{\v{c}}$, the unitary $\sum_{\mu} \ketbra{\mu} \otimes \PREP_{g_\mu}$
prepares, controlled on $\mu$, a uniform superposition over $x \in D_\mu$, hence producing
a state with coefficients $g_{\mu}(x)= c_\mu/\sqrt{|D_\mu|}$.
Its action is given as:
\begin{equation}
\sum_{\mu} \ketbra{\mu} \otimes \PREP_{g_\mu} \left(\frac{1}{\mcalN} \sum_\mu c_\mu \ket{\mu} \otimes \ket{0} \right) = \frac{1}{\mcalN} \sum_{\mu=1}^n \frac{c_\mu}{\sqrt{|D_\mu|}} \sum_{x \in D_\mu} \ket{\mu}\ket{x},
\end{equation}
where the function $\mu(x) = \mu$ if $x \in D_\mu$. 
In the power law case, the action of $\sum_{\mu} \ketbra{\mu} \otimes \PREP_{g_\mu}$ is
\begin{equation}\label{eq:powerlaw_ziggurat_select}
\sum_{\mu} \ketbra{\mu} \otimes \PREP_{g_\mu}\left( \frac{1}{\mcalN} \sum_{\mu=1}^{n} 2^{\frac{(\mu-1)(1-2\beta)}{2}} \ket{\mu} \otimes \ket{0} \right) = \frac{1}{\mcalN_g} \sum_{\mu=1}^{n} \sum_{x \in D_\mu} 2^{-\beta(\mu-1)} \ket{\mu} \otimes \ket{x} 
\end{equation}
which is the desired reference state $\ket{\psi_g}$ that we use in our algorithm. 
A circuit implementing $\sum_{\mu} \ketbra{\mu} \otimes \PREP_{g_\mu}$ for the one-dimensional case is shown in Fig.~\ref{fig:power_law_mu_select}. 
For the first, i.e. the least significant bit of $x$, we apply a CNOT to remove any amplitude on $x=0$. For each of the following bits, we apply an adder on a subset of qubits to shift (i.e. only consider points in the region $D_\mu$) followed by a controlled Hadamard to prepare a superposition. Specifically, for the second bit, we shift by $+2$, and apply a controlled-Hadamard followed by a CZ gate to fix the local phase. 
From the third bit onwards (i.e. $j=3,...,n$), one shifts by $2^{j-2}$ 
on the first $j$ qubits followed by a controlled-Hadamard. 
This addition requires only one Toffoli per iteration, as observed in Fig.~\ref{fig:power_law_mu_select}.
This completes the second part of the quantum circuit for $\PREP_g$, which is a detailed implementation of the circuit given in Fig.~\ref{fig:CircuitForG} for our case.

Before completing a full resource count with these choices of the reference function for $\beta>0$, i.e., $g= \sum_{\mu} g_{\mu}$, and the target function $f= 1/x^\beta$, we give a lower bound for the amplitude of the target state, and therefore an upper bound on the required number of rounds of amplitude amplification, as follows.
Note that similar bound on success probability holds for cases $\beta<0$ with a slightly modified reference function $g$.

\begin{lemma}[Probability of success for power law state preparation]\label{lem:powerlaw_succprob}
Let $\ket{\psi_f} = \frac{1}{\mcalN_f}\sum^{N-1}_{x=1} \frac{1}{x^\beta} \ket{x}$ be the target state to be prepared, for a given $\beta >0$.
Let the reference function $g= \sum^n_{\mu = 1} g_\mu: [1,N-1] \mapsto \mathbb{R}$ be chosen as in Eq.~\eqref{eq:powerlaw_ziggurat}, i.e.,
\begin{equation}
g_\mu(x) = \begin{cases}
2^{-\beta(\mu-1)},& \textrm{for}\;  x \in [2^{\mu-1}, 2^{\mu}),\\
0,& \textrm{otherwise}.
\end{cases}
\end{equation}
For $\beta \neq 1/2$, the probability of success, $P_{\textrm{succ}}:= (\mcalN_f/\mcalN_g)^2$, then satisfies
\begin{align}
P_{\textrm{succ}} \geq \frac{-1 + 2^{1-2\beta}}{1-2\beta}.
\end{align}
For large $|\beta|$, this implies that the number of amplitude amplification rounds, $R$, needed in Alg.~\ref{algo:QSPA} satisfies
\begin{align}
R= \mathcal{O}\left(\beta\right).
\end{align}
For $\beta=1/2$, $P_{\textrm{succ}}\geq \ln 2$, hence $R=1$.
\end{lemma}

\begin{proof}
The denominator of $P_{\textrm{succ}}$ can be written as

\begin{align}
\mcalN^2_g &= \sum_{\mu=1}^n |D_\mu|g^2_\mu\\
&= \sum^n_{\mu = 1} \frac{2^{\mu - 1}}{2^{2\beta(\mu - 1)}} = \sum^{n-1}_{\mu=0} \frac{2^\mu}{2^{2\mu \beta}}= \sum^{n-1}_{\mu = 0} 2^{\mu (1 - 2\beta)}\\
&= \begin{cases}
\frac{1-2^{n(1-2\beta)}}{1-2^{1-2\beta}},\; \textrm{for} \; \beta \neq 1/2;\\
n, \textrm{for}\; \beta=1/2.
\end{cases}
\end{align}

\noindent For the numerator of $P_{\textrm{succ}}$, or $\mcalN_f^2$, for $\beta \neq 1/2$ we have the following lower bound:

\begin{align}
    \int^{2^n}_1 \frac{1}{x^{2 \beta}} dx \leq \sum^{2^n}_{x=1} \frac{1}{x^{2\beta}}
\end{align}
which implies
\begin{align}
    \frac{2^{n(1-2\beta)} - 1}{1-2\beta} \leq \sum^{2^n}_{x =1} \frac{1}{x^{2\beta}}.
\end{align}

\noindent Then we obtain with the following lower bound on $P_{\textrm{succ}}$, for $\beta \neq 1/2$:

\begin{align}
    P_{succ} &\geq \frac{-1 + 2^{1-2\beta}}{1-2\beta}.
\end{align}
Given that $R= \left\lceil \frac{\pi}{4\arcsin(P_{\textrm{succ}})} - \frac{1}{2} \right\rceil$, for large $|\beta|$, i.e., for small $P_{\textrm{succ}}$, we obtain the claimed asymptotic scaling of $R$ in $|\beta|$.
For $\beta= 1/2$, the numerator is bounded from below as $\sum^{2^n}_1 \frac{1}{x} dx \approx n \ln 2$.
Hence, $P_{\textrm{succ}}\geq \ln 2$, which implies $R=1$ amplitude amplification round.
\end{proof}

\noindent We now tabulate the cost of preparing the one-dimensional power law state with $\beta=1$ and collate the costs in Table~\ref{table:powerlaw_REs} in terms of Toffoli gates. 
Firstly, the preparation of the reference state $\ket{\psi_g}$ is decomposed into operations $\PREP_{\v{c}}$ and $\sum_\mu \ketbra{\mu} \otimes \PREP_{g_{\mu}}$:
\begin{enumerate}
    \item The circuit for $\PREP_{\v{c}}$ consists of a single $R_{Y}$ and $(n-1)$ singly-controlled $R_Y$ rotations. 
    Assuming each rotation is realized up to $\epsilon' (= \epsilon/r_{\textrm{tot}})$ accuracy, where $r_{\textrm{tot}}$ is the total number of rotations in the quantum circuit, the number of Toffoli gates is then upper bounded by $(1+ 2(n-1))(\log(1/\epsilon'))$, catalyzed by the Fourier state.
    \item The circuit for $\sum_\mu \ketbra{\mu}{\mu} \otimes \PREP_{g_\mu}$ 
    uses $(n-2)$ Toffolis for the adders/shifts and $(n-1)$ controlled-Hadamards for generating the superposition of states in the annuli. 
    This results in $(n-2) + (n-1) = 2n-3$ Toffoli gates.
\end{enumerate}
After preparing the reference state, we consider the following inequality test:
\begin{align}\label{eq:InequalityInverse1D}
    m x \leq M 2^{\mu-1},
\end{align}
which is equivalent to the original inequality $m 2^{-(\mu-1)} \leq m/x$.

\noindent
The overall cost is broken down as follows:
\begin{enumerate}
    \item The LHS ($=m x$) is a multiplication between a $b_M$-bit number ($b_M= \lceil \log M \rceil$) and an $n$-bit number. 
    The cost is ($2 b_M  n -  \max\{b_M, n\}$) Toffoli gates.
    \item The RHS ($=M 2^{\mu-1}$) is a multiplication between a number $M$, and $2^{\mu-1}$, which can simply be achieved by reading the unary one-hot encoded $\mu$ register as a binary number.
    This costs no Toffoli gates since the multiplication can be performed by padding the $\mu$ register with $b_M$ zeros.
    
    \item Lastly, the cost of the comparator is $(b_M + n)$ Toffoli gates.
\end{enumerate}

\noindent We note that one may also use a recursive method for preparing the $1/x^\beta$ state. 
For a large $\beta$, we can use the $1/x^{\beta-1}$ preparation circuit for preparing the $1/x^{\beta}$ state. 
One may then iterate this and use $\PREP_{1/x^{\beta-k}}$ in the quantum circuit for realizing $\PREP_{1/x^{\beta-k+1}}$, etc.
We reserve investigating this alternative method for future work.

\begin{table}[h!]
\begin{center}
\begin{tabular}{|cc|c|}
\hline
\multicolumn{2}{|c|}{Subroutine} & Toffoli count \\ \hhline{===}
\multicolumn{1}{|c|}{$\PREP_g$} & $\PREP_{\v{c}}$ & $(1+2(n-1)) \log(1/\epsilon')$ \\ \cline{2-3} 
\multicolumn{1}{|c|}{} & $\sum_\mu \ketbra{\mu} \otimes \PREP_{g_\mu}$ & $2n-3$ \\ \hline
\multicolumn{2}{|c|}{$U_g$ (RHS)}  & 
$2 b_M n -  \max\{b_M, n\}$
\\ \hline
\multicolumn{2}{|c|}{$U_f$ (LHS)}  & $0$ \\ \hline
\multicolumn{2}{|c|}{$\Comp$} & $(b_M + n)$ \\ \hhline{===}
\multicolumn{2}{|c|}{Number of AA rounds, $R$} & $1$
\\ \hline
\end{tabular}
\end{center}
\caption{The total resource cost of Alg.~\ref{algo:QSPA} in terms of Toffoli gate count to prepare the power law state $\sum_{x=1}^{2^n - 1} \frac{1}{x^\beta}\ket{x}$, with $\beta=1$ broken down by the cost of the subroutines.
The clause is based on Eq.~\eqref{eq:InequalityInverse1D}, where $U_g$ and $U_f$ computes the left and right hand sides to ancilla, respectively.}
\label{table:powerlaw_REs}
\end{table}

\subsubsection{\texorpdfstring{$1/|\v{x}|^\beta$}{Lg} in three dimensions}\label{subsubsec:ExamplesPowerLaw_3d}

For the three-dimensional case,
\footnote{For power law states on $d$-dimensional anisotropic lattices, one can construct the domain partition as follows. First, specify a base (greatest common denominator) qubit precision $l^*$ followed by scaling factors $x_\alpha,\alpha=1,...,d$ such that each dimension has $l_\alpha=x_\alpha (l^*-1)+1$ qubits (cubic when $x_\alpha=1 \, \forall\alpha$). Second, rescale the domains of the reference state by the factor $x_\alpha$ in the exponent to give $l^*$ parallelepipedal annular regions. For the case of $\beta=1$, the probability of success is approximately $P_{\mathrm{succ}}\approx\frac{(2^{x-2}-1)(2^{l^{*}-1}-1)^{d-2}B_{d}(-2)}{2^{2}(1-\frac{1}{2^{x}})(2^{l^{*}(x-2)}-2^{x-2})}$ where $B_{d}(s)$ is the box integral given in (Bailey, Crandall 2010) and $x=\sum_{\alpha=1}^d x_\alpha$. For general lattices in physical/chemical scenarios, e.g. reciprocal lattices for plane waves, the physical plane wave vector enters only in the computation of the target function.
(Link for reference: https://www.ams.org/journals/mcom/2010-79-271/S0025-5718-10-02338-0/home.html)}
we consider preparing the state
\begin{align}
    \ket{\psi_f} \propto \sum_{x \in G} 1/|\v{x}|^\beta \ket{x_0, x_1, x_2},
\end{align}
where $x \in G = [0, N^{1/3}-1]^{\times 3} \backslash \{ (0,0,0) \}$. That is, 
\begin{align}
    \v{x}=(x_0, x_1, x_2), \text{where} \ x_j \in [0, L-1],
\end{align}
where $L = N^{1/3}$. Recall that $n = \lceil \log_2 N \rceil$. We define $l = \lceil \log_2 L \rceil$.
For this function, the $g_\mu(\v{x})$ is defined as 
\begin{align}
    g_\mu(\v{x}) = \begin{cases}
2^{-\beta(\mu-1)},& \textrm{for}  \;\v{x}= (x_0,x_1,x_2): 2^{\mu-1} \leq \max\{|x_0|, |x_1|, |x_2|\} < 2^{\mu},\\
0,& \textrm{otherwise}.
    \end{cases}
\end{align}
where $\mu \in [1, l]$ indexes subdomains which are cuboidal annuli.

To construct the state preparation circuit, we re-use and modify some circuit components from Ref.~\cite{babbush2019quantum}.
The inequality tests used in Ref.~\cite{babbush2019quantum} and in this work slightly differ due to the difference in settings, i.e. use of this state preparation within a LCU-type block-encoding versus for a general state preparation.
The modifications are due to the difference in the domain $G$, and the fact that we are explicitly concerned with the state preparation rather than the block-encoding problem.
More on the block-encoding problem and how it is related to Ref.~\cite{babbush2019quantum} is discussed in Sec.~\ref{sec:PREPSELPREP}.

In Ref.~\cite{babbush2019quantum}, the $\PREP_{\v{c}}$ uses a cascade of singly-controlled Hadamards followed by flagging on an ancilla. 
The circuit for $\sum_\mu \ketbra{\mu}{\mu} \otimes \PREP_{g_\mu}$ involves sampling all the possible values of $x_0$, $x_1$, and $x_2$ in the $\mu$-th cube and discard the invalid points, i.e. elements contained in the previous $(\mu-1)$-th cube. We show the overall circuit in Fig.~\ref{fig:powerlaw_prep_g_3d}.
\begin{figure}[t]
    \centering    
    \includegraphics[width=0.7\linewidth]{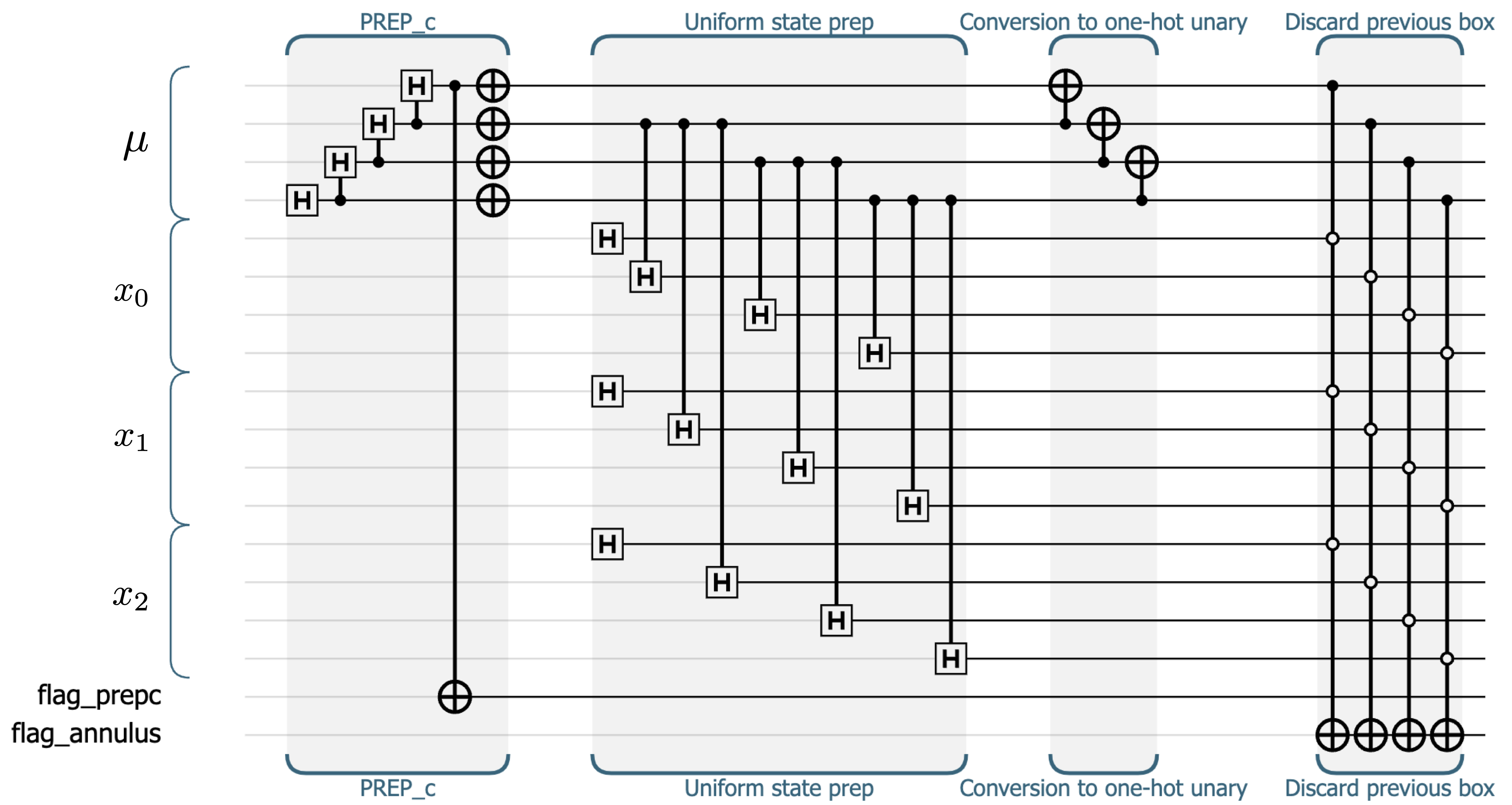}
    \caption{Circuit preparing $\text{PREP}_g$ in three dimensions for $\beta=1$ and $n=4$. This follows the circuit prescription in Ref~\cite{babbush2019quantum} but is simplified since we assume each $x_j$ to be positive.
    }
    \label{fig:powerlaw_prep_g_3d}
\end{figure}
The inequality test we implement is the following:
\begin{align}
   m \ 2^{-(\mu-1)} \leq \frac{M}{\sqrt{x_0^2 + x_1^2 + x_2^2}},
\end{align}
which can be squared and rearranged to
\begin{align}\label{eq:InequalityInverse3D}
   m^2 (x_0^2 + x_1^2 + x_2^2) \leq M^2 2^{2(\mu-1)}
\end{align}
to avoid computing the square root operation.
We review the Toffoli cost for this state preparation below and summarize in Table~\ref{table:powerlaw_3d_REs}.

In the $\text{PREP}_g$, we again have subcosts from $\PREP_{\v{c}}$ and $\sum_\mu \ketbra{\mu}{\mu} \otimes \PREP_{g_\mu}$:
\begin{enumerate}
    \item $\PREP_{\v{c}}$ uses $(l-1)$ controlled-Hadamards, which costs $(l-1)$ Toffoli gates.
    \item For $\sum_\mu \ketbra{\mu}{\mu} \otimes \PREP_{g_\mu}$, we have $3(l-1)$ controlled-Hadamards from the uniformly preparing states in $\v{x}$ and $l$ four-qubit-controlled Toffolis from discarding previous boxes. 
    This totals $3(l-1) + 3l = 6l-3$ Toffoli gates. 
\end{enumerate}
For the inequality test, we have:
\begin{enumerate}
    \item For the LHS ($=m^2 (x_0^2 + x_1^2 + x_2^2)$), we implement the circuit for computing the sum of three squares of $n$-bit numbers, which requires $3l^2-l-1$ Toffolis \cite{babbush2019quantum}. 
    This results in a $2l+2$-bit number.
    We square $m$, which costs $2 b_M^2 - b_M$ Toffoli gates and produces a $2b_M$-bit number.
    To multiply these two numbers together, the cost is $2(2b_M)(2l+2) - \max\{2b_M, 2l+2\}$ Toffoli gates.
    The subtotal is $3l^2-l-1 + 2 b_M^2 - b_M + 2(2b_M)(2l+2) - \max\{2b_M, 2l+2\}$ Toffoli gates.

    \item The RHS ($=M^2 2^{2(\mu-1)}$) is again a multiplication between $M$, a number, and $(2^{(\mu-1)})^2$. 
    As was shown in Ref~\cite{babbush2019quantum}, one could pad the $\mu$ register to implement this multiplication, thus at no Toffoli gate cost.
    \item Lastly, the cost of the comparator is $(b_M + 2l + 2)$ Toffoli gates.
\end{enumerate}

\begin{table}[h!]
\begin{center}
\begin{tabular}{|cc|c|}
\hline
\multicolumn{2}{|c|}{Subroutine} & Toffoli count \\ \hhline{===}
\multicolumn{1}{|c|}{$\PREP_g$} & $\PREP_{\v{c}}$ & $(l-1)$ \\ \cline{2-3} 
\multicolumn{1}{|c|}{} & $\sum_\mu \ketbra{\mu} \otimes \PREP_{g_\mu}$ & $6l-3$ \\ \hline
\multicolumn{2}{|c|}{$U_g$ (LHS)}  & 
$3l^2-l-1 + 2 b_M^2 - b_M + 2(2b_M)(2l+2) - \max\{2b_M, 2l+2\}$
\\ \hline
\multicolumn{2}{|c|}{$U_f$ (RHS)}  & $0$ \\ \hline
\multicolumn{2}{|c|}{$\Comp$} & $b_M + 2l + 2$ \\ \hhline{===}
\multicolumn{2}{|c|}{Number of AA rounds, $R$} & $1$ 
\\ \hline
\end{tabular}
\end{center}
\caption{The total resource cost of Alg.~\ref{algo:QSPA} in terms of Toffoli count to prepare the three-dimensional power law state proportional to $\sum_{x \in G} \frac{1}{|\v{x}|^\beta}\ket{\v{x}}$, with $\beta=1$ broken down by the cost of the subroutines. 
Here, $l = \lceil \log_2(L) \rceil$, where $L = N^{1/3}$.
The clause is based on Eq.~\eqref{eq:InequalityInverse3D}, where $U_g$ and $U_f$ computes the left and right hand sides to ancilla, respectively.}
\label{table:powerlaw_3d_REs}
\end{table}

\begin{figure}
\centering
\begin{subfigure}{.45\textwidth}
  \centering
  \includegraphics[width=.99\linewidth]{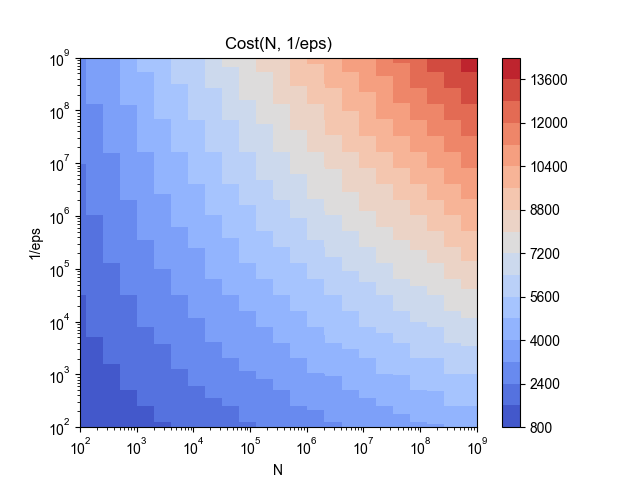}
  \caption{$f(x)= 1/x$}
  \label{fig:InverseExample}
\end{subfigure}%
\begin{subfigure}{.45\textwidth}
  \centering
  \includegraphics[width=.99\linewidth]{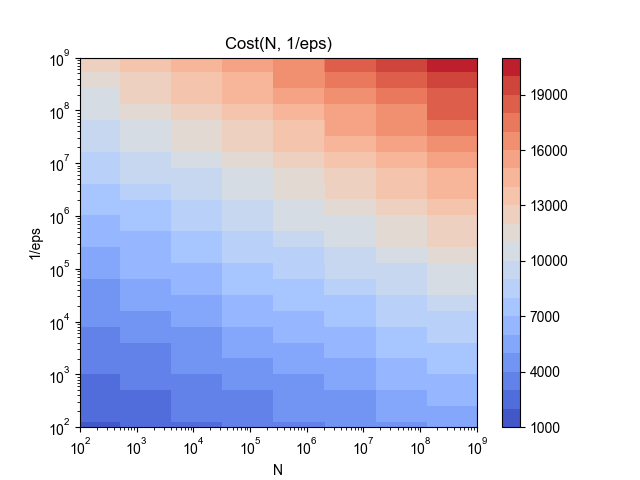}
  \caption{$f(\v{x})= 1/|\v{x}|^2$}
  \label{fig:Inverse2Example}
\end{subfigure}
\caption{The color map diagram for the Toffoli cost of preparing $\ket{\psi_f}$ where (a) $f(x)= 1/x$ in the domain $x \in \{1, 2, \ldots, N\}$, and (b) $f(\v{x})= 1/|\v{x}|$ in the domain $\v{x} \in [-N^{1/3}, \ldots, -1, 0, 1, \ldots, N^{1/3}]^{\times 3}$.
The horizontal and vertical axis are $N= 1/\delta \in [10^2, 10^9]$ and $1/\epsilon \in [10^2, 10^9]$, respectively.}
\label{fig:InverseExampleCombined}
\end{figure}

\subsection{Gaussian: \texorpdfstring{$\ket{\psi_f} \propto \sum_x \exp{ -(x-a)^2/{\sigma^2}} \ket{x}$}{Lg}}\label{subsec:ExamplesGaussian}

\noindent The target state is given by the coefficients that are proportional to a Gaussian given by

\begin{align}
\ket{\psi_f} \propto \sum_x \exp(-(x-a)^2/\sigma^2)\ket{x},
\end{align}

\noindent where, without loss of generality, we assume a Gaussian centered at $x=0$, i.e., $a=0$.
Our results straightforwardly generalize to the cases for which $a \neq 0$.
We pick the domain $x \in  \{-0.5, -0.5 +\delta , \ldots, -\delta/2, \delta/2 , \ldots, 0.5 - \delta, 0.5\}$ for a given $0 < \delta \ll 1$. 
We remark that the Gaussian state preparation is easy to handle with previous methods when $\sigma \sim \delta$ and $\sigma \sim N\delta$.
For $\sigma \sim \delta$, the Gaussian is highly peaked at its center, hence it is more efficient to use a standard Grover-Rudolph circuit for preparing a state with nonzero coefficients for $|x| \leq \delta + \log(1/\epsilon)$, and the rest is set to zero, which implies a state with only $N_0 \ll N$ coefficents.
In the latter case, for $\sigma \geq c \delta N$ for a high enough constant $c < 1$, we could use the usual black-box method with the uniform state reference~Ref.~\cite{sanders2019black} or with a ziggurat shaped reference state with extended domains given as above or in Ref.~\cite{bausch2022fast}, which will result in only a moderate amount of amplitude amplification rounds.
The interesting case for which our results would lead to an improvement lies in the region where 
\begin{align}\label{eq:GaussianInterestingRegion}
\delta \ll \sigma \ll \delta N=1,
\end{align}
which occurs when the decay of the Gaussian is significant over the domain and resolved well. For these cases, using a reference function $g$ with exponentially decaying tails is an appropriate way of reducing the number of amplitude amplification steps $R$ up until $R= 1$, see Sec.~\ref{subsec:DesignGuidelines}.
Hence this choice of reference state improves the results of Ref.~\cite{bausch2022fast} from $R= \mathcal{O}(\sqrt{N/(\sigma/\delta))}$ to simply $R=1$.
By choosing an appropriate reference state defined by $g$, we can minimize the number of amplitude amplification steps in the algorithm.
The following choice, for example, implies only one round of amplitude amplification round:
\begin{align}
g(x)= \begin{cases}
1, & \textrm{for} \; |x| \leq \sigma, \\
\exp(- |x|/\sigma), & \textrm{for} \; |x| > \sigma.
\end{cases}
\end{align}
We collect these remarks as follows.
\begin{lemma}[Probability of success for Gaussian state preparation]\label{lem:Guassian_succprob}
Let $\ket{\psi_f} \propto \sum_{x \in [-0.5, 0.5]_\delta} \exp(-x^2/\sigma^2) \ket{x}$ be the target state, where $[-0.5,0.5]_\delta:= \{-0.5, -0.5+\delta, \ldots, -\delta/2, +\delta/2, \ldots, 0.5\}$ with $0<\delta= \frac{1}{N} \ll 1$ and $\delta \ll \sigma \ll 1$.
Let the reference function $g$ be chosen such that
\begin{align}
g(x)= \begin{cases}
1, & \textrm{for} \; |x| \leq \sigma, \\
\exp(- |x|/\sigma), & \textrm{for} \; |x| > \sigma.
\end{cases}
\end{align}
The probability of success, $P_{\textrm{succ}}:= (\mcalN_f/\mcalN_g)^2$, then satisfies
\begin{align}
P_{\textrm{succ}} \approx 0.529.
\end{align}
This implies that $R=1$ amplitude amplification round is sufficient in Alg.~\ref{algo:QSPA}.
\end{lemma}

\begin{proof}
In the regime $\delta \ll 1$, the denominator $\mcalN^2_g$ can simply be approximated by
\begin{align}
\frac{2\sigma}{\delta} + \frac{2}{\delta}\int^{0.5}_\sigma dx \exp(-2x/\sigma).
\end{align}
This implies that
\begin{align}
\mcalN^2_g \approx \frac{2\sigma}{\delta} + \frac{\sigma}{e \delta}.
\end{align}
On the other hand, the numerator $\mcalN^2_f$ can be bounded from below, as follows:
\begin{align}
\mcalN^2_f \geq \frac{2}{\delta} \int^{0.5}_{0} \exp(-2x^2/\sigma^2).
\end{align}
In the regime where $\sigma \ll 1$, we can safely approximate this lower bound by
\begin{align}
\mcalN^2_f \gtrapprox \frac{2}{\delta} \int^{\infty}_{0} \exp(-2x^2/\sigma^2)= \frac{\sigma}{\delta} \sqrt{\frac{\pi}{2}}.
\end{align}
Putting together the approximate expressions for $\mcalN_g$ and $\mcalN_f$, we get
\begin{align}
P_{\textrm{succ}} \approx \frac{\sqrt{\pi}}{\sqrt{2}\left( 2 + e^{-1} \right)} \approx 0.529 > \frac{1}{4}.
\end{align}
Hence this implies that one round of amplitude amplification is sufficient to prepare the target Gaussian state.
\end{proof}

Note that the choice of $g$ given in Lemma~\ref{lem:Guassian_succprob} implies the following slightly modified inequality to be tested by the comparator when $x > \sigma$, which is:
\begin{align}\label{eq:GaussianComparator}
m \leq M e^{1/4} \exp\left[-\left( \frac{|x|}{\sigma} - \frac{1}{2} \right)^2\right].
\end{align}

\noindent where $M$ is large enough to result in at most an $\epsilon$ approximation error in the target state.
The final gate count depends on the gate cost of the subroutines as well, and for cases $\delta \ll \sigma \ll \delta N$, the bottleneck is computing the Gaussian to a register for the comparator (or the rotation angles via ratios of integrals of Gaussians), rather than the cost of preparing a reference state. Hence, it is not worth preparing a cheaper reference state (e.g. a ziggurat with extended domains such as in Section~\ref{subsec:ExamplesZiggurat}) at the cost of a larger number of amplitude amplification steps. Rather, it is best to minimize the number of the amplitude amplification steps by preparing a slightly more expensive reference (e.g., exponential, such as in Section~\ref{subsec:ExamplesExponential}).
\\

As we have seen the bottleneck is mostly in realizing the comparator given in Eq.~\eqref{eq:GaussianComparator}, e.g., computing the function values up to relative-precision $\epsilon_0$. 
With the goal of minimizing the number of Toffoli gates, we adopt a method that uses a mixture of look-up tables (which requires QROMs) and multiplication, similar to the one in Ref.~\cite{poirier2021efficient}.
For a numerical resource count, we make the following choices:
The function is given by $f(x)= \exp(-(2^3 x)^2)$, i.e., $\sigma= 2^{-3}$, in the domain $x \in \{-0.5, -0.5 + \delta, \ldots, -\delta/2, \delta/2, \ldots, 0.5\}$ for a given $\delta \in [2^{-7}, 2^{-26}]$ which implies $N \in (10^2, 10^{9})$.
We also analyze the resource count for various levels of errors, in particular in the region $\epsilon \in [10^{-2}, 10^{-9}]$.
For these choices, the number of required amplitude amplification rounds is $R=1$ with a simple choice of ziggurat as the reference function, such as
\begin{align}\label{eq:ExampleGaussianGZiggurat}
g(x)= \begin{cases}
1, & \textrm{for} \; |x| \leq 2^{-2}, \\
2^{-5}, & \textrm{for} \; |x| > 2^{-2}.
\end{cases}
\end{align}
The cost of $\PREP_g$ is just the cost of a single $Y$-rotation with an angle that is chosen to be precise up to, say, $\epsilon/8$.
This would take $\log (8/\epsilon)$ Toffoli gates.
It is assumed that the rest is performed with Hadamards, and shifts, and costs no (or only minor amount of) Toffoli gates.
$U_g$ costs no gates, given that the RHS of the comparator $m$ is already registered in the sampling space register.
The bottleneck is the cost of $U_f$. 
To be precise, we use the following comparator:
\begin{align}
M \exp(-(x/2^{-3})^2) \geq  & m, \; \; |x|\leq 2^{-2},\\
2^5 M \exp(-(x/2^{-3})^2) \geq  & m, \;\; |x|> 2^{-2}.
\end{align}
We use $U_f$ to compute the LHS to an ancilla register, and we do not need to compute the RHS, because the value $m$ is already in a quantum register as $\ket{m}$. 
Hence, most of the cost comes from $U_f$.
Below, we go through the resource cost analysis of $U_f$.
We first compute $x'= x^2 2^6$ coherently to an ancilla. 
Then, with the optimization of Toffoli count in mind, we use a mixture of coherent arithmetic and data-loading in order to coherently compute $e^{-x'}$ over the full domain of $x$.
We then multiply the resulting number by $M$ or $2^5 M$ controlled on whether $|x|\leq 2^{-2}$ or $|x|> 2^{-2}$, respectively.
Note that these precomputed numbers are of $b$-bit precision, where $b= \lceil \log_2 M \rceil$, due to the fact that the comparator operates with an inequality that has an integer (i.e., $m$) on the RHS.
Then, the comparator results in a flag that coherently labels the state of the sampling space with $\ket{0}$ or $\ket{1}$.
Furthermore, $n= \lceil \log N \rceil$ is the number of qubits holding the system registers, and $b_M:= \lceil \log M \rceil$ with $M$ chosen to be a positive integer power of $2$.
We now list these steps and their cost more clearly below: 

\begin{itemize}
    \item Compute $x^2$ coherently to a new register of size $2(n-1)$ (the relevant size of $x$ register is $n-1$ excluding the sign bit): 
    The cost is $2(n-1)^2-(n-1)$ Toffoli gates.
    
    \item Divide the register $\ket{x^2}$ of size $2(n-1)$ into $k$ pieces, with each of size at most $\lceil 2(n-1)/k \rceil$. 
    Controlled on the qubits in the $k$-th piece, load the precomputed values $\exp(-x^2/2^{-6})$: 
    The cost is $[(k-1) (2^{\lceil 2(n-1)/k \rceil)}-1) + (2^{\lfloor 2(n-1)/k\rfloor}-1)]$ Toffoli gates.
    
    \item Multiply the loaded values in $k$ pieces with each other.
    There are $k-1$ multiplications of numbers stored in $b$-bit registers:
    The cost is $2(k-1)b^2 - (k-1)b$ Toffoli gates.
\end{itemize}

\noindent Note that $k$ is optimized to minimize the Toffoli gate count. 
This concludes the resource estimate for the implementation of the LHS of the comparator, i.e., $U_f$.
The comparator itself costs $2(2 b_M - 1)$ Toffoli gates, consisting of a controlled comparator of the \emph{integer} pieces (of $\log M$ bits) of the two numbers (i.e., LHS and RHS of the comparator).
According to Lemma~\ref{lem:ChoosingM}, one can choose $M= \max\{4/\epsilon, 2 \exp(12)\}$.
The resource costs are tabulated in Tab.~\ref{table:Tanh_REs}, and
see Fig.~\ref{fig:GaussianExample} for the Toffoli count estimates with varying $N \in [10^2, 10^9]$ and $1/\epsilon \in [10^2, 10^9]$ with an optimized $k=6$ for lowest count on the top right, i.e., for $1/\epsilon = N = 10^9$.

\begin{table}[h!]
\begin{center}
\begin{tabular}{|c|c|}
\hline
Subroutine & Toffoli count \\ \hline\hline
$\PREP_g$ & $\log(8/\epsilon)$ \\ \hline
$U_f$ (LHS) & \makecell{$2(n-1)^2 - (n-1) + (k-1)[(2^{\lceil 2(n-1)/k \rceil)}-1) + (2^{\lfloor 2n/k\rfloor}-1)]$\\ $+ 2(k-1)b^2 - (k-1)b$} \\ \hline
$U_g$ (RHS) & $0$ \\ \hline
$\Comp$ & $2(2b_M - 1)$ \\ \hline
Number of AA rounds, $R$ & $1$ \\ \hline
\end{tabular}
\end{center}
\caption{The resource cost to prepare the Gaussian state $\ket{\psi_f} \propto \sum_{x\in \{-0.5, -0.5 + \delta, \ldots, 0.5\}} \exp(-(x/2^{-3})^2) \ket{x}$ in terms of Toffoli count. 
$n= \lceil \log(1/\delta) \rceil$ (where $N= 1/\delta$) is the number of system qubits,
$b_M= \log M$, and $b= \min \{b_M, \lceil \log1/\tilde{\epsilon} \rceil\}$ is the bit-precision that classically precomputed values for parts of $\exp(-(2^3 x)^2)$ are loaded and their multiplication are performed.
$M= \max\{4/\epsilon, 2 \exp(12)\}$ is the dimension of the sampling space, chosen according to Lemma~\ref{lem:ChoosingM}.
$\epsilon$ is the accuracy of the target state given as in the definition of Problem~\ref{problem:QuantumStatePreparation}.
Finally, $k$ is a variable integer that is optimized for the minimum Toffoli gate count.}
\label{table:Gaussian_REs}
\end{table}

\begin{figure}
    \centering
    \includegraphics[scale=0.50]{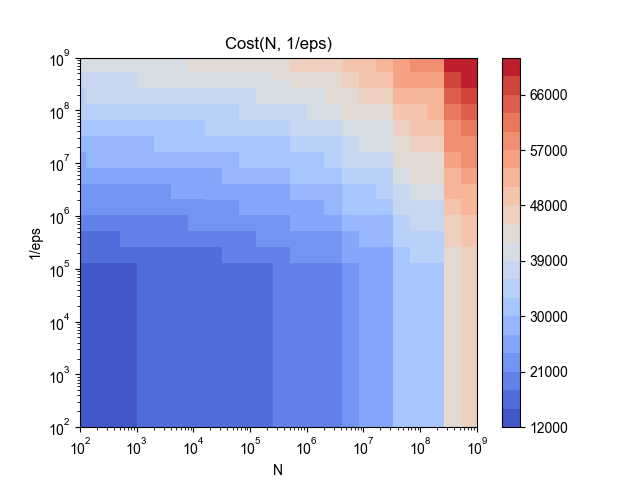}
    \caption{The color map diagram for the Toffoli gate cost of preparing $\ket{\psi_f}$ where $f= \exp(-(2^3 x)^2)$ in the domain $x \in \{-0.5, -0.5 + \delta/2, \ldots, -\delta/2, \delta/2, \ldots, 0.5 - \delta/2, 0.5\}$.
    The horizontal and vertical axis are $N= 1/\delta \in [10^2, 10^9]$ and $1/\epsilon \in [10^2, 10^9]$.}
    \label{fig:GaussianExample}
\end{figure}

\subsection{Hyperbolic tangent: \texorpdfstring{$\ket{\psi_f} \propto \sum_x \tanh(x)\ket{x}$}{Lg}}\label{subsec:ExamplesTanh}

The target state is given by the coefficients that are proportional to a hyperbolic tangent given by

\begin{align}
\ket{\psi_f} \propto \sum_{x} \tanh(x)\ket{x}.
\end{align}

We pick the domain $x \in  \{-0.5, -0.5 +\delta , \ldots, -\delta/2, \delta/2 , \ldots, 0.5 - \delta, 0.5\}$ for a given $0 < \delta \ll 1$, and $N= 1/\delta$.
Since $\tanh(x)= \tanh(-x)$, similar to the Gaussian case, we can prepare the state for the domain $x>0$, and the final state can then be obtained with an additional sign bit in equal weight superposition.
Below, we first show that there is a reference state with a high enough success probability that one round of amplitude amplification is sufficient.

\begin{lemma}[Probability of success for $\tanh$ state preparation]\label{lem:thanh_succprob}
Let $\ket{\psi_f} \propto \sum_{x \in [-0.5, 0.5]_\delta} \tanh(x) \ket{x}$ be the target state, where $[-0.5,0.5]_\delta:= \{-0.5, -0.5+\delta, \ldots, -\delta/2, +\delta/2, \ldots, 0.5\}$ with $0<\delta= \frac{1}{N} \ll 1$ and $\delta \ll 1$.
Let the reference function $g$ be a constant such that $g(x)= 1/2$.
Then, the probability of success, $P_{\textrm{succ}}:= (\mcalN_f/\mcalN_g)^2$, satisfies
\begin{align}
P_{\textrm{succ}} \approx 0.704.
\end{align}
This implies that $R=1$ amplitude amplification round is sufficient in Alg.~\ref{algo:QSPA}.
\end{lemma}

\begin{proof}
The denominator is given exactly by
\begin{align}
\mcalN^2_g= \frac{1}{4\delta}
\end{align}
where we have used $g(x)= 1/2$.
The numerator is given by
\begin{align}
\mcalN^2_f \approx \frac{2}{\delta} \int^{0.5}_{0} (\tanh(x))^2 dx= \frac{2}{\delta} (0.5 - \tanh(0.5)) \approx \frac{0.176}{\delta}.
\end{align}
Putting together the approximate expressions for $\mcalN_g$ and $\mcalN_f$, we get
\begin{align}
P_{\textrm{succ}} \approx 0.704.
\end{align}
Hence this implies that one round of amplitude amplification is sufficient to prepare the target $\tanh$ state.
\end{proof}

\noindent Since the probability of success is greater or equal to $\frac{1}{4}$, only one round of amplitude amplification will be sufficient. 
$\PREP_g$ prepares a uniform state, which costs no (or in general substantially less compared to other subroutines) Toffoli gates.
For the rest, we first determine the clause for the comparator for the purpose of efficient implementation.
Given that
\begin{align}
    \tanh{x} = \frac{e^{2x}-1}{e^{2x}+1},
\end{align}
we can rewrite the inequality $ f(x)M \geq m g(x)$ as
\begin{align}\label{eq:tanhineq}
    e^{2x}\left(M-m g(x)\right) - m g(x) \geq M,
\end{align}
or equivalently,
\begin{align}\label{eq:tanhineq1}
   e^{2x}\left(1-m g(x)\right/M) - m g(x)/M \geq 1.
\end{align}
For $\tanh(x)$, the choice of the constant function $g(x)= 0.5$ is a good reference function on the domain $x \in [0,0.5]$ that leads to only one round of amplitude amplification.
Hence, the comparator uses the following inequality:
\begin{align}\label{eq:tanhineq2}
    e^{2x}(1-m/(2M)) - m/(2M) \geq 1.
\end{align}

To efficiently compute $e^{2x}$, we make use of $e^{2x} = e^{ 2\sum_{j=0}^{k-1} x_j } = \prod_{j=0}^{k-1}e^{2x_j} $. 
We perform this via a combination of QROMs and multiplications: $\ket{x}$ is a register of $\lceil\log_2 N\rceil$ qubits, and we divide them into $k$ sub-registers, each of at most $b_{\text{QROM}}= \lceil (\log_2 N)/k \rceil$ qubits.
Then, we use QROM to load $2^{b_\text{QROM}}$ values for each sub-register, and multiply them.
The classical values $e^{2x}$ are loaded with bit precision $b_\Delta+2$ to achieve final bit precision $b_\Delta$, and the multiplications are performed only with $b_\Delta+2$ bit precision. This costs:
\begin{itemize}
\item The cost of QROM: $2^{b_\text{QROM}}-2$ Toffolis
, $b_\Delta+2 $ qubits to store precomputed classical values, and
$b_\text{QROM}$ temporary ancillae qubits.

\item $k-1$ multiplication of registers of size $b_\Delta + 2$ costs: $(k-1)(2(b_\Delta+2)^2-b_\Delta - 2)$ Toffolis, $(k-1)(2b_\Delta+4)$ to store the result each multiplication to ensure reversibility, and $b_\Delta+2$ temporary ancillae qubits.
\end{itemize}

\noindent The rest, including the comparator can be performed as follows.
\begin{itemize}
    \item First, we compute $m/(2M)$:
    This costs no Toffoli gates due to $2M$ being an integer power of two, it can be achieved by padding $0$s and pushing the value further in the decimal points.
    
    \item Then we compute $e^{2x} (1-m/(2M))$.
    $(1-m/(2M))$ can be performed only via $X$ gates, with $(b_\Delta+2)$-bit precision.
    Then, we multiply with $e^{2x}$, which costs $2(b_\Delta+2)^2-(b_\Delta+2)$ Toffolis, $2b_\Delta+4$ qubits to store the result and $b_\Delta+2$ temporary ancillae qubits.
    
    \item Subtracting $m/(2M)$:
    The cost is $b_\Delta + 1$ Toffolis and $b_\Delta + 1$ temporary ancillae qubits.
    
    \item The comparator checks whether RHS is strictly below $1$, which only needs checking whether the integer part (consisting of two qubits) is zero. 
    This can be done via $1$ Toffoli and $1$ qubit to store the result.
\end{itemize}

\noindent A sketch of the circuit implementation for the inequality test is given in Fig.~\ref{fig:ineqTesttanh}, and the resource costs are tabulated in Tab.~\ref{table:Tanh_REs}.

\begin{table}[h!]
\begin{center}
\begin{tabular}{|c|c|}
\hline
Subroutine & Toffoli count \\ \hline\hline
$\PREP_g$ & $0$ \\ \hline
$U_f$ (LHS) & \makecell{$k(2^{b_\text{QROM}}-2) + ((k-1)(2(b_\Delta+2)^2-b_\Delta - 2))$\\ $ +(2(b_\Delta+2)^2-(b_\Delta+2)) +(b_\Delta + 1)$}
\\ \hline
$U_g$ (RHS) & $0$ \\ \hline
$\Comp$ & $1$ \\ \hline
Number of AA rounds, $R$ & $1$ \\ \hline
\end{tabular}
\end{center}
\caption{The resource cost to prepare the $\tanh$ state $\ket{\psi_f} \propto \sum_{x\in \{\delta, 2\delta, \ldots, 0.5\}} \tanh(x) \ket{x}$ in terms of Toffoli gate count. 
$n= \lceil \log_2(1/\delta) \rceil$ (where $N= 1/\delta$) is the number of system qubits.
$b_\Delta= \lceil\log_2 1/\tilde{\epsilon}\rceil$ is the bit-precision that functions in the comparator are computed.
$M$ is the number of qubits that represents the sampling space, chosen according to Lemma~\ref{lem:ChoosingM}.
$\epsilon$ is the accuracy of the target state given as in the definition of Problem~\ref{problem:QuantumStatePreparation}.
Finally, $b_{\text{QROM}}:= \lceil n/k \rceil$, and $k$ is a variable integer that is optimized for the minimum Toffoli count.}
\label{table:Tanh_REs}
\end{table}

\begin{figure}
    \centering
    \includegraphics[width=0.65\linewidth]{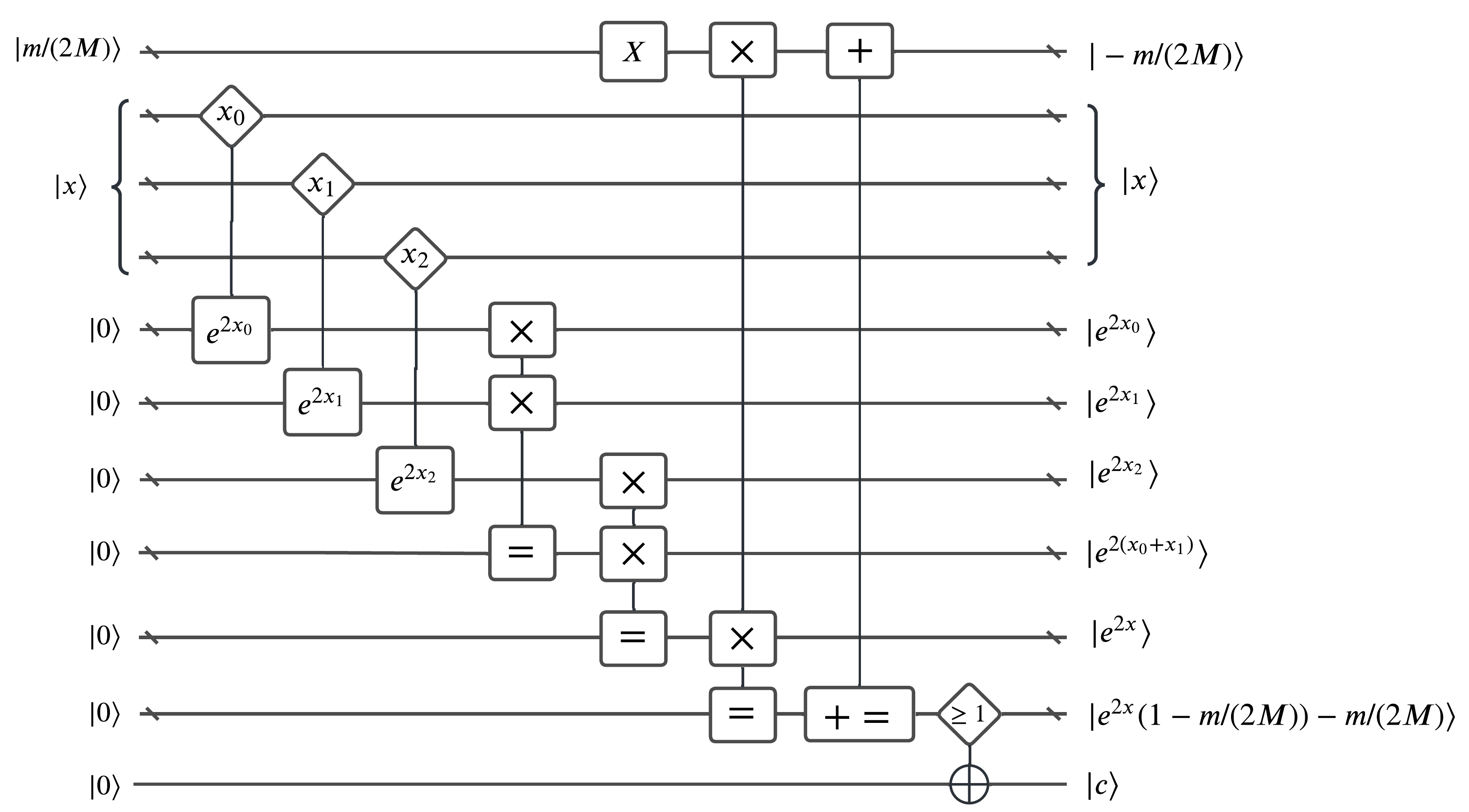}
    \caption{Inequality test implementation for $\tanh(x)$ for $k=3$}.
    \label{fig:ineqTesttanh}
\end{figure}

\noindent The bit precision $b_\Delta$ required with which the function computations are performed, must ensure to follow Lemma~\ref{lem:ChoosingM}.
For this, it is sufficient to compute $f(x)/g(x)$ with $\tilde{\epsilon}$ additive error.
The only quantity that will be approximately computed is $e^{2x}$.
Hence, let's start from $C(e^{2x})$ as the computed value of $e^{2x}$ on the computer, with additive error $\Delta$, i.e.,
\begin{align}
e^{2x} - \Delta \leq C(e^{2x}) \leq e^{2x} + \Delta
\end{align}
for all $x \in [0,0.5]$.
We wish to express $\Delta$ in terms of $\tilde{\epsilon}$, which is the additive error of computing the function $\tanh/2$ if we used the bare comparator $mg(x) \leq M f(x)$.
The inequality we need to satisfy reads as follows:
\begin{align}\label{eq:tanh-approx}
 & \frac{e^{2x}-1}{e^{2x}+1} - \frac{\tilde{\epsilon}}{2}  \leq \frac{C(e^{2x}) - 1}{C(e^{2x}) + 1} \leq \frac{e^{2x}-1}{e^{2x}+1} + \frac{\tilde{\epsilon}}{2}.
\end{align}

\noindent An additive error $\Delta \leq \tilde{\epsilon}$ ensures that Eq.~\eqref{eq:tanh-approx} is respected. 
See Fig.~\ref{fig:qre-tanh} for the Toffoli count estimates with varying $N \in [10^2, 10^9]$ and $1/\epsilon \in [10^2, 10^9]$
Lemma~\ref{lem:ChoosingM} implies that $M= \max\{4/\epsilon, N\}$, and $\tilde{\epsilon}= \min\{\epsilon/4, 1/N\}$ are valid choices.

\begin{figure}
    \centering
    \includegraphics[width=0.5\linewidth]{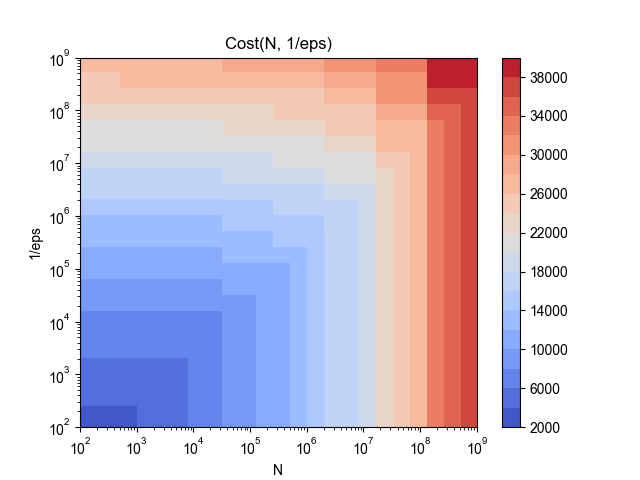}
    \caption{The color map diagram for the Toffoli cost of preparing $\ket{\psi_f}$ where $f= \tanh(x)$ in the domain $x \in \{\delta, -\delta/2, \delta/2, \ldots, 0.5 - \delta/2, 0.5\}$.
    The horizontal and vertical axis are $N= 1/\delta \in [10^2, 10^9]$ and $1/\epsilon \in [10^2, 10^9]$ with $k=3$.
    }
    \label{fig:qre-tanh}
\end{figure}

\section{Comparison with other methods}\label{sec:Comparison}

In this section, we first give a general comparison of the QRS-method proposed here with other methods that appear in the literature.
Then, we compare the Toffoli counts for specific states found in Sec.~\ref{sec:Examples} with the LKS method~\cite{low2018trading} (a Toffoli or T-gate count optimized version - with clean qubits - of the Grover-Rudolph algorithm) that assumes an arbitrary quantum state. 
While it is difficult to establish an absolute comparison between the two methods, our comparison is a good starting point for understanding the bottlenecks of each method and its implementation. 
For example, our method exploits the specific structure in the input while this has not been factored into any potential modification of LKS. 
On the other hand, the version of the LKS method benchmarked here is designed to be Toffoli count efficient at the expense of using a lot of additional ancilla qubits ($\tilde{\mathcal{O}}(\sqrt{N})$), while our method only uses at most $\textrm{poly}\log N$ additional ancilla qubits.

\subsection{A general comparison}\label{subsec:GeneralComparison}
As noted in the introduction, the quantum state preparation problem has been well investigated in the past using various approaches~\cite{grover2002creating, mottonen2004transformation, sanders2019black, garcia2021quantum, mcardle2022quantum}.
The earliest of these methods is due to Grover and Rudolph~\cite{grover2002creating,mcardle2022quantum, kitaev2008wavefunction}, and more recently, approaches based on rejection sampling~\cite{sanders2019black, bausch2022fast}, matrix product states~\cite{garcia2021quantum,holmes2020efficient}, and quantum singular value transformation~\cite{mcardle2022quantum} have been proposed and studied.
These methods have their own assumptions, hence their particular advantages and shortcomings, and the final verdict on which one is the most efficient is highly case dependent.
In this section, we aim to give a general comparison of these methods.\\

We start with the method by Grover and Rudolph. In this approach, the amplitude function $f(x) \in \mathbb{R}$ of the state $\sum^{2^n-1}_{x=0} f(x) \ket{x}$ is processed in such a way that the rotation angles for controlled-rotations are found. This step is performed either during the quantum computation, or classically beforehand and communicated to the quantum computer in some way.
In particular, the angles, $\alpha, \alpha_{0}, \alpha_{1}, \alpha_{00},\ldots, \alpha_{11}, \ldots, \alpha_{0\ldots0}, \alpha_{1\ldots1}$, each labeled by $0\leq k \leq n-1$ bits $i_0, i_1, i_{k-1}$ are given as follows.
The method works by creating the quantum state qubit-by-qubit, using $n$ controlled rotations.
It starts with 
\begin{align}
\nonumber \sum_{x} f(x)\ket{x}=
\sin(\alpha)&\ket{0}\sum^{1}_{x_1,\ldots, x_{n-1}=0} f(x_0=0, x_1, \ldots, x_{n-1})\ket{0, x_1, \ldots, x_{n-1}}\\
&+ \cos(\alpha) \ket{1}\sum^{1}_{x_1,\ldots, x_{n-1}=0} f(x_0=1, x_1, \ldots, x_{n-1})\ket{0, x_1, \ldots, x_{n-1}},
\end{align}
where $\sin \alpha= \arcsin \frac{\sum^{2^{n-1}-1}_{x=0} |f(x)|^2}{\sum^{2^{n}-1}_{x=0} |f(x)|^2}$.
Then, one continues this procedure, and peels off one qubit at a time using controlled-$R_Y$ rotations, and ends up with
\begin{align}
\sum_{x} f(x) \ket{x} = \sum^1_{x_0, \ldots, x_{n-1}=0} c_{x_0} c_{x_0 x_1} \ldots c_{x_0x_1 \ldots x_{n-1}} \ket{x_0, x_1,\ldots x_{n-1}},
\end{align}
where $c_{x_0} = \sin(\alpha)\delta_{x_0,0}+ \cos(\alpha)\delta_{x_0,1}$, $c_{x_0x_1}= \sin(\alpha_{x_0})\delta_{x_1,0} +  \cos(\alpha_{x_0}) \delta_{x_1,1}$, and more generally, $c_{x_0 x_1\ldots x_k}= \sin(\alpha_{x_0\ldots x_{k-1}})\delta_{x_k,0}+ \cos(\alpha_{x_0\ldots x_{k-1}})\delta_{x_k,1}$.
Eventually, at each branch specified by a given $x_0, x_1, \ldots, x_{k-1}$, $c_{x_0 x_1 \ldots x_k}$ is either $\sin$ or $\cos$ of the the angle $\alpha_{x_0 \ldots x_{k-1}}$.
This, hence, requires either the angles $\alpha_{x_0\ldots x_{k-1}}$ to be computed coherently, depending on the qubit registers holding the values $x_0, x_1, \ldots, x_{k-1}$, or computed classically offline and loaded to quantum registers.
Furthermore, the angles depend on the coefficients $f(x)$ in a nontrivial manner, i.e., by ratios of sums of squares of the function among different intervals.
It is usually not a good idea to perform these calculations on the quantum computer, due to the fact that it requires implementing sums, their ratios, and implementing $\arcsin$ coherently.
If performed classically before the quantum computation, this in principle requires calculating integrals and their ratios $\mathcal{O}(2^n)$ times.
While this is in general a hard problem, it may be easy for efficiently integrable functions (e.g., an analytic formula is known or integrals can be performed via other computational methods such as Monte Carlo).
Even then, the values must be loaded to the quantum computer and routed to the relevant branch $x_0x_1 \ldots x_k$, and this takes $\mathcal{O}(2^k)$ gates.
In the original proposal by Grover and Rudolph, it had been imagined that the quantum computer would implement the angle computations coherently, which would include the quantum implementation of the function integration, either analytically or with a Monte Carlo based method, such as in Ref.~\cite{applegate1991sampling}. 
Indeed, there may be cases where computation or sampling of the integrals/sums are far more efficient than computing the functions' pointwise values, in which case the method of Grover-Rudolph would almost certainly be more efficient than other methods.

\noindent Note that above does not exploit any structure of the coefficients and the Grover-Rudolph method can indeed be made far more efficient for structured cases.
This is indeed the case for other quantum state preparation methods as well.
The general pattern is that for the arbitrary state preparation problem costs $\Theta(2^n)$ in the worst case, regardless of the method.
However, each method still has their advantages for different type of cases, and below we discuss this further.
One advantage of our and other relevant methods~\cite{sanders2019black, bausch2022fast} is that they only use the functions $f(x)$ without the need to compute angles.
Instead one finds a reference function that resembles the target function and can be efficiently prepared, and the target state is sampled out of it.
The price paid is the amplitude amplification which leads to repeating the same operations over and over until the target state is reached.
Our method, hence, is expected to give better or at least competitive results, when the number of amplitude amplifications can be made small with a simple reference function.\\

Another method that uses amplitude amplification includes a recent method given in Ref.~\cite{mcardle2022quantum}.
In this method, the specific input/reference state with coefficients $\sin(x/N)$ is modified via a quantum circuit $U_h$ that block-encodes the action of the function $h(y)= f(\arcsin(y))$, which hence modifies the coefficients to those of the desired target state.
This method requires amplitude amplification as well, due to the fact that the target still needs to be sampled because it is in a specific branch of the quantum state after the application of the block-encoding $U_{h}$.
This method relies on having an efficient block-encoding of $U_h$, and the complexity of one round of the circuit depends on the degree and the coefficients of the Chebyshev expansion of $U_h$.
We expect this method to perform well, when the degree and the sum of the expansion coefficients are low. This is more likely to happen for smooth functions and fails to happen for functions with discontinuities in the function or its derivatives.
For example, as also pointed out in Ref.~\cite{mcardle2022quantum}, for the target function $\sqrt{x}$ (or also for $1/x$, etc.), the QSP-based method would require a regularization of the function around $x=0$ due to the discontinuities.
A highly accurate approximation of the state would require a higher degree polynomial, hence may drive the cost of the algorithm higher compared to others. 
QRS-based method does not face this issue.
Note further that a modification of their reference function will also improve the number of amplitude amplification steps, similarly to how it helps in our case.
While it may not perform as competitively as our method in some cases, one other advantage of this method is the elimination of any additional coherent quantum arithmetic (which may be costly in some cases) in favor of the precomputation of rotation angles that are used in the quantum circuit for the block-encoding $U_h$.\\

Another promising method that uses minimal extra ancilla qubits is to use matrix product states (MPS).
This method requires extensive classical preprocessing to find the MPS representation of the desired target state.
After finding an $\epsilon$-MPS approximation of the $n$-qubit target state with some resulting bond dimension $D \in \mathbb{N}^+$, the quantum algorithm costs merely an extra $\log_2 D$ ancilla qubits, and $\mathcal{O}(n D^2)$ gates.
In the worst case, when $D~2^{n/2}$, we fall back to the worst case scaling, but otherwise the quantum algorithm is efficient for small $D$.
For given classes of functions (such as Gaussians, etc.) upper bounds for $D$ can be studied.
However, this method is most powerful when the optimization of local tensors are performed in a classical computer with a tensor network algorithm, with a target accuracy $\epsilon$ and with the goal of finding the minimal bond dimension $D$ across all partitions of the quantum state where the qubits are arranged on a line.
Note that this may even be seen as a generalization of the Grover-Rudolph method, where we are allowed to expand the Hilbert space along the execution of the quantum circuit, which hopefully helps with packing the long-range/large-size multi-controls to short-range/small-size multi-controls.
Given that MPS is universal, and the approximation error is a monotonically decreasing function of the bond dimension $D$, this is a promising generic method to address the quantum state preparation problem.
However, in practice, its implementation requires sophisticated methods especially when high precision is desired for the quantum state preparation problem.\\

We limit our discussion of the aforementioned methods, since a complete comparison requires a detailed analysis of each method, and is beyond the scope of the current work.
While our general comments can be taken as a guideline, ultimately a detailed performance analysis should be carried out in combination with the best classical precomputation methods as needed by the quantum algorithm. In fact, for any given instance of the state preparation problem, optimal algorithmic efficiency will likely result from carefully combining classical and quantum methods.
In the next section, we compare our resource counts to the optimized implementation of the Grover-Rudolph algorithm, where the number of Toffoli gates is minimized at the expense of additional ancilla qubits.

\subsection{Comparison with Low-Kliuchnikov-Schaefer state preparation}

Here we detail the construction and costings of a gate optimized version of Grover-Rudolph state preparation with precomputed angles loaded to the quantum computer via QROMs, which we call ``Low-Kliuchnikov-Schaefer'' state preparation (the ``LKS method''), named after the authors of Ref.~\cite{low2018trading} where this method was first described. 
The method uses QROMs and addition with a Fourier state to implement a cascade of multiplexed rotations described in Ref.~\cite{mottonen2004transformation} to prepare a target aribitrary superposition. 
Asymptotically this leads to an almost quadratic improvement in Toffoli count over Ref.~\cite{mottonen2004transformation} at the expense of using additional ancilla qubits.

We follow the constructions given in Appendix D of Ref.~\cite{low2018trading}.
For an $n$-qubit quantum state, the circuit consists of an initial rotation followed by $n-1$ (multi)-controlled (or multiplexed) rotations.
The $k$-th (for $k= 0,1, \ldots, n-1$) rotation's angle does, in principle, depend on the values registered in the first $k-1$ qubits.
Hence, each multi-controlled rotation consists of a data-lookup QROM, that coherently loads the values of the $2^{k-1}$ angles, applying a $Y$ rotation depending on these angles, and unloading or uncomputing these values coherently.
Also, note that it is sufficient to load each of the angles with $b$-bit precision where $b= \lceil \log_2 (2\pi n/\epsilon) \rceil$ for obtaining the target quantum state with $\epsilon$ approximation in Euclidean norm.
The rotation is realized by a quantum Fourier state and phase gradient addition.
The one time cost of preparing the quantum Fourier state of $b$-qubits is omitted as it has been omitted in the QRS-based method above.
Furthermore, performing the addition with a $b$-bit Fourier state can be done slightly cheaper with $(b-2)$ Toffoli gates as noted in Ref.~\cite{Sanders_2020}.
A substantial part of the cost comes from the data-lookup QROMs, which can be optimized for minimizing the number of Toffoli gates, at the expense of using $b\lambda_k$ ancilla at the $k$-th stage.
More precisely, Ref.~\cite{low2018trading} constructs a quantum circuit that uses an additional $b\lambda_{n-1}$ ancilla, such that the Toffoli cost of the data-lookup QROMs, adder, and the uncomputes of QROMS (which are cheaper than the QROMs) is given by
\begin{align}\label{eq:LKSCost}
\sum^{n-1}_{k=0} \left(\lceil 2^k/\lambda_k \rceil + b(\lambda_k-1) + \lceil 2^k/\lambda'_k \rceil + (\lambda'_k - 1) + (b-2) \right),
\end{align}
where $\lambda_k$ and $\lambda'_k$ are chosen as either the ceil or floor of $\sqrt{2^k/b}$ and $\sqrt{2^k}$, respectively, to optimize the final count.
This analysis has been used in getting the Toffoli counts given Fig.~\ref{fig:LKSTcount}.

\begin{figure}
\begin{subfigure}{.33\textwidth}
  \centering
  \includegraphics[width=.99\linewidth]{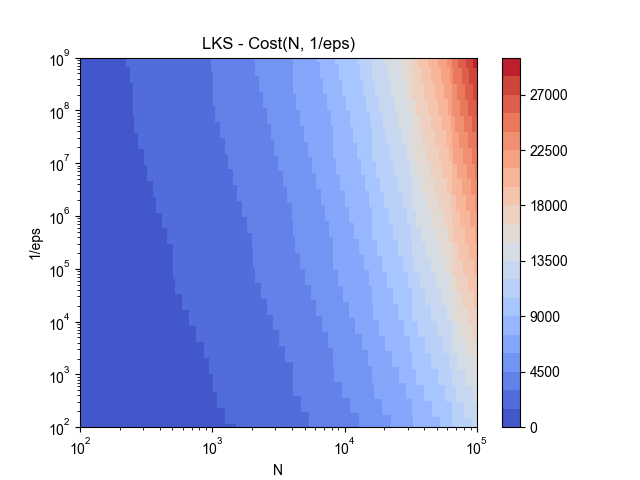}
  \caption{$N 
  \in [10^2, 10^5]$}
  \label{fig:LKSExample2to5}
\end{subfigure}%
\begin{subfigure}{.33\textwidth}
  \centering
  \includegraphics[width=.99\linewidth]{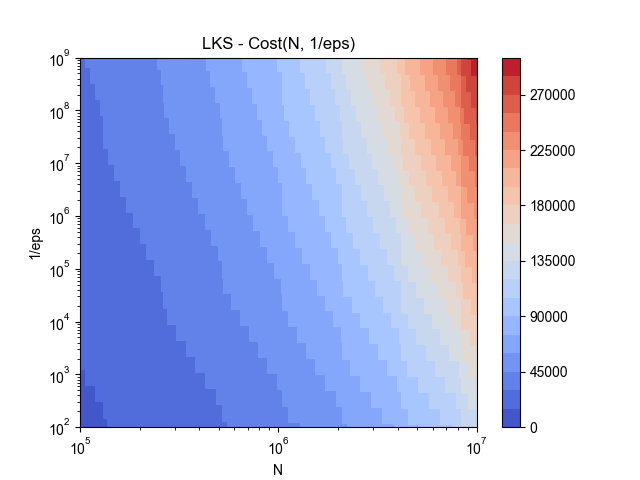}
  \caption{$N 
  \in [10^5, 10^7]$}
  \label{fig:LKSExample5to7}
\end{subfigure}
\begin{subfigure}{.33\textwidth}
  \centering
  \includegraphics[width=.99\linewidth]{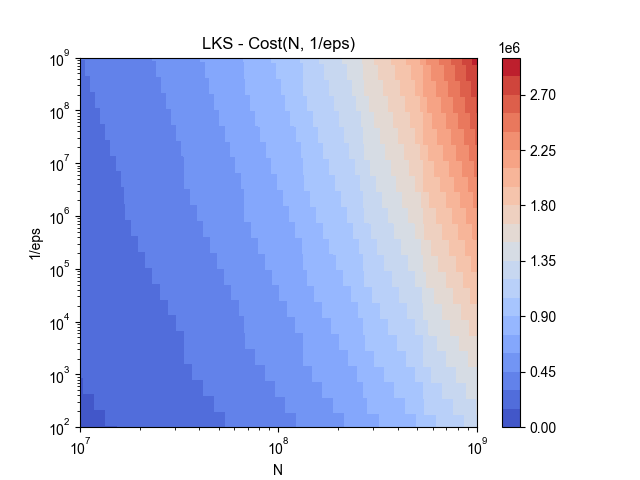}
  \caption{$N 
  \in [10^7, 10^9]$}
  \label{fig:LKSExample7to9}
\end{subfigure}
\caption{Toffoli count of the LKS method derived from Eq.~\eqref{eq:LKSCost} with optimal ancilla usage for minimal number Toffoli gates, in the parameter regime $\epsilon \in [10^{-2}, 10^{-9}]$ and various ranges for $N$ as given in the subcaptions.}
\label{fig:LKSTcount}
\end{figure}

\subsubsection{Inverse in one and three dimensions}

In this section, we compare the gate counts that result from our method with those that result from the LKS methods for preparing the states $\ket{\psi_f}$ with $f(x)= 1/x$ for $x \in [1, N]$, and $f(\v{x})= 1/|\v{x}|$ for $\v{x}= (x_0, x_1, x_2) \in [1,L]^{\times 3}$ where $|\v{x}|:= \sqrt{x_0^2 + x_1^2 + x_2^2}$.
The comparison for $f(x)= 1/x$ is given in Fig.~\ref{fig:InverseLKSComp}, and Figs.~\ref{fig:InverseLKSCompEps3},~\ref{fig:InverseLKSCompEps6},~\ref{fig:InverseLKSCompEps9} correspond to the cases $\epsilon= 10^{-3}, 10^{-6}, 10^{-9}$, respectively.
The comparison for $f(\v{x})= 1/|\v{x}|$ is given in Fig.~\ref{fig:Inverse2LKSComp}, and Figs.~\ref{fig:Inverse2LKSCompEps3},~\ref{fig:Inverse2LKSCompEps6},~\ref{fig:Inverse2LKSCompEps9} correspond to the cases $\epsilon= 10^{-3}, 10^{-6}, 10^{-9}$, respectively.\\

\begin{figure}
\centering
\begin{subfigure}{.3\textwidth}
  \centering
  \includegraphics[width=.99\linewidth]{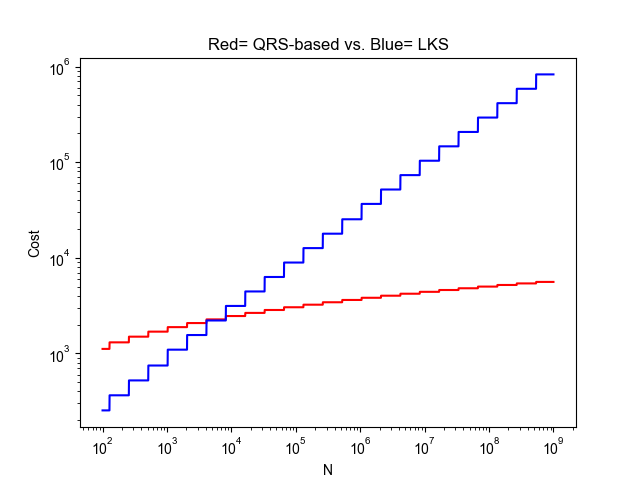}
  \caption{$\epsilon= 10^{-3}$}
  \label{fig:InverseLKSCompEps3}
\end{subfigure}%
\begin{subfigure}{.3\textwidth}
  \centering
  \includegraphics[width=.99\linewidth]{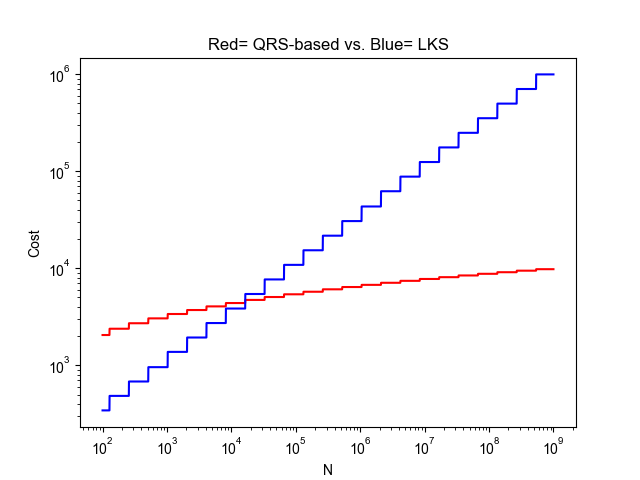}
  \caption{$\epsilon= 10^{-6}$}
  \label{fig:InverseLKSCompEps6}
\end{subfigure}
\begin{subfigure}{.3\textwidth}
  \centering
  \includegraphics[width=.99\linewidth]{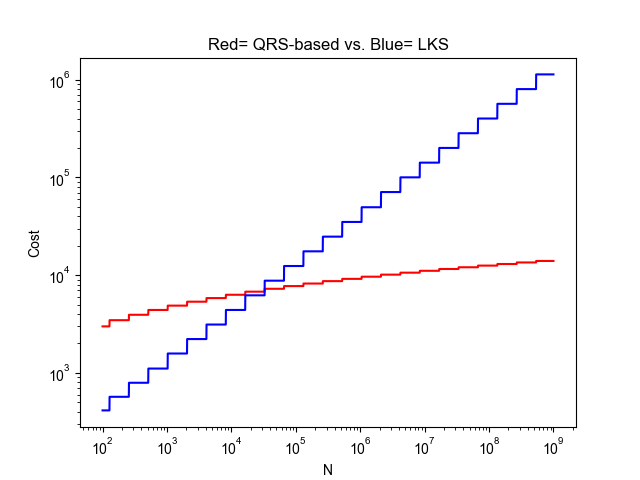}
  \caption{$\epsilon= 10^{-9}$}
  \label{fig:InverseLKSCompEps9}
\end{subfigure}
\caption{Comparison between our QRS-based method (in red) and the LKS method (in blue) for preparing the state $\propto \sum^N_{x=1}1/x \ket{x}$ over $N \in [10^2,10^9]$, with errors $\epsilon \in \{10^{-3}, 10^{-6}, 10^{-9}\}$.}
\label{fig:InverseLKSComp}
\end{figure}

\begin{figure}\centering
\begin{subfigure}{.3\textwidth}
  \centering
  \includegraphics[width=.99\linewidth]{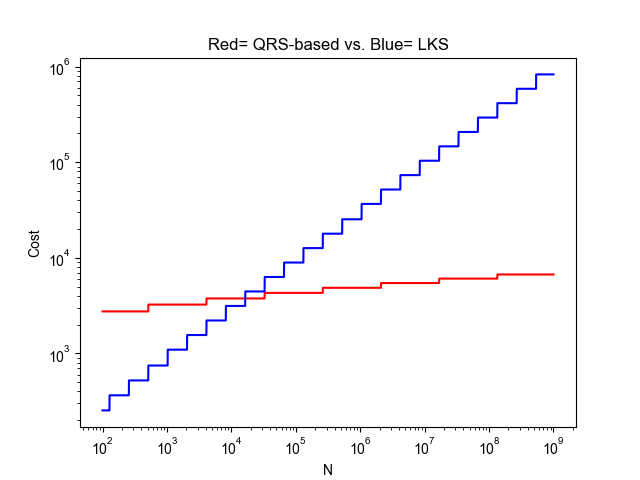}
  \caption{$\epsilon= 10^{-3}$}
  \label{fig:Inverse2LKSCompEps3}
\end{subfigure}%
\begin{subfigure}{.3\textwidth}
  \centering
  \includegraphics[width=.99\linewidth]{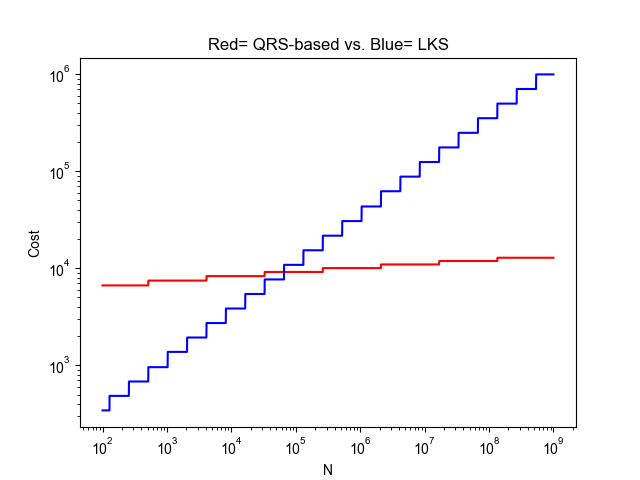}
  \caption{$\epsilon= 10^{-6}$}
  \label{fig:Inverse2LKSCompEps6}
\end{subfigure}
\begin{subfigure}{.3\textwidth}
  \centering
  \includegraphics[width=.99\linewidth]{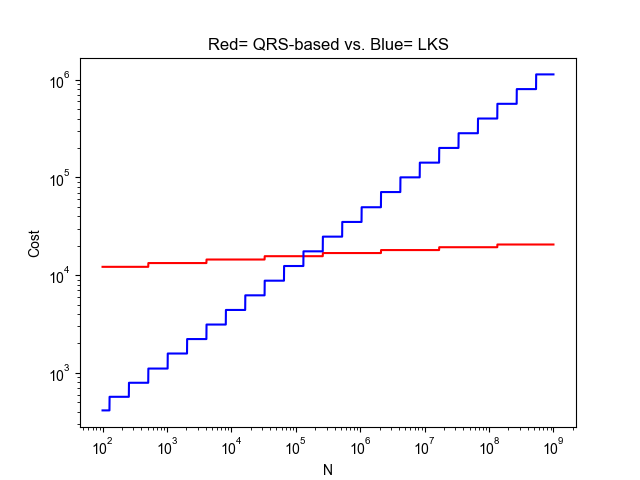}
  \caption{$\epsilon= 10^{-9}$}
  \label{fig:Inverse2LKSCompEps9}
\end{subfigure}
\caption{Comparison between our method (in red) and the LKS method (in blue) for preparing the state $\propto \sum_{\v{x}}1/|\v{x}| \ket{\v{x}}$ over $N \in [10^2,10^9]$, with errors $\epsilon \in \{10^{-3}, 10^{-6}, 10^{-9}\}$.}
\label{fig:Inverse2LKSComp}
\end{figure}

These results instruct us to use our method rather than the LKS method, given that over these large and relevant parameter regimes our method performs better.
In fact, our method will always perform better for larger points $N$.
Some comments are in order for explaining why this is the case.
Note that both of these cases bypass evaluating operations like inverse or inverse-square-root, by modifying the comparator.
For example, for the case $f(x)= 1/x$, instead of setting up the comparator as $M/x \geq m 2^{-(\mu-1)}$ (where $\mu$ has an $x$-dependent value and gives the reference state in ziggurat form), we set it up as $M 2^{(\mu-1)} \geq m x$.
In this way, we avoid calculating the inverse of $x$.
Note that $x$ is readily available in the system register, hence no further computation than simple multiplications are needed for the above comparator.
Similarly, for the case $f(\v{x})= 1/|\v{x}|$, instead of setting up the comparator as $M/|\v{x}| \geq m 2^{-(\mu-1)}$ (where $\mu$ has an $\v{x}$-dependent value and gives the reference state in a three-dimensional ziggurat form), we set it up as $M^2 2^{2(\mu-1)} \geq m^2 (x_0^2 + x_1^2 + x_2^2)$.
In this way, we avoid calculating inverse of $|\v{x}|$, and furthermore any computation of square roots.
Note that, $x_0, x_1, x_2$ are readily available in the system register, hence we need a few multiplications and addition for implementing the given comparator.

For the one dimensional inverse as the target function, and other computational details given as in Sec.~\ref{subsec:ExamplesPowerLaw}, the threshold values $N^*$ where QRS-based method outperforms the LKS method are $\approx 10^{4}$,  $\approx 1.6 \times 10^4$ and $\approx 3.3 \times 10^4$, respectively, for $\epsilon= 10^{-3}$, $\epsilon= 10^{-6}$, $\epsilon= 10^{-9}$.
For the three dimensional inverse as the target function, the threshold values $N^*$ where QRS-based method outperforms the LKS method are $\approx 1.6 \times 10^{4}$,  $\approx 6.5 \times 10^4$ and $\approx 1.3 \times 10^5$, respectively, for $\epsilon= 10^{-3}$, $\epsilon= 10^{-6}$, $\epsilon= 10^{-9}$.

\subsubsection{Gaussian}

We compare the Toffoli counts that result from our method with those from LKS method for preparing the Gaussian with $\sigma= 2^3$ and $\mu=0$ within the interval $x \in [-0.5,0.5]$, for $\epsilon\in \{10^{-3}, 10^{-6}, 10^{-9}\}$ and for $N \in [10^2, 10^9]$.
The comparison is given in Fig.~\ref{fig:GaussianLKSComp}, Figs.~\ref{fig:GaussianLKSCompEps3},~\ref{fig:GaussianLKSCompEps6}, ~\ref{fig:GaussianLKSCompEps9} correspond to the cases $\epsilon= 10^{-3}, 10^{-6}, 10^{-9}$, respectively.
As the figure indicates, our QRS-based method has more advantage for higher values of $N$, which is not surprising given the asymptotic scaling of $\mathcal{O}(\log(N))$ of our method versus $\mathcal{O}(\sqrt{N})$ of the LKS method.

What is perhaps more surprising is how the threshold value of $N$, the point where our method becomes less costly than the LKS method, gets pushed to higher values, when more precision, i.e., a smaller $\epsilon$, is demanded.
This is also understandable when considering the asymptotic scaling of the resource estimates.
Computing functions like a Gaussian, via Taylor series, requires higher order of series elements, which implies as many multiplications as the order of the series.
Similarly, using a combination of arithmetic and data-lookup tables comes with an upfront cost for performing the arithmetic.
More fundamentally, the dependence on the accuracy $\epsilon$ of the two methods are different.
Computing the function values coherently scales as $\textrm{poly}(\log(1/\epsilon))$, while the asymptotic scaling of the cost of the LKS method is only a polynomial in $\mathcal{O}(\log\log(1/\epsilon))$, up to further factors.
While this is one fundamental reason that one needs high enough $N$ to see the benefits of the QRS-based method, there are a few ways that it can be further improved by almost an order of magnitude.
First, we could use a better reference function $g$ that would yield a much smaller value for $M$.
This would lead to a lower bit-precision in calculations, hence,
as it can be seen in Tab.~\ref{table:Gaussian_REs}, would lead to a good amount of improvement in terms that depends quadratically on $b$.
Second, we do not claim the method to compute the Gaussian above is optimal, and we do not preclude more efficient constructions. 
Furthermore, the optimized value of $k$ for smaller values of $\epsilon$ or $N$ are different and would result in lower gate counts.

\begin{figure}
\centering
\begin{subfigure}{.3\textwidth}
  \centering
  \includegraphics[width=.99\linewidth]{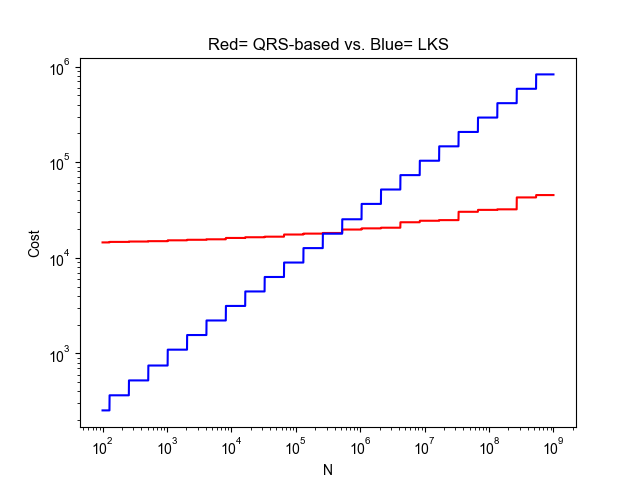}
  \caption{$\epsilon= 10^{-3}$}
  \label{fig:GaussianLKSCompEps3}
\end{subfigure}%
\begin{subfigure}{.3\textwidth}
  \centering
  \includegraphics[width=.99\linewidth]{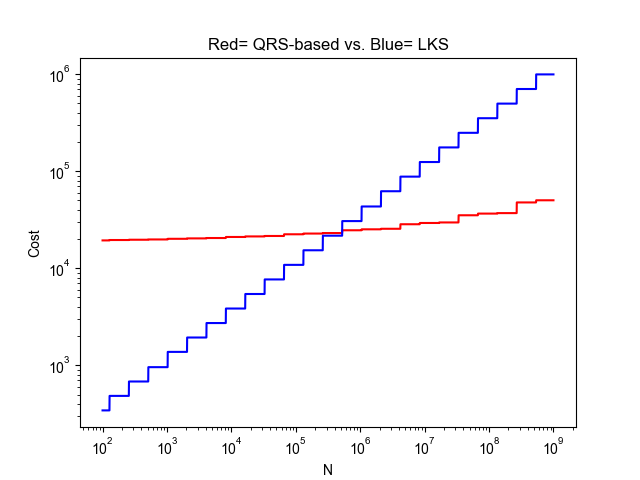}
  \caption{$\epsilon= 10^{-6}$}
  \label{fig:GaussianLKSCompEps6}
\end{subfigure}
\begin{subfigure}{.3\textwidth}
  \centering
  \includegraphics[width=.99\linewidth]{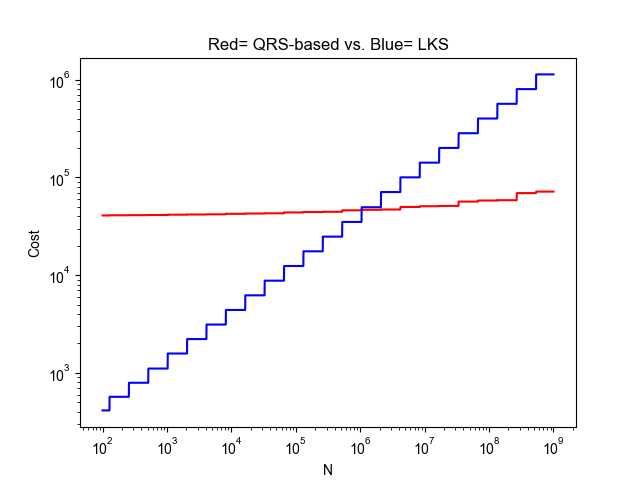}
  \caption{$\epsilon= 10^{-9}$}
  \label{fig:GaussianLKSCompEps9}
\end{subfigure}
\caption{Comparison between our QRS-based method (in red) and the LKS method (in blue) for preparing the Gaussian state over $N \in [10^2,10^9]$ equally spaced point within the domain $x \in [-0.5, 0.5]$, with errors $\epsilon \in \{10^{-3}, 10^{-6}, 10^{-9}\}$.}
\label{fig:GaussianLKSComp}
\end{figure}

While these improvements are expected to give lower gate counts for the QRS-based method, we again note the general trend of increasing threshold value $N^*$ with higher precision.
For the given reference function, and other computational details given as in Sec.~\ref{subsec:ExamplesGaussian}, $N^*$ becomes as high as  $\approx 10^{6}$ for $\epsilon = 10^{-9}$.
While there is perhaps little reason to use the QRS-based method with this particular reference function as long as $N < 10^6$, it is dramatically more beneficial to use it for $N > 10^6$ as seen in Fig.~\ref{fig:GaussianLKSCompEps9}.
The threshold values $N^*$ are both $\approx 5.2 \times 10^6$ for $\epsilon= 10^{-3}$ and $\epsilon= 10^{-6}$.
 
\subsubsection{Hyperbolic tangent}

The trends and the results are similar to those for the Gaussian state preparation.
We compare the Toffoli counts that result from our method with those from LKS method for preparing the hyperbolic tangent state in the interval $x \in [-0.5,0.5]$, for $\epsilon\in \{10^{-3}, 10^{-6}, 10^{-9}\}$ and for $N \in [10^2, 10^9]$.
The comparison is given in Fig.~\ref{fig:tanhLKSComp}, Figs.~\ref{fig:tanhLKSCompEps3},~\ref{fig:tanhLKSCompEps6},~\ref{fig:tanhLKSCompEps9} correspond to the cases $\epsilon= 10^{-3}, 10^{-6}, 10^{-9}$, respectively.
As the figure indicates, our QRS-based method has more advantage for higher values of $N$, which is not surprising given the asymptotic scaling of $\mathcal{O}(\log(N))$ of our method versus $\mathcal{O}(\sqrt{N})$ of the LKS method.

For the given reference function, and other computational details given as in Sec.~\ref{subsec:ExamplesTanh}, $N^*$ becomes as high as  $\approx 5.2 \times 10^{5}$ for $\epsilon = 10^{-9}$.
While there is perhaps little reason to use the QRS-based method with this particular reference function as long as $N < 5.2 \times 10^5$, it is dramatically more beneficial to use it for $N > 5.2 \times 10^5$ as seen in Fig.~\ref{fig:tanhLKSCompEps9}.
The threshold values $N^*$ are $\approx 1.3 \times 10^5$ for $\epsilon= 10^{-3}$ and $\epsilon= 10^{-6}$.

\begin{figure}
\centering
\begin{subfigure}{.3\textwidth}
  \centering
  \includegraphics[width=.99\linewidth]{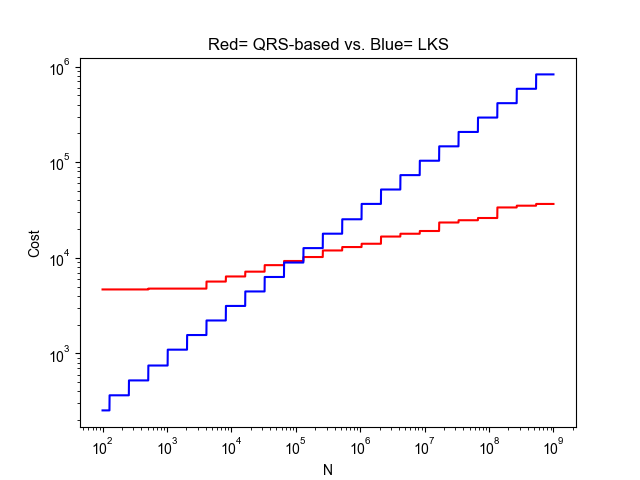}
  \caption{$\epsilon= 10^{-3}$}
  \label{fig:tanhLKSCompEps3}
\end{subfigure}%
\begin{subfigure}{.3\textwidth}
  \centering
  \includegraphics[width=.99\linewidth]{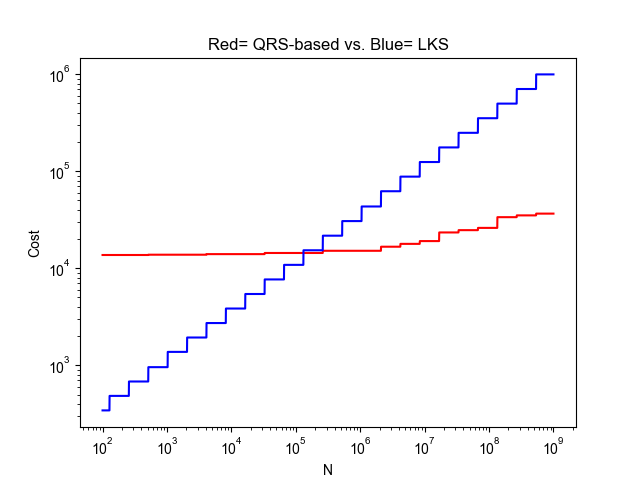}
  \caption{$\epsilon= 10^{-6}$}
  \label{fig:tanhLKSCompEps6}
\end{subfigure}
\begin{subfigure}{.3\textwidth}
  \centering
  \includegraphics[width=.99\linewidth]{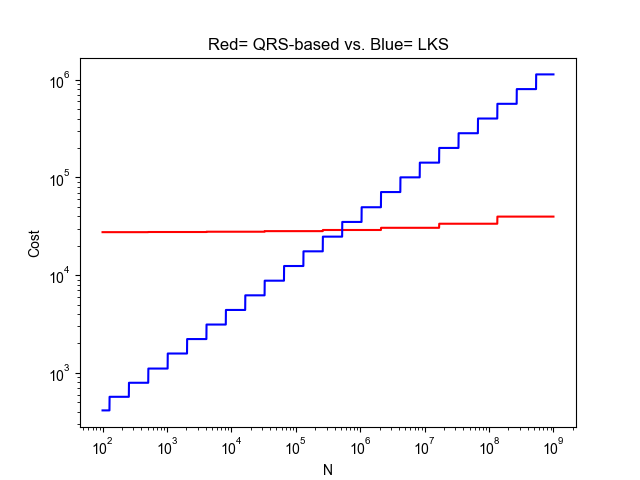}
  \caption{$\epsilon= 10^{-9}$}
  \label{fig:tanhLKSCompEps9}
\end{subfigure}
\caption{Comparison between our QRS-based method (in red) and the LKS method (in blue) for preparing the hyperbolic tangent state over $N \in [10^2,10^9]$ equally spaced point within the domain $x \in [-0.5, 0.5]$, with errors $\epsilon \in \{10^{-3}, 10^{-6}, 10^{-9}\}$.}
\label{fig:tanhLKSComp}
\end{figure}

\section{A general quantum sampling method for block-encodings}\label{sec:QSamplingUnitaries}

The quantum rejection sampling methods applied to the quantum state preparation for Problem~\ref{problem:QuantumStatePreparation} can be extended to a solution to the matrix block-encoding (Problem~\ref{problem:MatrixBlockEncoding}).
In particular, we introduce a \emph{reference matrix}~$G$, whose elements $G_{ij}$ bound the elements $A_{ij}$ of the target matrix $A$. 
The reference matrix $G$ is loaded via two state preparation subroutines, $\PREP_{\chi}$ and $\PREP_{\phi}^\dagger$, which are implicitly defined via a class of matrix decompositions for $G$. This allows us to unify 
and generalize, via coherent sampling techniques, the main classes of block-encoding that appear in the literature, which include several different access models to the data encoded by the target matrix $A$.
We first start with the definitions that introduce the problem and the subroutines used, and then introduce Alg.~\ref{algo:MBEA} with its proof.
We then consider various access model instantiations of Alg.~\ref{algo:MBEA}, which recover as special cases and extend different types of block-encodings, such as $\PREP^\dagger-\SEL-\PREP$ (Sec.~\ref{sec:PREPSELPREP}), submatrix partitioning (Sec.~\ref{sec:MatrixZiggurat}), row-column (Sec.~\ref{sec:RowColumnBE}), and column (Sec~\ref{sec:ColumnBE}) block-encodings.

\subsection{Definitions and assumptions}\label{subsec:DefinitionsAndAssumptionsBlockEncoding}

\begin{problem}[Matrix block-encoding]\label{problem:MatrixBlockEncoding}
Let $\{A_{ij} \in \mathbb C\}^{N}_{i,j=1}$ be a family sets of complex numbers with increasing $N \in \mathbb{N}^+$, and let $\epsilon \in \mathbb{R}^+$ be a given error parameter such that $0 < \epsilon < 1$.
The matrix block-encoding problem is the problem of producing a unitary $U$ that embeds an $\epsilon$-approximation to the $\alpha$-rescaled version of the matrix $A:= \sum^{N}_{i,j=1} A_{ij} \ket{i}\bra{j}$, i.e., for every $N$, 
\begin{align} 
\|\bra{0}U_{A/\alpha}\ket{0}_a - A/\alpha\| \leq \epsilon/\alpha,
\end{align}
where $\|\cdot\|$ is the operator norm, $0< \epsilon < 1$, $\ket{0}_a$ identifies the subspace of the ancilla system where the approximation of $A$ is embedded, and $\alpha \in \mathbb{R}^+$ is the rescaling factor.
\end{problem}
Ideally, for fixed $N$ the goal is to find the circuit with minimal gates.
In general, though, we design the algorithm with a family of problems in mind, and the the regime of interest is for large $N$.

Most of the subroutines we use for this problem are, in essence, the same as the subroutines given in Section~\ref{subsec:DefinitionsAndAssumptionsStatePrep} used for the quantum state preparation problem.
We additionally have $\SWAP$ and the $\PHASE_\varphi$ gates acting on two registers (a system register holding the column index of the matrix, and a copy of the system register holding the row index of the matrix).
Furthermore, instead of having one reference state preparation, there are two families of them, associated to coefficients $\chi_{ij}$ and $\phi_{ij}$. 
When combined in our algorithm, they effectively give rise, by Hadamard  (i.e. element-by-element) product, to a reference matrix $G$ which upper bounds $|A_{ij}|$:
\begin{align}
   \frac{\chi_{ij}}{\mathcal{N}_{\chi_j}} \frac{\phi_{ji}}{\mathcal{N}_{\phi_i}} \propto G_{ij} \geq |A_{ij}|,
\end{align}
where $\mathcal{N}_{\chi_j} = \sqrt{ \sum_i |\chi_{ij}|^2}$, $\mathcal{N}_{\phi_i} = \sqrt{\sum_j |\phi_{ji}|^2}$.
A slight generalization then leads to the introduction of an extra index $k$ and to consider decompositions of $G$ as given in the following definition.

\begin{definition}[$\PREP_{\chi}$, $\PREP_{\phi}$, $U_A$, $U_G$ and their relation to matrix $G$]\label{def:PREPChiPREPPhiRelationToG}
Let $A$ be a given target matrix to be block-encoded, and let $G$ be a reference matrix, i.e., a matrix with nonnegative entries $G_{ij}$ such that $G_{ij} \geq |A_{ij}|$ for all $i,j \in \mathbb{N}^+$ labeling the row and column indices.
Let $\nu, N \in \mathbb{N}^+$, $\chi: [1,\nu] \times [1,N] \times [1,N] \rightarrow \mathbb{R}$ with $\chi: (k,i,j) \mapsto \chi_{kij}$ and $\phi: [1,\nu] \times [1,N] \times [1,N] \rightarrow \mathbb{R}$ with $\phi: (k,i,j) \mapsto \phi_{kij}$, be real-valued non-negative functions, such that
\begin{align}\label{eq:coreBErelation}
\sum^{\nu}_{k=1} \frac{\chi_{kij} \phi_{kji}}{\mcalN_{\chi_j}\mcalN_{\phi_i}}= \frac{G_{ij}}{\alpha},
\end{align}
for all $i,j \in [1,N]$, where $\mathcal{N}^2_{\chi_j} = \sum_{i,k} |\chi_{kij}|^2$, $\mathcal{N}^2_{\phi_i} = \sum_{j,k} |\phi_{kji}|^2$. 
$\PREP_\chi$ and $\PREP_\phi$ are defined as unitary circuits that act as

\begin{align}\label{def:prepchiprepphi}
\PREP_\chi\ket{0}\ket{0}  \ket{j} &= \frac{1}{\mcalN_{\chi_j}}\sum_{k=1}^\nu \sum^{N}_{i=1} \chi_{kij} \ket{k} \ket{i} \ket{j}:= \ket{\chi_j} \ket{j} , \\
\PREP_{\phi}  \ket{0}\ket{0} \ket{i}&= \frac{1}{\mcalN_{\phi_i}}\sum^{\nu}_{k=1} \sum^{ N}_{j=1} \phi_{kji} \ket{k} \ket{j} \ket{i} := \ket{\phi_i} \ket{i},
\end{align}
where the first, second, and third registers are of dimension $\nu$, $N$, and $N$, respectively.

Finally we denote by $U_X$ with $X=A,G$ the unitary
\begin{align}
    U_{X} \ket{i} \ket{j} \ket{0} = \ket{i} \ket{j} \ket{ X_{ij}},
\end{align}
for all $i,j \in [1,N]$, where the third register has $b$ qubits, so that the values of $X$ are computed to finite precision.
\end{definition}

Note that this definition suggests further extensions of the techniques presented in this work could be obtained from tensor decompositions of the matrix $G$, but we will not pursue this further here.
Finally, we define the standard operations $\SWAP$ and $\PHASE_\varphi$, $\USP_M$,  $\Comp$ as follows:

\begin{definition}[$\SWAP$ and $\PHASE_\varphi$, $\USP_M$,  $\Comp$  ]\label{def:SWAPandPHASE}
Let $\mathcal{H}= \mathbb{C}^N \otimes \mathbb{C}^N$.
Then $\SWAP: \mathcal{H} \rightarrow \mathcal{H}$ is defined as $\SWAP: \ket{i} \ket{j} \mapsto \ket{j} \ket{i}$ for all $\ket{i}, \ket{j} \in \mathbb{C}^N$. Given a function $\varphi: [1,N] \times [1,N] \rightarrow [0,2\pi)$,  $\PHASE_\varphi: \ket{i} \ket{j} \mapsto  e^{i\varphi(i,j)} \ket{i} \ket{j}$. $\USP_M$ is defined for $M \in \mathbb{N}^+$ by $\USP_M \ket{0} = \frac{1}{\sqrt{M}} \sum_{m=1}^M \ket{m}$. The unitary $\Comp$ coherently flags with $|0\rangle$ the sampling space (i.e., those $m$ for which $ M |A_{ij}| \geq G_{ij} m$) and coherently flags with $\ket{1}$ the rest (i.e., those $m$ for which $M |A_{ij}| < G_{ij} m$):
\begin{align}
    \Comp \ket{m} \ket{|A_{ij}|} \ket{G_{ij}} \ket{0} = \begin{cases}
    \ket{m} \ket{|A_{ij}|} \ket{G_{ij}} \ket{0}, \quad \textrm{if } M |A_{ij}| \geq G_{ij} m, \\
     \ket{m} \ket{|A_{ij}|} \ket{G_{ij}} \ket{1}, \quad \textrm{otherwise. }
    \end{cases}
\end{align}
\end{definition}

\noindent With these definitions, we are ready to present the main block-encoding algorithm.

\subsection{The algorithm and its cost in terms of subroutines}

We first introduce algorithm~\ref{algo:MBEA} below and the quantum circuit in Fig.~\ref{fig:GeneralPurposeCircuitBE}. We then prove that it works and lay out its cost in Thm.~\ref{thm:mainBE}.

\begin{breakablealgorithm}
\caption{Matrix block-encoding algorithm - Quantum circuit given in Fig.~\ref{fig:GeneralPurposeCircuitBE}}\label{algo:MBEA}
\begin{algorithmic}[1]
\Require A given target matrix $A_{ij} = |A_{ij}| e^{i \varphi(i,j)}$ to be block-encoded according to Problem~\ref{problem:MatrixBlockEncoding}. A reference matrix $G$ such that $G_{ij} \geq |A_{ij}|$ for all $i,j \in \mathbb{N}^+$. Real-valued non-negative functions $\chi: [1,\nu] \times [1,N] \times [1,N] \rightarrow \mathbb{R}$, $\phi: [1,\nu] \times [1,N] \times [1,N] \rightarrow \mathbb{R}$, such that
$\sum^{\nu}_{k=1} \frac{\chi_{kij} \phi_{kji}}{\mcalN_{\chi_j}\mcalN_{\phi_i}}= \frac{G_{ij}}{\alpha}$, where $\mathcal{N}^2_{\chi_j} = \sum_{i,k} |\chi_{kij}|^2$, $\mathcal{N}^2_{\phi_i} = \sum_{j,k} |\phi_{kji}|^2$.  $\PREP_{\chi}$, $\PREP_{\phi}$, $U_A$, $U_G$, $\SWAP$, $\PHASE_\varphi$, $\USP_M$,  $\Comp$  unitary circuits as in Definitions~\ref{def:PREPChiPREPPhiRelationToG}-\ref{def:SWAPandPHASE}.
Dimension of the sampling space $M \in \mathbb{N}^+$ for a desired approximation to the target matrix.

\setcounter{ALG@line}{0}

\Ensure Block-encoding of matrix $A$ with a rescaling factor $\alpha$.
\vspace{0.2 cm}
\State Act on the ancilla $\ket{0}\ldots \ket{0}$ and system state  $\ket{\Psi}$ by the unitaries $\PREP_\chi$ and $\USP_M$.

\State Compute the values $G_{ij}$ and $|A_{ij}|$ to ancilla registers by acting with $U_G$ and $U_A$.

\State Apply $\Comp$ and $\PHASE_\varphi$.

\State Uncompute the values of $|A_{ij}|$ and $G_{ij}$ in the registers, using $U^\dagger_A$ and $U^\dagger_G$.

\State Swap the last two registers via $\SWAP$.

\State Apply $\USP^\dagger_M$ and $\PREP^\dagger_{\phi}$.
\end{algorithmic}
\end{breakablealgorithm}

\begin{figure}
    \centering
    \includegraphics[scale=0.60]{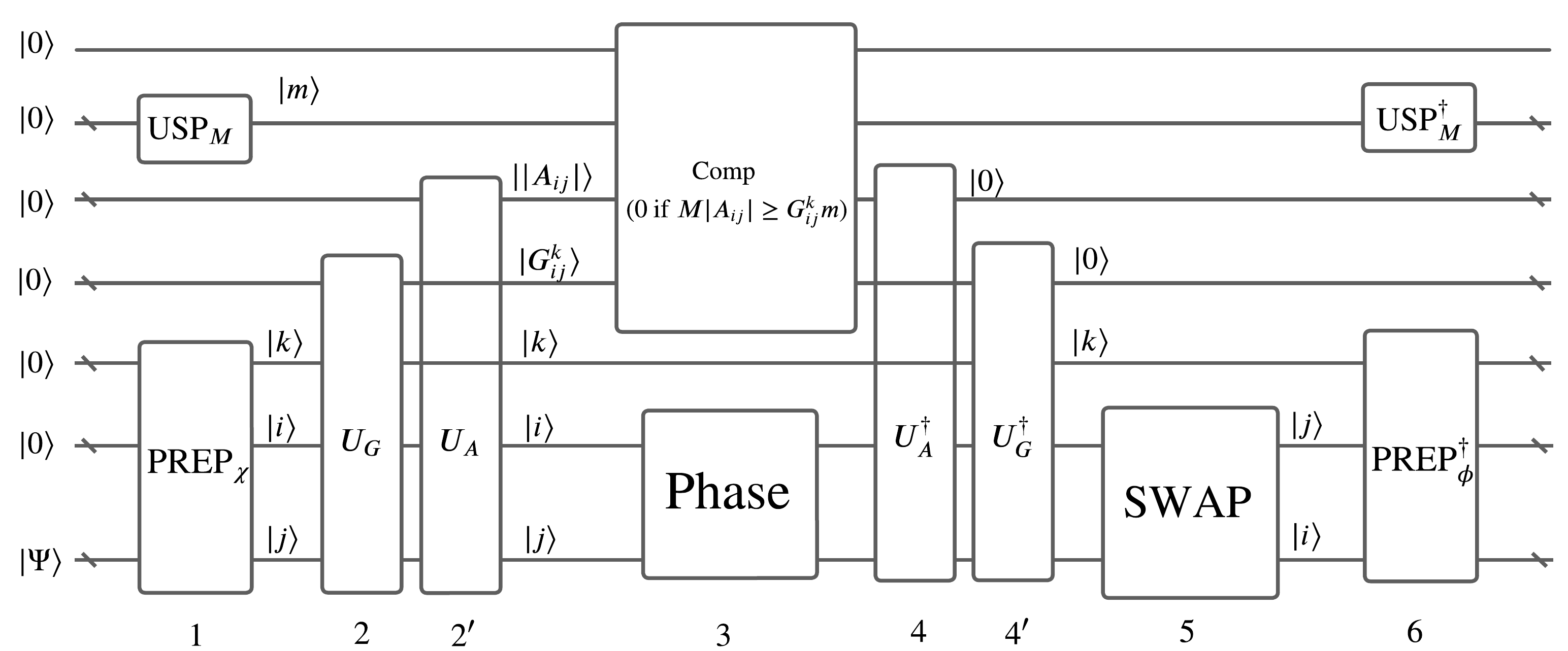}
    \caption{The quantum circuit for the general purpose block-encoding. ${\PREP}_\chi$, ${\PREP}_\phi$, $U_{G}$, $U_{A}$, $\Comp$, $\USP_M$, $\SWAP$, and $\PHASE_\varphi$ are subroutines that are used, and defined separately in Definitions~\ref{def:PREPChiPREPPhiRelationToG},~\ref{def:SWAPandPHASE}.
    $G_{ij}$ is such that $G_{ij}= \sum^\nu_{k=1} \phi_{kji} \chi_{kij}$, and $M \in {\mathbb{N}}^+$ is chosen such that the result of the comparator lead to an approximation to the target function up to a given desired accuracy.}
    \label{fig:GeneralPurposeCircuitBE}
\end{figure}

\begin{theorem}[Quantum sampling algorithm for block-encoding]\label{thm:mainBE}
Let $A: [1,N] \times [1,N] \mapsto \mathbb{C}$ be a given complex function that can be written as $A_{ij}= |A_{ij}|e^{\varphi(i,j)}$ with real-valued function $\varphi$. Let $\chi, \phi : [1,\nu] \times [1,N] \times [1,N] \rightarrow \mathbb{R}$ and $G: [1,N] \times [1,N] \rightarrow \mathbb{R}^+$ be functions such that $G_{ij} \geq |A_{ij}|$ for all $i,j \in [1,N]$, and 
\begin{align}\label{eq:coreBEInTheorem}
\sum^{\nu}_{k=1} \frac{\chi_{kij} \phi_{kji}}{\mcalN_{\chi_j}\mcalN_{\phi_i}}= \frac{G_{ij}}{\alpha},
\end{align}
for all $i,j \in [1,N]$, where $\mathcal{N}^2_{\chi_j} = \sum_{i,k} |\chi_{kij}|^2$, $\mathcal{N}^2_{\phi_i} = \sum_{j,k} |\phi_{kji}|^2$.
Then Algorithm~\ref{algo:MBEA} block-encodes the matrix $A/\alpha$ with rescaling factor $\alpha$ implicitly defined by Eq.~\eqref{eq:coreBEInTheorem}.
\end{theorem}

\begin{proof}
We go through the steps of the algorithm by explicitly writing the output resulting from each stage of the circuit in Fig.~\ref{fig:GeneralPurposeCircuitBE}.
\begin{enumerate}
\item Act on the ancilla and system state $\ket{\Psi}= \sum^N_{j=1} \psi_j \ket{j}$, by $\PREP_\chi$ and $\USP_M$ which results in
\begin{align}
\ket{\phi_1}=\sum^{N}_{i,j}\sum^{\nu}_{k=1} \frac{\chi_{kij}}{\mathcal{N}_{\chi_j}} \psi_j |0\rangle \frac{1}{\sqrt{M}}\sum^M_{m=1}  |m\rangle |0\rangle |0\rangle |k\rangle |i\rangle |j\rangle.
\end{align}

\item Compute the values $G_{ij}$ and $|A_{ij}|$ to ancilla registers by acting with $U_G$ and $U_A$, and obtain
\begin{align}
\ket{\phi_{2'}}=  \sum^{N}_{i,j}\sum^{\nu}_{k=1} \frac{\chi_{kij}}{\mathcal{N}_{\chi_j}} \psi_j |0\rangle \frac{1}{\sqrt{M}}\sum^M_{m=1}  |m\rangle ||A_{ij}|\rangle |G_{ij}\rangle |k\rangle |i\rangle |j\rangle.
\end{align}

\item Coherently flag $m$ with $\Comp$ and apply $\PHASE_\varphi$, obtaining the state
\begin{align}
|\phi_3\rangle = \sum^{N}_{i,j}\sum^{\nu}_{k=1} e^{i\varphi(i,j)} \frac{\chi_{kij}}{\mathcal{N}_{\chi_j}} \psi_j |0\rangle \frac{1}{\sqrt{M}}\sum^{\left\lfloor\frac{M |A_{ij}|}{G_{ij}} \right\rfloor}_{m=1}  |m\rangle ||A_{ij}|\rangle |G_{ij}\rangle |k\rangle |i\rangle |j\rangle + |1\rangle |\ldots\rangle.
\end{align}

\item Uncompute the values of $|A_{ij}|$ and $G_{ij}$ in the registers, using $U^\dagger_A$ and $U^\dagger_G$, to obtain
\begin{align}
|\phi_{4'}\rangle  = \sum^{N}_{i,j}\sum^{\nu}_{k=1} e^{i\varphi(i,j)} \frac{\chi_{kij}}{\mathcal{N}_{\chi_j}} \psi_j |0\rangle \frac{1}{\sqrt{M}}\sum^{\left\lfloor\frac{M |A_{ij}|}{G_{ij}} \right\rfloor}_{m=1}  |m\rangle |0\rangle |0\rangle |k\rangle |i\rangle |j\rangle + |1\rangle |\ldots\rangle.
\end{align}

\item Swap the last two registers via $\SWAP$ to obtain
\begin{align}
|\phi_{5}\rangle  = \sum^{N}_{i,j}\sum^{\nu}_{k=1} e^{i\varphi(i,j)} \frac{\chi_{kij}}{\mathcal{N}_{\chi_j}} \psi_j |0\rangle \frac{1}{\sqrt{M}}\sum^{\left\lfloor\frac{M |A_{ij}|}{G_{ij}} \right\rfloor}_{m=1}  |m\rangle |0\rangle |0\rangle |k\rangle |j\rangle |i\rangle + |1\rangle |\ldots\rangle.
\end{align}

\item Apply $\USP_M^\dagger$ and $\PREP^\dagger_{\phi}$ to obtain
\begin{align}
|\phi_{6}\rangle &= \sum^{N}_{i, j}\sum^{\nu}_{k=1} e^{i\varphi(i,j)} \frac{\phi_{kji} \chi_{kij}}{\mathcal{N}_{\phi_i} \mathcal{N}_{\chi_j}} \psi_j \frac{|A_{ij}|}{G_{ij}} |0\rangle  |0\rangle |0\rangle |0\rangle |0\rangle |0\rangle |i\rangle + |1\rangle |0^\perp\rangle \ket{0} \ket{0} |\ldots\rangle\\
&= \frac{1}{\alpha}\sum^{N}_{i,j} A_{ij} \psi_j |0\rangle   |0\rangle |0\rangle |0\rangle |0\rangle |0\rangle |i\rangle + |1\rangle |0^\perp\rangle |0\rangle |0\rangle |\ldots\rangle\\
&= \ket{0}\ket{0}\ket{0}\ket{0} \ket{0}\ket{0} \frac{A}{\alpha} \ket{\Psi}+ \ket{1}\ket{0^\perp} \ket{0}\ket{0} |\ldots\rangle. 
\end{align}
The resulting state after conditioning on the first two registers being in $\ket{0}$ state shows that the block-encoding of $A$ with the claimed rescaling factor is achieved.
Note that $M$ is chosen sufficiently large along with computations performed sufficiently precisely, as in Lem.~\ref{lem:ChoosingM}, such that the block-encoding is achieved with the desired accuracy.
\end{enumerate}
\end{proof}
\begin{remark}[Variants of the circuit with/without explicit sampling]\label{rem:SamplingOrNot}
We remark that the general quantum circuit given in Fig.~\ref{fig:GeneralPurposeCircuitBE} covers a special case where the sampling is not performed explicitly.
For example, if $\chi$ and $\phi$ are chosen such that
\begin{align}
\label{eq:perfectmatrixreference}
\sum^{\nu}_{k=1} \frac{\chi_{kij} \phi_{kji}}{\mcalN_{\chi_j}\mcalN_{\phi_i}}= \frac{A_{ij}}{\alpha}
\end{align}
there is no need for sampling, i.e., the comparator, function evaluations $U_G$, $U_A$, and the sampling space with relevant operations $\USP_M$, $\Comp$ can be eliminated.
This, however, in most cases, implies costlier subroutines $\PREP_\chi$ and $\PREP_\phi$, despite a lower rescaling factor $\alpha$.
In fact, one way to achieve the state preparations required in Eq.~\eqref{eq:perfectmatrixreference} is to use the algorithm given in Alg.~\ref{algo:QSPA}, which itself uses sampling.
This clearly shows that there are ways to move part of complexity of the circuit between different subroutines, in particular inside and outside of the $\PREP$s. We refer to these two broad strategies as \emph{implicit} and \emph{explicit} quantum rejection sampling, respectively.
\end{remark}

\noindent One concrete example of this remark has been implemented in Sec.~\ref{sec:PREPSELPREP} as two variants in Cor.~\ref{cor:Toeplitzf} and Cor.~\ref{cor:Toeplitzg}.
Quantum circuits for other block-encodings given in Secs.~\ref{sec:MatrixZiggurat}, \ref{sec:RowColumnBE}, \ref{sec:ColumnBE} can also be seen through the lens this remark, and modified accordingly, depending on whether it is convenient to perform the rejection sampling at the level of the reference matrix, or whether it should be pushed inside $\PREP_{\chi}$, $\PREP_{\phi}$, or even removed altogether.
The rest of this section is devoted to special cases of this algorithm for different classes of matrices, such as Toeplitz, hierarchical, etc., which naturally imply their own access models.
Each corresponding access model requires its own subroutines which are either one of the subroutines introduced in Section~\ref{subsec:DefinitionsAndAssumptionsBlockEncoding}, or new subroutines that are defined within its own subsection (such as the operation $\SEL$).

\subsection{Linear combination of unitaries: \texorpdfstring{$\PREP^\dagger-\SEL-\PREP$}{Lg} with implicit and explicit quantum rejection sampling}\label{sec:PREPSELPREP}

For a first class of block-encodings, we shall see that the general circuit in Fig.~\ref{fig:GeneralPurposeCircuitBE} includes the well-known linear combination of unitaries method~\cite{childs2012hamiltonian} and its combination with rejection sampling techniques. 
In fact, it recovers and extends techniques that have been used in quantum chemistry algorithms for the block-encoding of the Coulomb potential.\\

The linear combination of unitaries method first expresses the matrix $A$ as a linear combination of unitaries $A= \sum^L_{a=1} \alpha_a U_a$, where each $U_a$ is a unitary.
Then, one uses a $\PREP^\dagger-\SEL-\PREP$ type circuit which results in a block-encoding with rescaling factor $\alpha= \sum^L_{a=1} |\alpha_a|$, when the $\PREP$ and $\SEL$ subroutines are chosen that
\begin{align}
\PREP \ket{0}= \frac{1}{\sqrt{\alpha}}\sum^L_{a=1} \sqrt{|\alpha_a|}\ket{a}, \; \SEL= \sum^L_{a=1} \ketbra{a}{a} \otimes U_a.
\end{align}

\noindent Note that this block-encoding, while always possible, is not always efficient unless the structure of the matrix allows it and the decomposition into unitaries is chosen appropriately.
One case where this is efficient is when the matrix $A$ is the sum of polynomially many terms, each of which are of bounded strength, e.g., lattice Hamiltonians with bounded degree and locality, as examined in Refs.~\cite{berry2014exponential, berry2015simulating}.
Another case is Toeplitz matrices, where the matrix elements are of the form $A_{ij}= f(i-j)$, i.e., a function $f$ of only the difference between the row index $i$ and column index $j$ (a case in point being the Coulomb potential).

Below, we focus on the latter case, i.e., block-encoding Toeplitz matrices with the $\PREP^\dagger-\SEL-\PREP$ method.
We first present the method, and then showcase two different ways in which rejection sampling techniques can be included in the construction block-encodings: inside or outside PREP, i.e., \emph{implicit} or \emph{explicit} quantum rejection sampling.

\begin{corollary}[A $\PREP^\dagger$-$\SEL$-$\PREP$ block-encoding of Toeplitz matrices]\label{cor:Toeplitzf}

Let $A$ be a given complex Toeplitz matrix $A_{ij}= |f(\delta)|e^{i \varphi(\delta)}$ where $\delta= j-i$, with real-valued function $\varphi$.
Then, given access to $\PREP_{\sqrt{|f|}}: \ket{0} \mapsto \frac{1}{\sqrt{\alpha_f}} \sum^{N-1}_{\delta= -N+1} \sqrt{|f(\delta)|} \ket{\delta}$,
$\SEL: \ket{\delta} \ket{j} \mapsto \ket{\delta} \ket{j-\delta}$,
$\PHASE_\varphi: \ket{\delta} \mapsto e^{i\varphi(\delta)} \ket{\delta}$, and a comparator $\Comp$ that flags the cases $1\leq j-\delta \leq N$ with $\ket{0}$ and other cases with $\ket{1}$, a modification of Algorithm~\ref{algo:MBEA}, given in Figure~\ref{fig:ToeplitzBE}, gives a block-encoding of the matrix $A$ with a rescaling factor 
\begin{align}
    \alpha_f = \sum^{N-1}_{\delta= -N+1} |f(\delta)|.
\end{align} 
Namely, acting on any state of the form $\ket{0} \ket{\Psi}$, the algorithm outputs the state
\begin{align}
\ket{0} \frac{A}{\alpha_f}\ket{\Psi} + \ket{0^\perp} \ket{\ldots}.
\end{align}
\end{corollary}

\begin{proof}
For completeness, we give the well-known proof here. We go through the quantum circuit given in Fig.~\ref{fig:ToeplitzBE} and write the outputs after each stage explicitly.

\begin{enumerate}

\item After applying $\PREP_{\sqrt{|f|}}$ on the ancilla and expanding $\ket{\Psi} = \sum_{j=1}^N \psi_j \ket{j}$, we get
\begin{align}
\ket{\phi_1}= \ket{0}\sum^{N-1}_{\delta= -N+1} \frac{\sqrt{|f(\delta)|}}{\sqrt{\alpha_f}} \ket{\delta} \sum^N_{j=1} \psi_j \ket{j}.
\end{align}

\item After applying $\PHASE_\varphi$, we obtain the state
\begin{align}
\ket{\phi_2}= \ket{0}\sum^{N-1}_{\delta= -N+1} e^{i\varphi(\delta)} \frac{\sqrt{|f(\delta)|}}{\sqrt{\alpha_f}} \ket{\delta} \sum^N_{j=1} \psi_j \ket{j}.
\end{align}

\item After applying the $\SEL$ operation, we obtain the state
\begin{align}
\ket{\phi_3}= \ket{0}\sum^{N-1}_{\delta= -N+1} \sum^N_{j=1}  e^{i\varphi(\delta)}\frac{\sqrt{|f(\delta)|}}{\sqrt{\alpha_f}} \psi_j \ket{\delta} \ket{j-\delta}.
\end{align}

\item We then use the comparator to flag when the output index $i= j-\delta$ is in the given range $[1,N]$.
As a result, we obtain the state

\begin{align}
\ket{\phi_4}= \ket{0}\sum^{N-1}_{\substack{\delta= -N+1: \\j-\delta \in [1,N]}} \sum^N_{j=1}  e^{i\varphi(\delta)}\frac{\sqrt{|f(\delta)|}}{\sqrt{\alpha_f}} \psi_j \ket{\delta} \ket{j-\delta} + \ket{1}\sum^{N-1}_{\substack{\delta= -N+1: j-\delta \in \\ [-N+1,0] \cup [N+1, 2N-1]}} \sum^{N}_{j=1}  e^{i\varphi(\delta)}\frac{\sqrt{|f(\delta)|}}{\sqrt{\alpha_f}} \psi_j \ket{\delta} \ket{j-\delta}.
\end{align}

\item We finally apply $\PREP^\dagger_{\sqrt{|f|}}$, and end up with

\begin{align}
\ket{\phi_5}&= \ket{0}\ket{0}\sum^{N-1}_{\substack{\delta= -N+1: \\j-\delta \in [1,N]}} \sum^N_{j=1}  e^{i\varphi(\delta)}\frac{|f(\delta)|}{\alpha_f} \psi_j \ket{j-\delta} + \ket{1}\ket{0^\perp}\ket{\ldots}\\
 & = \ket{0}\ket{0}\sum^{N-1}_{\substack{\delta= -N+1: \\j-\delta \in [1,N]}} \sum^N_{j=1}  \frac{A_{j-\delta,j}}{\alpha_f} \psi_j \ket{j-\delta} + \ket{1}\ket{0^\perp}\ket{\ldots}\\
 &= \ket{0}\ket{0}\sum^{N}_{i=1} \sum^N_{j=1}  \frac{A_{ij}}{\alpha_f} \psi_j \ket{j} + \ket{1}\ket{0^\perp}\ket{\ldots}\\
&= \ket{0}\ket{0} \frac{A}{\alpha_f} \ket{\Psi} + \ket{1}\ket{0^\perp}\ket{\ldots}.
\end{align}
\end{enumerate}

\end{proof}

Note that the general circuit in Fig.~\ref{fig:GeneralPurposeCircuitBE} is seemingly different than the one given in Fig.~\ref{fig:ToeplitzBE}, however in essence it is a special case without any need to perform rejection sampling.
In fact, it corresponds to the scenario where $\nu =1$ and
\begin{align}
    \chi_{1 i j} = \phi_{1 j i} = \sqrt{|f(\delta)|} = \sqrt{|A_{i j}|},
\end{align}
where $\delta = j-i$, with the modifications of a new comparator, and a $\SEL$ instead of a $\SWAP$ operation. 
The reasons for these modifications are given as follows.
The above choice of $\chi$ and $\phi$ will not in general satisfy Eq.~\eqref{eq:coreBErelation} unless $A$ is circulant, which is a special class of Toepliz matrices.  
However, any Toeplitz matrix can be embedded as part of a circulant matrix. 
Using the indexing $(\delta=j-i, j)$ rather than $(i,j)$, this corresponds to extending the index range from $\{(i,j)\}_{i \in [1,N], j \in [1,N]}$ to $\{(\delta,j)\}_{\delta \in [-N+1,N-1], j \in [1,N]}$.
Furthermore, note that the Toeplitz structure enables us to replace the $\SWAP$ with a $\SEL$ operator, that only subtracts the difference $\delta \in \mathbb{N}$ from the input row index $j$, in order to get the column index $i= j-\delta$. Notably, the $U_A$ and $U_G$ unitaries are independent of the system register, with only the $\SEL$ subroutine coupling the system to the ancilla qubits, again due to the Toeplitz structure.

The extra comparator imposing the condition $1 \leq j -\delta \leq N$ is then placed in the circuit in Fig~\ref{fig:ToeplitzBE}, in order to flag the relevant part of the circulant matrix, and discard the rest.
These minor adjustments are the reason why Corollary~\ref{cor:Toeplitzf} has been proved independently.
Accounting for these minor adjustments, the cost of the block-encoding of the circulant matrix embedding the target Toeplitz is correctly predicted by Eq.~\eqref{eq:coreBErelation}. 
In fact, $\mathcal{N}_{\chi_j} = \mathcal{N}_{\phi_i} = \sqrt{\sum_{\delta=-N+1}^{N-1} |f(\delta)|}$ for all $i,j$, and so $\alpha_f = \sum_{\delta={-N+1}}^{N-1} |f(\delta)|$.

So far we have assumed access to $\PREP_{\sqrt{|f|}}$, which leaves open the question of how these unitaries are constructed. There are, in general, two main alternatives here.
A first alternative uses, e.g., the state preparation method introduced in Section~\ref{sec:AlgoState}, to construct an efficient quantum circuit for $\PREP_{\sqrt{|f|}}$ to be used in the block-encoding circuit.
Then, following the common $\PREP^\dagger-\SEL-\PREP$ constructions from Refs.~\cite{babbush2019quantum, camps2022explicit, sunderhauf2023block}, the circuit given in Fig.~\ref{fig:ToeplitzBE} block encodes the matrix $A$ with the rescaling factor $\alpha_f= \sum^{N-1}_{\delta= -N+1} |f(\delta)|$. 
This is the quantum circuit given as in Figure~\ref{fig:ToeplitzBE}, where inside the PREPs one may be using a quantum rejection sampling algorithm as given in Alg.~\ref{algo:QSPA}, with a carefully chosen reference function $g$. 
We refer to this as $\PREP^\dagger$-$\SEL$-$\PREP$ with \emph{implicit} quantum rejection sampling.

A second alternative is to realize the rejection sampling step explicitly outside the state preparation subroutines,
as given in Cor.~\ref{cor:Toeplitzg} and Fig.~\ref{fig:ToeplitzBEg}, which we refer to as $\PREP^\dagger$-$\SEL$-$\PREP$ with \emph{explicit} quantum rejection sampling.
Note that, for the special case of $f(\v{x})= 1/|\v{x}|$, this is also the block-encoding given in Refs.~\cite{babbush2019quantum, su2021fault}. This second alternative is formalized as follows:

\begin{figure}
    \centering
    \includegraphics[scale=0.65]{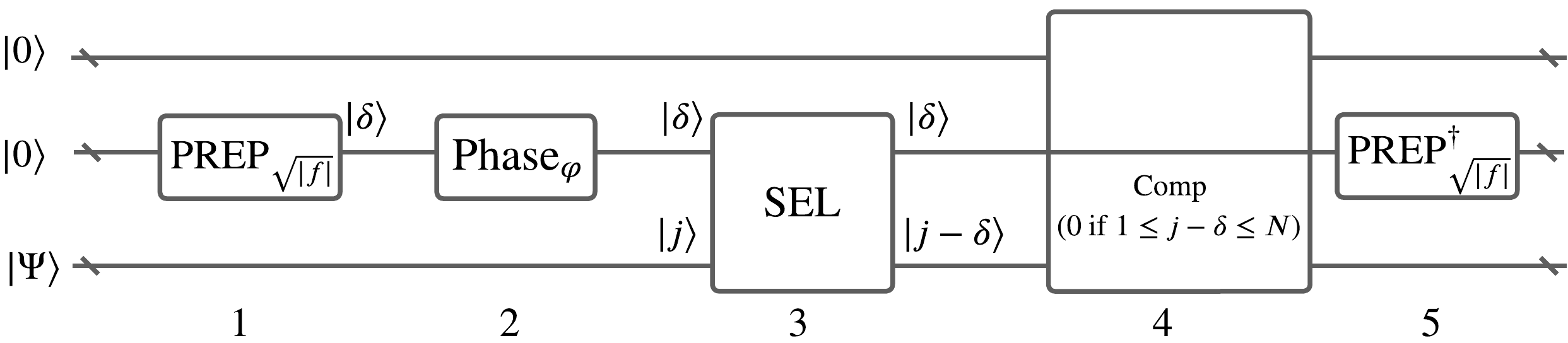}
    \caption{A $\PREP^\dagger-\SEL-\PREP$ type block-encoding circuit with implicit QRS for Toeplitz matrices $A$ where $A_{ij}= f(i-j)= |f(i-j)| e^{\varphi(i-j)}$.}
    \label{fig:ToeplitzBE}
\end{figure}

\begin{corollary}[$\PREP^\dagger$-$\SEL$-$\PREP$ block-encoding of Toeplitz matrices with \emph{explicit} quantum rejection sampling]\label{cor:Toeplitzg}
Let $A$ be a given complex Toeplitz matrix $A_{ij}= |f(\delta)|e^{i \varphi(\delta)}$ where $\delta= j-i$, with real-valued function~$\varphi$.
Let $g$ be a given real-positive valued function such that $g(\delta)> |f(\delta)|$ for all $\delta$.
Then, given access to $\PREP_{\sqrt{g}}: \ket{0} \mapsto \frac{1}{\sqrt{\alpha_g}} \sum^{N-1}_{\delta= -N+1} \sqrt{g(\delta)} \ket{\delta}$, 
$\SEL: \ket{\delta} \ket{j} \mapsto \ket{\delta} \ket{j-\delta}$, a comparator $\Comp$ that flags the cases $1\leq j-\delta \leq N$ with $\ket{0}$ and other cases with $\ket{1}$, $U_G$ and $U_A$ that compute the functions $g$ and $|f|$ to an ancilla, respectively, $U_G \ket{0} \ket{\delta} = \ket{g(\delta)} \ket{\delta}$, $U_A \ket{0} \ket{\delta} = \ket{|f(\delta)|} \ket{\delta}$ and 
$\PHASE_\varphi$, $\USP_M$,  $\Comp$ from Def.~\ref{def:SWAPandPHASE}, a modification of Algorithm~\ref{algo:MBEA}, given as in Figure~\ref{fig:ToeplitzBEg}, gives a block-encoding of the matrix $A$ with a rescaling factor
\begin{align}
\alpha_g= \sum^{N-1}_{\delta= -N+1} g(\delta).
\end{align}
Namely, acting on any state of the form $\ket{0} \ket{\Psi}$, the algorithm outputs the state
\begin{align}
\ket{0} \frac{A}{\alpha_g}\ket{\Psi} + \ket{0^\perp} \ket{\ldots}.
\end{align}

\end{corollary}

\begin{proof}
We go through the quantum circuit given in Fig.~\ref{fig:ToeplitzBEg} and write the outputs after each stage explicitly.

\begin{enumerate}

\item After applying $\PREP_{\sqrt{g}}$ and $\USP_M$ on the ancilla we get
\begin{align}
\ket{\phi_1}= \ket{0}\frac{1}{\sqrt{M}} \sum^M_{m=1} \ket{m} \ket{0} \ket{0}\sum^{N-1}_{\delta= -N+1} \frac{\sqrt{g(\delta)}}{\sqrt{\alpha_g}} \ket{\delta} \sum^N_{j=1} \psi_j \ket{j}.
\end{align}

\item After applying $U_G$ and $U_A$, we obtain the state
\begin{align}
\ket{\phi_{2'}}= \ket{0}\frac{1}{\sqrt{M}} \sum^M_{m=1} \ket{m} \ket{|f(\delta)|} \ket{g(\delta)}\sum^{N-1}_{\delta= -N+1} \frac{\sqrt{g(\delta)}}{\sqrt{\alpha_g}} \ket{\delta} \sum^N_{j=1} \psi_j \ket{j}.
\end{align}

\item After applying $\Comp$ and $\PHASE_\varphi$, we obtain the state
\small
\begin{align}
\ket{\phi_{3}}= \left[\ket{0}\frac{1}{\sqrt{M}} \sum^{\left\lfloor \frac{ M|f(\delta)|}{g(\delta)}\right\rfloor}_{m=1} \ket{m}  + \ket{1}\frac{1}{\sqrt{M}} \sum^{M}_{m=\left\lfloor \frac{ M|f(\delta)|}{g(\delta)}\right\rfloor+1} \ket{m} \right]\ket{|f(\delta)|} \ket{g(\delta)}\sum^{N-1}_{\delta= -N+1} \frac{e^{i\varphi(\delta)}\sqrt{g(\delta)}}{\sqrt{\alpha_g}} \ket{\delta} \sum^N_{j=1} \psi_j \ket{j}.
\end{align}
\normalsize

\item We then uncompute via $U^\dagger_A$ and $U^\dagger_G$, and obtain
\begin{align}
\ket{\phi_{4}}= \left[\ket{0}\frac{1}{\sqrt{M}} \sum^{\left\lfloor \frac{ M|f(\delta)|}{g(\delta)}\right\rfloor}_{m=1} \ket{m}  + \ket{1}\frac{1}{\sqrt{M}} \sum^{M}_{m=\left\lfloor \frac{ M|f(\delta)|}{g(\delta)}\right\rfloor+1} \ket{m} \right]\ket{0} \ket{0}\sum^{N-1}_{\delta= -N+1} \frac{e^{i\varphi(\delta)} \sqrt{g(\delta)}}{\sqrt{\alpha_g}} \ket{\delta} \sum^N_{j=1} \psi_j \ket{j}.
\end{align}

\item We apply $\SEL$, and get
\begin{align}
\ket{\phi_{5}}= \left[\ket{0}\frac{1}{\sqrt{M}} \sum^{\left\lfloor \frac{ M|f(\delta)|}{g(\delta)}\right\rfloor}_{m=1} \ket{m}  + \ket{1}\frac{1}{\sqrt{M}} \sum^{M}_{m=\left\lfloor \frac{ M|f(\delta)|}{g(\delta)}\right\rfloor+1} \ket{m} \right]\ket{0} \ket{0}\sum^{N-1}_{\delta= -N+1}\sum^N_{j=1} \frac{e^{i\varphi(\delta)} \sqrt{g(\delta)}}{\sqrt{\alpha_g}} \ket{\delta}  \psi_j \ket{j-\delta}.
\end{align}

\item We apply another comparator $\Comp$ that flags the system's output index to be in the range of $[1,N]$ and obtain
\begin{align}
\ket{\phi_{6}}= \ket{0}\sum^{\left\lfloor \frac{ M|f(\delta)|}{g(\delta)}\right\rfloor}_{m=1}  \sum^{N-1}_{\delta=-N+1} \sum^N_{j=1}\frac{e^{i\varphi(\delta)} \sqrt{g(\delta)}}{\sqrt{\alpha_g} \sqrt{M}} \psi_j \ket{m} \ket{0}\ket{0} \ket{\delta} \ket{j-\delta: j-\delta \in [1,N]}  + \ket{0^\perp} \ket{\ldots}.
\end{align}

\item Finally, after $\USP_M$ and $\PREP^\dagger_g$, we obtain the desired output given as
\begin{align}
\ket{\phi_{7}}&= \ket{0} \sum^{N-1, N}_{\substack{\delta=-N+1, j=1:\\ j-\delta \in [1,N]}} \frac{e^{i\varphi(\delta)} |f(\delta)|}{\alpha_g} \psi_j \ket{0} \ket{0}\ket{0} \ket{0} \ket{j-\delta}  + \ket{0^\perp} \ket{\ldots}\\
&= \ket{0}\frac{A}{\alpha_g} \ket{\Psi} + \ldots,
\end{align}
where $M$ is chosen sufficiently large along with computations performed sufficiently precisely, as in Lem.~\ref{lem:ChoosingM}, such that the block-encoding is achieved with the desired accuracy.
\end{enumerate}
\end{proof}

We emphasize again that Cor.~\ref{cor:Toeplitzf} (circuit given in Fig.~\ref{fig:ToeplitzBE}) and Cor.~\ref{cor:Toeplitzg} (circuit given in Fig.~\ref{fig:ToeplitzBEg}) give block-encodings of the same matrix, with different rescaling factors. 
The difference between these two circuits is that in Cor.~\ref{cor:Toeplitzf} the access model assumes a subroutine that prepares a state related to the original matrix elements given by the function $f$, while Cor.~\ref{cor:Toeplitzg} only assumes an access model that prepares a state related to another matrix whose elements are given by the bounding function $g$.
While the rescaling factor resulting in Cor.~\ref{cor:Toeplitzf} is, by construction, less than the rescaling factor resulting in Cor.~\ref{cor:Toeplitzg}, the state preparation subroutine $\PREP_{\sqrt{f}}$ is more costly than $\PREP_{\sqrt{g}}$.
In fact, the rejection sampling subroutine that is explicit in Fig.~\ref{fig:ToeplitzBEg} is absorbed in the state preparation subroutine in Fig.~\ref{fig:ToeplitzBE}.
While ultimately the comparison between the two versions depends on the specific case, and even though explicit quantum rejection sampling leads to a higher rescaling factor, we expect the explicit version to be about half as costly compared the implicit version, due to the fact that the costly subroutines are called about half as many times (i.e., only once in the whole algorithm, rather than once for $\PREP$ and once for $\PREP^\dagger$).

\begin{figure}
    \centering
    \includegraphics[scale=0.55]{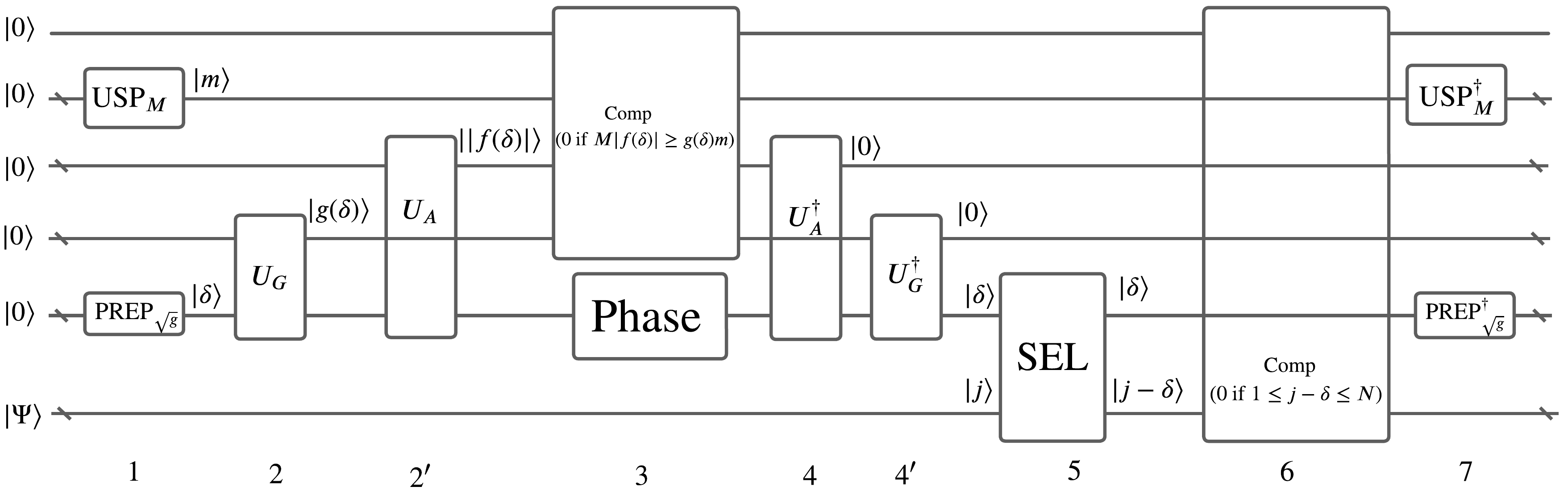}
    \caption{A $\PREP^\dagger-\SEL-\PREP$ type block-encoding circuit with explicit QRS for Toeplitz matrices $A$ where $A_{ij}= f(i-j)= |f(i-j)| e^{\varphi(i-j)}$.}
    \label{fig:ToeplitzBEg}
\end{figure}

\subsection{Submatrix-partitioning block-encoding: matrix ziggurat}\label{sec:MatrixZiggurat}

A second class of block-encodings that are special case of the general circuit in Fig.~\ref{fig:GeneralPurposeCircuitBE} are what we refer to as \emph{matrix ziggurats}. The basic idea is similar to what we have already discussed for quantum states. Modulo phases, a matrix $A$ is a function $(i,j) \mapsto |A_{ij}|$. We are looking for a `simpler' function $(i,j) \mapsto G_{ij}$ that upper bounds $|A_{ij}|$. Here we shall consider the case where $G$ is piece-wise constant across many of its indices (hence the name ``ziggurat''). 

Let's present this in more detail. We partition the matrix $A$ into disjoint submatrices, where each submatrix of $A$ is labeled by $k \in [1,\nu]$ for $\nu \in \mathbb{N}^+$. The submatrices are disjoint in the sense that they do not share any elements. This decomposition is captured by  a collection of lists, $\{S_k\}^{\nu}_{k=1}$. Each $S_k$ is a set of ordered pairs $(i,j) \in [1,N] \times [1,N]$, e.g., $S_2 = \{(1,2), (3,5), (5,3), (1,5)\}$. Each list $S_k$ contains the row and column indices $(i,j)$ of the original matrix that belongs to the submatrix labeled by $k$, i.e., if $A_{ij}$ belong to the submatrix labeled by $k$, then $(i,j)$ only appears in the list $S_k$. Note that in most cases we are interested in a setting where these lists are highly structured, e.g., $S_k$ contains all labels in the $k$th off-diagonals 
$S_k = \{(j,r) : |j-r| = k \text{ and } 1 \leq j,r \leq N \} $.

We further assume that the matrix elements of each submatrix labeled by $k$ is of absolute value at most $g_k$, i.e., for every $k$,
\begin{align}
|A_{ij}| \leq g_k \; \textrm{for all} \; (i,j) \in S_k.
\end{align}
For fixed $j$, let $I^{(k)}_j = \{ I^{(k)}_{rj} \}_{r=1}^{d_{kj}}$ be the set of row indices $r$ such that $(r,j) \in S_k$. Similarly, for fixed $i$ let $J^{(k)}_i=\{J^{(k)}_{ic}\}_{c=1}^{d'_{ki}}$ be the set of column indexes $c$ such that $(i,c) \in S_k$.
Note that, in general, the total number of row indices $|I^{(k)}_j|$ in column $j$ of submatrix $k$ depends on $j$. 

\begin{figure}[h]
    \centering
    \includegraphics[width=0.75\linewidth]{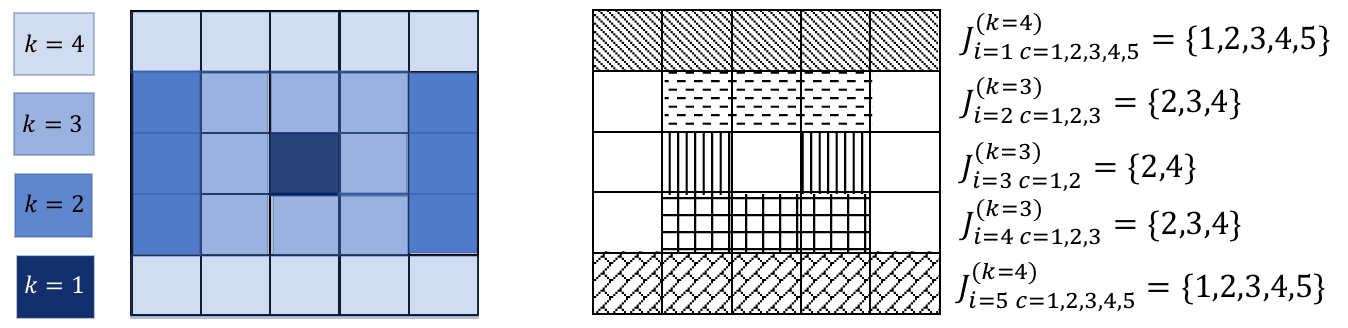}
    \caption{Left: Example of a matrix ziggurat with four regions $g_1>g_2>g_3>g_4$ with the largest values having darker shades. Right: nonempty sets $J^{(k)}_{ic}$ for $k=3,4$. Note that the sparsity parameters are $d_1=d'_1 =1$, $d_2 = 3$, $d'_2 =2$, $d_3 = d'_3 =3$, $d_4 =2$, $d'_4=5$. }
\label{fig:matrix_ziggurat}
\end{figure}

Similarly, the total number $|J^{(k)}_i|$ of column indices in row $i$ of submatrix $k$ depends on $i$. 
However, we can always include some zeros in the definition of submatrix $k$ so that there is a number $d_k$ satisfying $|I^{(k)}_j|=d_k$ for all $j$ and there is a number $d'_k$ such that $|J^{(k)}_i|= d'_k$ for all~$i$, see Fig.~\ref{fig:matrix_ziggurat} for a simple example to clarify the notation.
Define the states
\begin{align}
\ket{\chi_j} &\propto \sum^{\nu}_{k=1} \sum^{d_k}_{r=1} \left(\frac{d'_k}{d_k}\right)^{1/4} \sqrt{g_k} \ket{k} \ket{I^{(k)}_{rj}} , \nonumber \\
\ket{\phi_i} &\propto \sum^{\nu}_{k'=1} \sum^{ d'_{k'}}_{c=1} \left(\frac{d_{k'}}{d'_{k'}}\right)^{1/4} \sqrt{g_k} \ket{k'} \ket{J^{(k')}_{ic}} .
\label{eq:matrixzigguratstates}
\end{align}

In this case, in the general block-encoding method of Algorithm~\ref{algo:MBEA} $\PREP_\chi$ and $\PREP_\phi$ in Eq.~\eqref{def:prepchiprepphi}
prepare the above reference states. We will discuss below how to prepare these states.

As a result, we obtain the block-encoding of matrix $A$, with the rescaling factor $\alpha= \sum^{\nu}_{k=1} \sqrt{d_k d'_{k}} g_k$.

\noindent More precisely, the result is given as follows, with the explicit circuit in Fig.~\ref{fig:SubmatrixPartitioning}.

\begin{corollary}[A block-encoding by submatrix-partitioning - matrix ziggurat]\label{cor:SubmatrixPartitioning}
Let $A$ be a given matrix with complex entries $A_{ij}= |A_{ij}|e^{\varphi(i,j)}$ with real-valued function $\varphi$.
Let $\{S_k\}^{\nu}_{k=1}$ be a collection of lists of ordered pairs $(i,j) \in [1,N] \times [1,N]$ that lists the $d_k$ row and $d'_k$ column indices $(i,j)$ of the original matrix which belongs to the submatrix labeled by $k$, and assume $g_k$ are given such that
\begin{align}
|A_{ij}| \leq g_{k} \; \textrm{for all} \; (i,j) \in S_k.
\end{align}
Also, by adding a number of zeros to each $k$ submatrix, assume that the total number of row indices $|I^{(k)}_j|$ is independent of the column $j$, and the total number $|J^{(k)}_i|$ of column indices in row $i$ is independent of $i$. 

Let $\PREP_\chi\ket{0}\ket{0}  \ket{j} = \ket{\chi_j} \ket{j}$,  $\PREP_\phi  \ket{0}\ket{0} \ket{i} = \ket{\phi_i} \ket{j}$,
where $\ket{\chi_j}$, $\ket{\phi_i}$ are given in Eq.~\eqref{eq:matrixzigguratstates}, and $U_A \ket{0} \ket{i} \ket{j} = \ket{|A_{ij}|} \ket{i} \ket{j}$, $U_G \ket{0} \ket{k} = \ket{g_k} \ket{k}$.
Then Algorithm~\ref{algo:MBEA}, that specializes as in Fig.~\ref{fig:SubmatrixPartitioning}, gives a block-encoding of the matrix $A$ with a rescaling factor 
\begin{align}
\alpha_{\mathrm{zig}}= \sum^{\nu}_{k=1} \sqrt{d_k d'_{k}} g_k.
\end{align}
Namely, acting on any state of the form $\ket{0} \ket{\Psi}$, the algorithm outputs the state
\begin{align}
\ket{0} \frac{A}{\alpha_{\mathrm{zig}}}\ket{\Psi} + \ket{0^\perp} \ket{\ldots}
\end{align}
\end{corollary}

\begin{proof}
This is a direct application of Theorem~\ref{thm:mainBE}. From Eq.~\eqref{def:prepchiprepphi},
\begin{align}
    \chi_{kij} = \left(\frac{d'_k}{d_k}\right)^{1/4} \sqrt{g_k} \delta(i-I^{(k)}_{rj}), \\
  \phi_{kji} =  \left(\frac{d_k}{d'_k}\right)^{1/4} \sqrt{g_k} \delta(j-J^{(k)}_{ic}),
\end{align}
where $\delta(a-b)$ is $1$ if and only if $a=b$, and it is zero ortherwise.
Then, for every $i,j$,
\begin{align}
    \mathcal{N}_{\chi_j} = \sqrt{\sum_{k=1}^{\nu} \sqrt{\frac{d'_k}{d_k}}  g_k d_{k}} = \sqrt{\sum_{k=1}^\nu \sqrt{d_k d'_k} g_k}, \\
    \mathcal{N}_{\phi_i} = \sqrt{\sum_{k=1}^\nu \sqrt{\frac{d_k}{d'_k}}  g_k d'_{k}} = \sqrt{\sum_{k=1}^\nu \sqrt{d_k d'_k} g_k}. 
\end{align}
Hence, Eq.~\eqref{eq:coreBEInTheorem} is satisfied with $\alpha = \sum_{k=1}^\nu \sqrt{d_k d'_k} g_k$:
\begin{align}
   \sum_{k=1}^\nu \frac{\chi_{kij} \phi_{kji}}{\mathcal{N}_{\chi_j} \mathcal{N}_{\phi_i} } = \frac{g_k}{\sum_{k=1}^\nu \sqrt{d_k d'_k} g_k}.
\end{align}
Finally note that the matrix $G$ depends on $i$, $j$ only through $k$, hence we gave a cheaper implementation of $U_G$, which is acting on the label $k$ only, rather than on $\ket{i}$, $\ket{j}$. Taking this small change into account, the result follows from Theorem~\ref{thm:mainBE}.
\end{proof}

\begin{figure}
    \centering
    \includegraphics[scale=0.53]{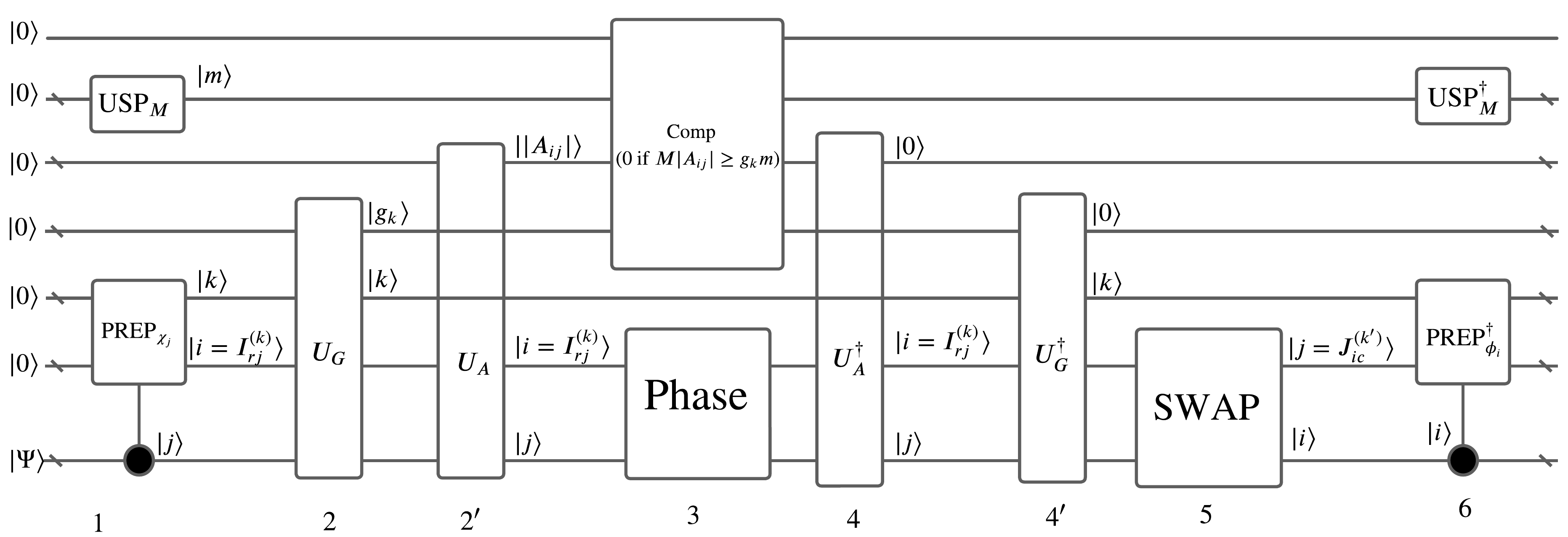}
    \caption{Quantum circuit for submatrix-partitioning block-encoding.}
    \label{fig:SubmatrixPartitioning}
\end{figure}

\begin{remark}[Constructing $\PREP_\chi$,  $\PREP_\phi$ in Eq.~\eqref{def:prepchiprepphi} for matrix ziggurat]
\label{remark:matrixzigguratreferencestates}
In Cor.~\ref{cor:SubmatrixPartitioning}, we assumed access to $\PREP_\chi$,  $\PREP_\phi$. 
This leaves open the question of how these state preparation unitaries can be constructed assuming a more natural access model. 
Let us discuss one such construction. 
Assume access to a unitary on $\lceil \log_2 \nu \rceil$ qubits acting as follows
\begin{align}
\label{eq:weightedupperboundsstate}
U_g \ket{0} = \frac{1}{\sqrt{\alpha_{\mathrm{zig}}}} \sum_{k=1}^\nu \sqrt{(d_k d_k')^{1/2} g_k} \ket{k}. 
\end{align}
We are typically, although not necessarily, interested in the case where $\nu = O(\log N)$, meaning we are splitting the matrix $A$ into $O(\log N)$ submatrices, and the unitary $U_g$ can be realized via any general state preparation routine with cost $O(\log N)$. If instead $\nu = O(N)$, the ``weighted upper bounds'' quantum state vector in Eq.~\eqref{eq:weightedupperboundsstate} must be efficient to prepare. This can be enforced with a careful design of the ziggurat (e.g., the ziggurat may correspond to the discretized version of an efficiently preparable function).

Furthermore, assume access to the set of row indices $I^{(k)}_{j}$ in the submatrix $k$ for a given column index $j$, and to the set of column indices $J^{(k)}_{i}$ in the submatrix $k$ for a given row index $i$, through the two unitaries
\begin{align}
O_I \ket{k} \ket{r} \ket{j} = \ket{k} \ket{I^{(k)}_{rj}} \ket{j}, \quad O_J \ket{k} \ket{c} \ket{i} = \ket{k} \ket{J^{(k)}_{ic}} \ket{i}.
\end{align}
This can be realized with cost $O(\log N)$ when there is an efficient procedure to compute the set of nonzero elements within each row/column of each submatrix. Note that to define the action of $O_I$, (respectively, $O_J$) on labels $r$ (respectively $c$) larger than $d_k$ (respectively, $d'_k$), we can introduce an extra qubit and simply map $r \mapsto r+ N$ (respectively, $c \mapsto c + N$), as in Lemma~48 of Ref.~\cite{gilyen2019quantum}.

Finally, consider the unitaries $\USP$, that given $k$ performs uniform state preparation over $d_k$ states of an auxiliary register, and $\USP'$, that given $k$ performs uniform state preparation over $d'_k$ states of an auxiliary register. 
Then,

\begin{align}
\ket{0} \ket{0} \ket{j} & \stackrel{U_g}{\mapsto}  \frac{1}{\sqrt{\alpha}} \sum_{k=1}^{\nu} \sqrt{(d_k d'_k)^{1/2} g_k} \ket{k} \ket{0} \ket{j}  \\ & \stackrel{\USP}{\mapsto} \sum_{r=1}^{d_k} \frac{1}{\sqrt{\alpha}} \sum_{k=1}^{\nu} \frac{\sqrt{(d_k d'_k)^{1/2} g_k}}{\sqrt{d_k}}  \ket{k} \ket{r} \ket{j}
\\ & \stackrel{O_I}{\mapsto} \sum_{r=1}^{d_k} \frac{1}{\sqrt{\alpha}} \sum_{k=1}^{\nu} \left( \frac{d'_k}{d_k}\right)^{1/4}  \ket{k} \ket{I_{rj}^{(k)}} \ket{j} = \ket{\chi_j} \ket{j}.
\end{align}
Hence, $\PREP_\chi= O_I \cdot \USP \cdot U_g$, and similarly $\PREP_\phi= O_J \cdot \USP' \cdot U_g$.
\end{remark}

This type of block-encoding is well-suited for matrices that have hierarchical structure, e.g., the Coulomb potential kernel matrix.
Indeed, we can see the hierarchical block-encoding that appeared in Ref.~\cite{nguyen2022block} as a special case of this type of block-encoding for a particular choice of matrix decomposition $A = \sum_{k=1}^\nu A^{(k)}$, where the submatrices $A^{(k)}$ are coherently sampled from the upper bound submatrices $G= \sum_{k=1}^\nu G^{(k)}$, where each $G^{(k)}$ is constant and equal to $g_k$ on hierarchically defined regions. 
Furthermore, as noted in Remark.~\ref{rem:SamplingOrNot}, the circuit could be designed such that no sampling is needed, in which case $\PREP_\chi$ and $\PREP_\phi$ generate states with coefficients $\chi_{kij}$ and $\phi_{kji}$ such that $\sum^\nu_{k=1} \frac{\chi_{kij} \phi_{kji}}{\mathcal{N}_{\chi_j} \mathcal{N}_{\phi_i}}= \frac{A_{ij}}{\alpha}$.

Another important special case is when $\nu=1$, i.e., there is a single submatrix which corresponds to the entire matrix $A$. 
In this case $g_k \equiv g_{\max}$ is a uniform upper bound on the elements of the whole matrix, $d_k \equiv d$ is the row sparsity of $A$ and $d'_k \equiv d'$ is the column sparsity. 
Then 
\begin{align}
    \alpha_{\mathrm{zig}} = \sqrt{d d'} g_{\max}, 
\end{align}
and with the construction in Remark~\ref{remark:matrixzigguratreferencestates} of $\PREP_{\chi}$, $\PREP_{\phi}$ we recover the sparse access block-encoding in Ref.~\cite{gilyen2019quantum}. Hence, another way to look at matrix ziggurats is as an extension of the well-known sparse access model in the form of a linear combination of sparse access oracles. Conversely, we can look at sparse access as a special case of matrix ziggurat rejection sampling with respect to a uniform reference.

\subsection{Row-column block-encoding}\label{sec:RowColumnBE}

Another class of block-encodings which constitute a special case of the general circuit in Fig.~\ref{fig:GeneralPurposeCircuitBE} are what we refer to as \emph{row-column} block-encodings. 
In this case, we assume access to quantum states that encode upper bounds to (square roots of) each column and row vectors, as well as to the corresponding normalizations. 
We give the quantum circuit in Fig.~\ref{fig:RowColumnBE}, which is a special case of the general circuit given in Fig.~\ref{fig:GeneralPurposeCircuitBE} and Algorithm~\ref{algo:MBEA}, with $\nu=4$. Define the reference states
\begin{align}
\ket{\chi_j} &= \left( \sqrt{\frac{\|G_{\cdot, j}\|_1}{\|G\|_\infty}} |0\rangle + \sqrt{1 - \frac{\|G_{\cdot, j}\|_1}{\|G\|_\infty}}|1\rangle \right) \ket{0} \sum^{N}_{i=1} \sqrt{\frac{G_{ij}}{\|G_{\cdot, j}\|_1}}  |i\rangle, \nonumber \\
\ket{\phi_i} &=  \ket{0} \left( \sqrt{\frac{\|G_{i, \cdot}\|_1}{\|G\|_1}} |0\rangle + \sqrt{1 - \frac{\|G_{i, \cdot}\|_1}{\|G\|_1}}|1\rangle \right) \sum^{N}_{j=1} \sqrt{\frac{G_{ij}}{\|G_{i, \cdot}\|_1}}  |j\rangle,
\label{eq:rowcolumnstates}
\end{align}
where the first two qubits label $k=1,2,3,4$, and the norms are defined as 
\begin{align}
&\| G_{i,\cdot}\|_1 = \sum^N_{j=1} |G_{ij}|, \quad \| G_{\cdot,j}\|_1 = \sum^N_{i=1} |G_{ij}|,\\
& \|G\|_1 := \max_i \| G_{i,\cdot}\|_1,  \quad \|G\|_\infty := \max_j \| G_{\cdot,j}\|_1. \label{eq:G1andGinfty}
\end{align}

\noindent These states define the $\PREP_\chi$ (that prepares $\ket{\chi_j}$ controlled on $j$) and $\PREP_{\phi}$ (that prepares $\ket{\phi_i}$ controlled on $i$), respectively, in Algorithm~\ref{algo:MBEA}, see Eq.~\eqref{def:prepchiprepphi}.
As a result, we obtain a block-encoding of matrix $A$, with the rescaling factor

\begin{align}
\alpha= \sqrt{\|G\|_1 \|G\|_\infty},
\end{align}
where $\|G\|_1$ and $\|G\|_\infty$ are given as in Eq.~\eqref{eq:G1andGinfty}. More formally:

\begin{corollary}[Row-column block-encoding]\label{cor:RowColumnBE}
Let $A$ be a given matrix with complex entries $A_{ij}= |A_{ij}|e^{\varphi(i-j)}$ with real-valued function $\varphi$, and $G$ be a matrix such that $|A_{ij}| \leq G_{ij}$ for all $i,j$.
Furthermore, let $\|G_{i,\cdot}\|_1,$  $\|G_{\cdot,j}\|_1,$ $\|G\|_1, \|G\|_\infty$,  defined as in Eq.\eqref{eq:G1andGinfty}, as the relevant norms related to the matrix $G$.
Let $\PREP_\chi\ket{0}\ket{0}  \ket{j} = \ket{\chi_j} \ket{j}$,  $\PREP_\phi  \ket{0}\ket{0} \ket{i} = \ket{\phi_i} \ket{i}$, where $\ket{\chi_j}$, $\ket{\phi_i}$ are given in Eq.~\eqref{eq:rowcolumnstates}, $U_A \ket{0} \ket{i} \ket{j} = \ket{|A_{ij}|} \ket{i} \ket{j}$, and $U_G \ket{0} \ket{i} \ket{j} = \ket{G_{ij}} \ket{i} \ket{j}$.
Then, Algorithm~\ref{algo:MBEA}, that specializes as in Fig.~\ref{fig:RowColumnBE}, gives a block-encoding of the matrix $A$ with a rescaling factor
\begin{align}
\alpha_{\mathrm{rc}}= \sqrt{\|G\|_1 \|G\|_{\infty}}.
\end{align}
Namely, when acting on any state of the form $\ket{0} \ket{\Psi}$, the algorithm outputs the state
\begin{align}
\ket{0} \frac{A}{\alpha_{\mathrm{rc}}}\ket{\Psi} + \ket{0^\perp} \ket{\ldots}. 
\end{align}
\end{corollary}
\begin{proof}
It suffices to check that the assumptions of Theorem~\ref{thm:mainBE} are satisfied. In this case, in the definition of the coefficients $\chi_{kij}$ and $\phi_{kji}$, the first two qubits correspond to the label $k$. 
Then
\begin{align}
   \chi_{1ij} = \sqrt{\frac{G_{ij}}{\| G\|_\infty}}, \quad \phi_{1ji} = \sqrt{\frac{G_{ij}}{\| G\|_1}}
\end{align}
Also, for every $i,j$, $\mathcal{N}_{\chi_j} = \sqrt{\sum_{i,k} | \chi_{kij}|^2} =1$,  $\mathcal{N}_{\phi_i} =  \sqrt{\sum_{j,k} | \phi_{kji}|^2} = 1$.
It follows that
\begin{align}
    \sum_k \frac{\chi_{kij} \phi_{kji}}{ \mathcal{N}_{\chi_j} \mathcal{N}_{\phi_i}} =  \frac{G_{ij}}{\sqrt{\|G\|_\infty \|G\|_1}} = \frac{G_{ij}}{\alpha_{\mathrm{rc}}},
\end{align}
as required by Eq.~\eqref{eq:coreBEInTheorem}. 
Applying Theorem~\ref{thm:mainBE}, we obtain a block-encoding of $A$  with rescaling prefactor of $\alpha_\mathrm{rc}= \sqrt{\|G\|_1 \|G\|_{\infty}}$.
\end{proof}

\begin{figure}
    \centering
    \includegraphics[scale=0.55]{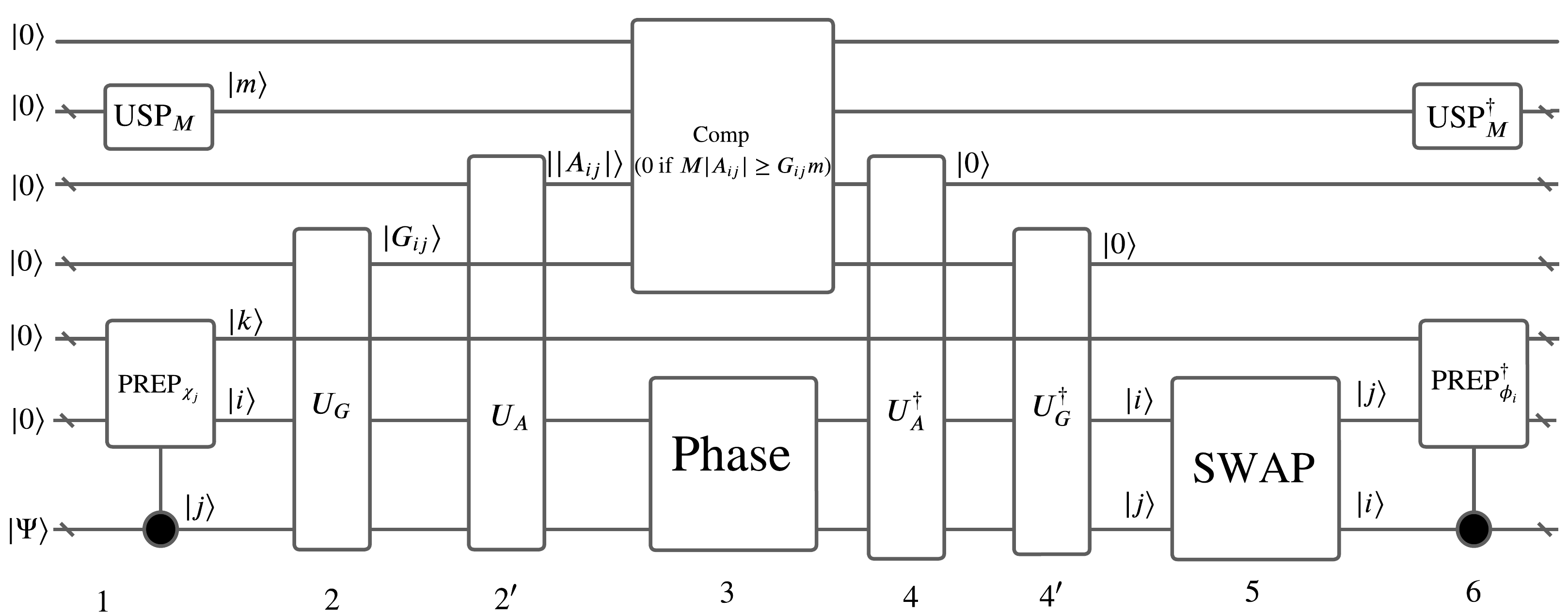}
    \caption{Quantum circuit for Row-Column block-encoding.}
    \label{fig:RowColumnBE}
\end{figure}

Let us now discuss the state preparations required by this block-encoding technique.

\begin{remark}[Constructing $\PREP_\chi$,  $\PREP_\phi$ in Eq.~\eqref{def:prepchiprepphi} for row-column block-encoding]
\label{remark:rowcolumnreferencestates}
In the previous algorithm we assumed access to $\PREP_\chi$,  $\PREP_\phi$, which prepare the states in Eq.~\eqref{eq:rowcolumnstates}. These can be rewritten as
\begin{align}
    \ket{\chi_j} = \ket{G_{\cdot,j}} \ket{0} \ket{C_j}, \quad   \ket{\phi_i} = \ket{0} \ket{G_{i,\cdot}} \ket{R_i},
\end{align}
where $\ket{C_j}$ and $\ket{R_i}$ encode square roots of the column and row upper bounds:
\begin{align}
   \ket{C_j} = \frac{1}{\sqrt{\|G_{\cdot,j}\|}} \sum_{i=1}^N  \sqrt{\frac{G_{ij}}{\|G_{\cdot, j}\|_1}}  |i\rangle, \quad \ket{R_i} = \frac{1}{\sqrt{\|G_{i,\cdot}\|_1}} \sum^{N}_{j=1} \sqrt{\frac{G_{ij}}{\|G_{i, \cdot}\|_1}}  |j\rangle,
\end{align}
while $\ket{G_{\cdot,j}}$, $\ket{G_{i,\cdot}}$ are qubit states that are needed to adjust the normalization of each row or column,
\begin{align}
    \ket{G_{\cdot,j}} = \sqrt{\frac{\|G_{\cdot,j\|_1}}{\|G\|_\infty}} \ket{0} + \sqrt{1-\frac{\|G_{\cdot,j\|_1}}{\|G\|_\infty}} \ket{1}, \quad \ket{G_{i,\cdot}} = \sqrt{\frac{\|G_{i, \cdot}\|_1}{\|G\|_1}} |0\rangle + \sqrt{1 - \frac{\|G_{i, \cdot}\|_1}{\|G\|_1}}|1\rangle.
\end{align}
Note that here we are requiring to be able to prepare $O(N)$ different states, which is generally inefficient even when each state preparation can be made efficient. However, in the presence of structure in the bounding matrix~$G$ (e.g., translational symmetry of a kernel matrix), all these states may be achievable by simple modifications of a single state preparation unitary. Also note that the qubit states $\ket{G_{\cdot,j}}$, $\ket{G_{i,\cdot}}$  can be constructed if $\|G\|_1$, $\|G\|_\infty$, $\| G_{i,\cdot}\|_1 \|$, $ \|G_{\cdot,j}\|_1$ can be efficiently computed into a register. The choice of $G$ needs to be made taking this requirement in mind.
\end{remark}

If we set $G_{ij} = A_{ij}$, and further assume that $A_{ij}$ is symmetric, the above protocol recovers Ref.~\cite{childs2010relationship}, Eq.~(27) and Ref.~\cite{berry2009black} as a special case. In fact, in the black-box block-encoding given in Ref.~\cite{berry2009black}, the required states $\ket{\chi_j}$ and $\ket{\phi_i}$ are obtained by a rejection sampling step happening inside the
$\PREP_\chi$ and $\PREP_\phi$ (we called this `implicit' rejection sampling in Remark~\ref{rem:SamplingOrNot}). 
However, Ref.~\cite{berry2009black} performs the rejection sampling with respect to a constant reference, hence a uniform state, which correponds to taking $G_{ij} = \max_{i,j} |A_{ij}|$, giving  $\alpha = d \max_{i,j} |A_{ij}|$, where $d$ is the sparsity of $A$. Here we extend these ideas by proposing to perform a more general rejection sampling, either inside the $\PREP$s, or explicitly as in Corollary~\ref{cor:RowColumnBE}.

\subsection{Column block-encoding}\label{sec:ColumnBE}

A final class of block-encodings considered here which are special cases of Theorem~\ref{thm:mainBE} are column block-encodings, whose quantum circuit is given in Fig.~\ref{fig:ColumnBE}.
It is a special case of the original circuit given in Fig.~\ref{fig:GeneralPurposeCircuitBE} and Algorithm~\ref{algo:MBEA} when $\nu =1$ and
\begin{align}
\label{eq:columnstates}
\ket{\chi_j} &=\sum^{N}_{i=1} \frac{G_{ij}}{\|G_{\cdot, j}\|}  |i\rangle, \quad \quad 
\ket{\phi} = \frac{1}{\|G\|_F} \sum^N_{j=1} \|G_{\cdot, j}\||j\rangle,
\end{align}
where 
\begin{align}\label{eq:FrobeniusNorm}
\|G_{\cdot, j}\|= \sqrt{\sum^N_{i=1} |G_{ij}|^2}, \quad \|G\|_F= \sqrt{\sum^N_{i,j=1} |G_{ij}|^2}.
\end{align}

These states define the $\PREP_\chi$ (that prepares $\ket{\chi_j}$ controlled on $j$) and $\PREP_{\phi}$ (that prepares $\ket{\phi}$ independent of $i$), respectively, in the Algorithm~\ref{algo:MBEA}. As a result, we obtain the block-encoding of matrix $A$, with the rescaling factor
\begin{align}
\alpha_{\mathrm{c}} = \|G\|_F,
\end{align}
where $\|G\|_F$ is given in Eq.~\eqref{eq:FrobeniusNorm}.
Formally, this block-encoding is given as follows.

\begin{corollary}[Column block-encoding]\label{cor:ColumnBE}
Let $A$ be a given matrix with complex entries $A_{ij}= |A_{ij}|e^{\varphi(i-j)}$ with real-valued function $\varphi$, and $G$ be a matrix such that $|A_{ij}| \leq G_{ij}$ for all $i,j$.
Furthermore, let $\|G_{\cdot, j}\|$,   $\|G\|_F$ defined as in Eq.\eqref{eq:FrobeniusNorm}, as the relevant norms related to the matrix $G$.
Let $\PREP_\chi\ket{0}  \ket{j} = \ket{\chi_j} \ket{j}$,  $\PREP_\phi \ket{0} \ket{i} = \ket{\phi_i} \ket{i}$, where $\ket{\chi_j}$, $\ket{\phi_i}$ are given in Eq.~\eqref{eq:columnstates}, $U_A \ket{0} \ket{i} \ket{j} = \ket{|A_{ij}|} \ket{i} \ket{j}$, and $U_G \ket{0} \ket{i} \ket{j} = \ket{G_{ij}} \ket{i} \ket{j}$. 
Then, Algorithm~\ref{algo:MBEA}, that specializes as in Fig.~\ref{fig:ColumnBE}, gives a block-encoding of the matrix $A$ with a rescaling factor $\alpha_{\mathrm{c}} = \|G\|_F$, namely, acting on any state of the form $\ket{0} \ket{\Psi}$, the algorithm outputs the state
\begin{align}
\ket{0} \frac{A}{\alpha_{\mathrm{c}}}\ket{\Psi} + \ket{0^\perp} \ket{\ldots}
\end{align}
\end{corollary}

\begin{proof}
It suffices to check that the assumptions of Theorem~\ref{thm:mainBE} are satisfied. The only obvious modification on the general algorithm has been to drop the qubits of the label $k$, since $\nu=1$ and hence this is a trivial label. Then
\begin{align}
    \chi_{1ij} = \frac{G_{ij}}{\|G_{\cdot,j}\|}, \quad \phi_{1j} = \frac{\|G_{\cdot,j}\|}{\|G\|_F}. 
\end{align}
Also, for every $i,j$,     $\mathcal{N}_{\chi_j} = \mathcal{N}_{\phi} = 1 $.
It follows that
\begin{align}
    \sum_k \frac{\chi_{kij} \phi_{kji}}{ \mathcal{N}_{\chi_j} \mathcal{N}_{\phi}} =  \frac{G_{ij}}{\|G\|_F}
\end{align}
as required by Eq.~\eqref{eq:coreBEInTheorem}. Applying Theorem~\ref{thm:mainBE}, we obtain a block-encoding of $A$ with rescaling factor \mbox{$\alpha_\mathrm{c}= \|G\|_F$}.
\end{proof}

\begin{figure}
    \centering
    \includegraphics[scale=0.54]{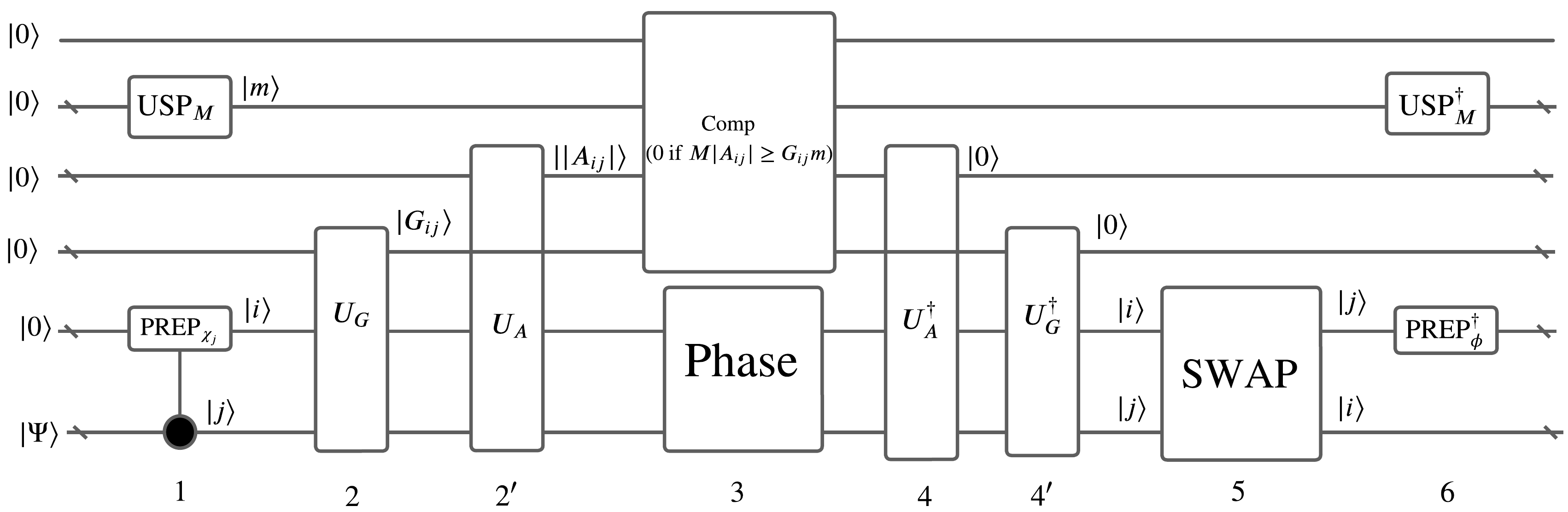}
    \caption{Quantum circuit for Column block-encoding.}
    \label{fig:ColumnBE}
\end{figure}

Let us now discuss the state preparations required by the previous block-encoding technique.

\begin{remark}[Constructing $\PREP_\chi$,  $\PREP_\phi$ in Eq.~\eqref{def:prepchiprepphi} for column block-encoding]
\label{remark:columnreferencestates}
In the previous algorithm we assumed access to $\PREP_\chi$,  $\PREP_\phi$, which prepare the states in Eq.~\eqref{eq:columnstates}. Similarly to the block-encoding technique discussed in the previous section, we are requiring to be able to prepare $O(N)$ different states, which is generally inefficient even when each state preparation can be made efficient. However, in the presence of structure in the matrix $G$, all these states may be achievable by simple modifications of a single state preparation unitary. Also note that the choice of $G$ needs to be made such that $\|G_{\cdot, j}\|$, $\|G\|_F$ can be efficiently computed.
\end{remark}

 We can compare the above with Ref.~\cite{wossnig2018quantum} and Ref.~\cite{kerenidis2016quantum} (pag. 14), where a block-encoding with $\alpha = \|A\|_F$ was achieved. These were invoked in recent end-to-end analysis of quantum linear solvers for finance applications, see Ref.~\cite{dalzell2023end}.  Our protocol can be understood as a quantum rejection sampling generalization of previous work. Our protocol allows for the possibility to trade a worse rescaling factor for a simpler access model. Prior work assumed the ability to efficiently prepare the $N+1$ states $\ket{\chi_j}$ and $\ket{\phi}$ in Eq.~\eqref{eq:columnstates} with $G=A$. As discussed, this is a strong assumption: computing the required norms of $A$ may be not efficient, and preparing $O(N)$ states is not efficient unless they have (and so $A$ has) a very special structure. In our protocol instead these assumptions are weakened, as it is only the boundary matrix $G$ that needs to have a special structure, not $A$ itself. It is not difficult to construct examples where the upper bounding matrix can be accessed efficiently via the above state preparations, even if the matrix $A$ itself cannot. Of course, we still need an efficient way of writing the matrix elements of $A_{ij}$ into a register given indexes $(i,j)$ as input, and the upper bound needs to be tight enough to ensure that $\alpha_c$ is not too large. 

\subsection{Example: A Toeplitz matrix with power-law decaying elements}\label{sec:ExampleBE}

In this section, we work out the rescaling factors for the block-encoding of the following Toeplitz matrix:

\begin{align}
A= \sum_{{\v{i} \neq \v{j}}\in [-N^{1/3}, N^{1/3}]^{\times 3}} \frac{1}{\|\v{i}-\v{j}\|^2} \ketbra{\v{i}}{\v{j}} + \sum_{\v{i} \in [-N^{1/3}, N^{1/3}]^{\times 3}} \ketbra{\v{i}}{\v{i}},
\end{align}
for the block-encoding methods we have enumerated in previous sections. 
Various versions of this appear in problems for block-encodings in quantum chemistry (Coulomb potential in momentum basis), PDEs, linear systems problems, etc.
Some of the block-encoding circuits perform much worse compared to others, and this is due to how well each particular method exploits the structure in the matrix.
Particularly for this matrix, which is dense, it is advantageous to exploit the Toeplitz structure, and the fact that the off-diagonal elements follow a power-law decay.
Hence, we can see, without any detailed analysis, that a naive approach based on the sparse access model is inefficient.
Two promising approaches are the $\PREP^\dagger - \SEL - \PREP$ method (see Section~\ref{sec:PREPSELPREP}) and the hierarchical block-encoding method~\cite{nguyen2022block}, which is a special case of the matrix-ziggurat method~(see Section~\ref{sec:MatrixZiggurat}).
For convenience, and without much loss of generality, for the rest of this section we assume $N^{1/3}= 2^\nu$ for $\nu \in \mathbb{N}^+$. \\

\noindent In this part, we present the block-encoding constant prefactors for these methods given in the previous subsections.
The block-encoding with the smallest rescaling factor is the row-column block-encoding, where it is approximately half of the rescaling factor of the $\PREP^\dagger-\SEL-\PREP$ block-encoding, due to the fact that the latter in fact block-encodes a larger circulant matrix and discards part of it.
The rescaling factor of the hierarchical block-encoding is in between the row-column and LCU methods.
The rescaling factor of the column block-encoding is instead asymptotically worse ($\mathcal{O}(N^{2/3})$) than the others ($\mathcal{O}(N^{1/3})$), because it does not exploit the relation between the rows of the matrix.
We leave the question of which one of the two methods is more qubit and gate efficient for future work, as it would require a detailed compilation of all the circuits involved.

\subsubsection{\texorpdfstring{$\PREP^\dagger - \SEL - \PREP$}{Lg} block-encoding}\label{sec:BEExamplePREPSELPREP}

The $\PREP$ and $\SEL$  in this case read
\begin{align}
\PREP \ket{0}=  \frac{1}{\sqrt{\alpha}}\left(\sum_{\v{\delta} \in D_0} \sqrt{\frac{1}{\|\v{\delta}\|^2}}\ket{\v{\delta}} + \ket{\v{0}} \right), \quad \SEL= \sum_{\v{\delta} \in D_0} \ketbra{\v{\delta}}{\v{\delta}} \otimes U_{\v{\delta}},
\end{align}
where
\begin{align}
D_0 = \{ (\delta_1, \delta_2, \delta_3)\neq \v{0} \ | \ \delta_i \in [ -2N^{1/3},2N^{1/3}], \ i=1,2,3\}.
\end{align}

\noindent In the algorithm, we use the following reference function $g$ that bounds the function $f(\v{\delta})= 1/|\v{\delta}|$ as follows, where we assume $N^{1/3}= 2^\nu$ for $\nu \in \mathbb{N}^+$ without loss of generality.

\begin{align}
g(\v{\delta})= \begin{cases}
1,& \; \textrm{for} \; \delta_1= \delta_2 = \delta_3 = 0,\\
g_k(\v{\delta})= 2^{-2(k-1)},& \; k \in [1,\nu +1]:  \; {2^{k-1} \leq \max\{|\delta_1|, |\delta_2|, |\delta_3|\} < 2^{k}}.
\end{cases}
\end{align}
The block-encoding rescaling factor is then
\begin{align}
\alpha = 1+ \sum^{\nu + 1}_{k=1} 2^{-2(k-1)} |D_k|.
\end{align}
The size of the region labeled by $k$ is $|D_k|= 2^{3} (2^{3k} - 2^{3(k-1)})= 7 \times 2^{3k}$. 
Hence, we get
\begin{align}
\alpha= 1 + 28 \sum^{\nu + 1}_{k=1}2^k \approx 112 N^{1/3}.
\end{align}

\subsubsection{Hierarchical block-encoding}\label{sec:BEExampleHierarchical}

Employing our assumption that $N^{1/3} = 2^\nu$, for some integer $\nu$, in the hierarchical block-encoding technique we decompose 
\begin{align}
    A = \sum_{k=2}^\nu A^{(k)} + A^{(1)},
\end{align}
according to the hierarchical decomposition described in Ref.~\cite{nguyen2022block}. At each level $k$, one looks at the interaction of any reference point with points in a number of non-adjacent blocks of size $2^{\nu-k} \times 2^{\nu-k} \times 2^{\nu-k}$. 
In $3D$, there are at most $6^3 - 3^3 = 189$ such blocks. 
In turns out that this fixes both row and column sparsity of $A^{(k)}$ to be $d_k = 189 \times 2^{3(\nu-k)}$. 
At the same time, since the interactions within blocks at level $k$ only involve points at least $2^{\nu-k}$ apart, the largest value of the potential at the $k$th level is $g_k = 2^{-2(\nu-k)}$. 
For the special adjacent block $k=1$, the sparsity is $d_{1} = 3^3 = 27$ and the maximum element is $g_{1} = 1$. 
Hence, we get
\begin{align}
\alpha_{\mathrm{zig}} = \sum_{k=1}^\nu d_k g_k = 27 + 189 \sum_{k=2}^\nu 2^{\nu - k} \approx 94.5 N^{1/3} 
\end{align}

\noindent We see that the scaling with $N$ in the two methods is the same, and the constant prefactors are similar.

\subsubsection{Row-column block-encoding}\label{sec:BEExampleRowColumn}

In this type of block-encoding, similarly to Sec.~\ref{sec:BEExamplePREPSELPREP} we use the following reference function $g$ that bounds the target function as follows, where we assume $N^{1/3}= 2^\nu$ for $\nu \in \mathbb{N}^+$ without loss of generality.

\begin{align}
G_{\v{ij}}= \begin{cases}
1,& \; \textrm{for} \; \v{i}= \v{j},\\
2^{-2(k-1)},& \; k \in [1,\nu]: \v{\delta}= \v{i} - \v{j},  \; {2^{k-1} \leq \max\{|\delta_1|, |\delta_2|, |\delta_3|\} < 2^{k}}.
\end{cases}
\end{align}
\noindent The block-encoding rescaling factor is then
\begin{align}
\alpha_{\mathrm{rc}} = \sqrt{\|G\|_1 \|G\|_\infty},
\end{align}
where, in fact, for this matrix,
\begin{align}
\|G\|_1= \|G\|_\infty\leq 1 + 56\sum^\nu_{k=1} 2^k,
\end{align}
which is nothing but the sum of rows, or columns, of the matrix $G$, where each element of magnitude $2^{-2(k-1)}$ appears with multiplicity $|D_k|= 7 \times 2^{3k}$.
Hence, we find that,
\begin{align}
\alpha_{\mathrm{rc}} \approx 56 N^{1/3}.
\end{align}
Out of the three methods considered so far, this is the one with the smallest prefactor, by almost a factor of~$2$. Whether or not this improvement translates to a similar improvement at the gate level requires a detailed compilation. 

\subsubsection{Column block-encoding}\label{sec:BEExamplewColumn}
In this type of block-encoding we use again the following reference function $g$ that bounds the target function as follows, where we assume $N^{1/3}= 2^\nu$ for $\nu \in \mathbb{N}^+$ without loss of generality.
\begin{align}
G_{\v{ij}}= \begin{cases}
1, & \; \textrm{for} \; \v{i}= \v{j},\\
2^{-2(k-1)},& \; k \in [1,\nu]: \v{\delta}= \v{i} - \v{j},  \; {2^{k-1} \leq \max\{|\delta_1|, |\delta_2|, |\delta_3|\} < 2^{k}}.
\end{cases}
\end{align}
\noindent The block-encoding rescaling factor is then
\begin{align}
\alpha_{\mathrm{c}} = \sqrt{\sum_{\v{i},\v{j}} |G_{\v{ij}}|^2},
\end{align}

\noindent which turns out to be equal to 
\begin{align}
\alpha_{\mathrm{c}} = \sqrt{N \|G\|_1}
\end{align}
where $\|G\|_1\leq 1+ 56 N^{1/3}$.
Hence we find

\begin{align}
\alpha_{\mathrm{c}} \approx \sqrt{56} N^{2/3}.
\end{align}

\noindent Note that, this is asymptotically worse than any previous rescaling factors, due to the fact that it does not exploit the structure that is given by the relationship between the rows.

\section{Conclusions}\label{sec:Conclusions}

A leading cost of many quantum algorithms originates from initializing a quantum state encoding information about the problem at hand (initial conditions in a Hamiltonian simulation problem, the known vector of a nonhomogeneous linear system etc.). Similarly, many quantum algorithms also require block-encoding a matrix involving other problem parameters (the Hamiltonian in a quantum simulation problem, the matrix of coefficients of a linear system, etc.). This data loading problem can be tackled efficiently only if we can somehow exploit the structure of the problem to avoid incurring in worst-case cost by exploiting some underlying structure in the data.  \\

In this work, we introduced a systematic QRS-based method for efficient state preparation and matrix block-encoding, when there is exploitable structure in the coefficients of the state and the matrix elements, respectively.
While this kind of approach underlies the so-called black-box methods, beginning in the early days for the state preparation problem and more recently for the matrix block-encoding problem, the systematic approach we follow makes it explicit that using better reference states improves the efficiency drastically, as it is the case for classical rejection sampling.
Furthermore, this point of view also enables us to see the matrix block-encoding problem for matrices with structured elements through the same lens.
As a result, we clearly see how to make the subroutines of many quantum algorithms that use structured data fit into QRS-based methods.
In particular, for successful application of the QRS-method for a given instance of a state-preparation or matrix block-encoding problem, the data structure should allow for efficient function computations and efficiently preparable reference states.\\

We furthermore instantiated our subroutines for a few example cases, such as states with power-law, Gaussian, and hyperbolic tangent coefficients, and obtained explicit gate counts showcasing how our QRS-based methods work in practical situations.
We compared this to the Toffoli gate counts that one would obtain from an optimized version of the Grover-Rudolph state preparation algorithm, assuming arbitrary data structure.
While such comparisons should be taken with a grain of salt, given the diversity of assumptions made by each method, we demonstrate that the QRS-based method provides a sufficiently versatile and efficient alternative to arbitrary state preparation methods, especially, when the target state is far from random.
For the matrix block-encoding problem we have computed the rescaling factor for a particular Toeplitz matrix with different types of access models, and showed how the rejection sampling can be moved in and out of the state preparation subroutines.
These methods can be seen as generalizations of block-encoding methods that have appeared in the literature.
Similar to the case of state preparation, common examples of matrices that are used in applications are structured matrices where the matrix elements can be computed efficiently.
Our methods allow to leverage this structure in various ways, which can avoid the formidable costs of commonly used methods such as sparse access models.
\\

There are several future directions.
As highlighted in the text, the ultimate optimization of gate count, and the decision of which method to use is a non-trivial and cumbersome problem.
While we have put forth guidelines in Sec.~\ref{subsec:GeneralComparison} to compare different methods, if the final goal is to optimize a given metric for a given fault-tolerant architecture, one should perform the entire algorithmic and compilation analysis for each method and each problem instantiation.
While some architectures with certain assumptions may favor minimizing the Toffoli count, others may favor optimizing space-time volume, or as more recently introduced, the active volume~\cite{litinski2022active}.
Given a fault-tolerant architecture with its constraints specified, it would be interesting to study the cost of each method for common target states and matrices.
As also emphasized in the main text, merging methods may provide lower resource counts.
Beyond these practicalities, it would be valuable to study a characterization or classification of target functions that are well-suited to different methods, such as QRS-based, MPS-based, Grover-Rudolph, etc.
While there are some guidelines, and a related program that was started in Refs.~\cite{aaronson2004multilinear,aaronsonZoo}, it would be fruitful to pursue the question further.
Another direction, alluded to in the recent works Refs.~\cite{ozols2013quantum, holmes2022quantum}, is more generally applicable to quantum algorithms.
Ref.~\cite{ozols2013quantum} discusses the advantage of using a better (than uniform) reference state, in a scenario where the state index $\ket{x}$ comes paired with an unknown state $\ket{\psi_x}$, such as an energy eigenstate with energy $x$. 
This is useful, for example, in preparing thermal states.
Most of the struggle in preparing thermal states arises from the fact that the target state is sampled out of the uniform state or a purified infinite temperature state.
A special case of the general algorithm in Ref.~\cite{holmes2022quantum} uses the quantum rejection sampling ideas together with results from fluctuation theorems~\cite{jarzynski1997nonequilibrium} to devise an efficient block-encoding of the exponential of the work operator.
This enables one to sample the target purified thermal state from a simple purified thermal state that is also much `closer' to the target state compared to the infinite temperature state, thus improving the query complexity of the algorithm drastically.
Generalizing these and similar methods to other problems beyond thermal state preparation using the results of the present investigation would be an exciting direction for future research.

\section*{Author contributions and acknowledgments}
Author list and contributions are alphabetical.
SP, KS, B\c{S} developed the quantum state preparation algorithm, 
ML, B\c{S} developed the matrix block-encoding algorithm.
JL, WP, SS, B\c{S} performed technical analysis and resource counts.
All authors contributed to discussing, reading, and writing the manuscript.
We acknowledge clarifying discussions with Terry Rudolph, and are thankful for collaboration with our colleagues at PsiQuantum on related and other topics.

\newpage

\begin{appendix}

\section{Potential modifications}

\subsection{The circuit with general clauses}\label{app:GeneralClause}

The final state can be obtained using clauses that are different from checking the truth value of $g(x) m \leq |f(x)| M$.
In some cases, it is more efficient to use functions $\tilde{f}$ and $\tilde{g}$ with potential modifications to the clause $C$. 
The quantum circuit looks almost exactly the same, except modifications to the subroutines $U_{f}$, $U_{g}$, and $\Comp$, with the circuit is given in Figure~\ref{fig:GeneralPurposeCircuitGeneralClause}.
The algorithm works as long as the clause and the functions are chosen such that

\begin{align}
\sum^M_{m=1} \delta_{C(\tilde{f}(x), \tilde{g}(x), m),0} \approx  M \frac{f(x)}{g(x)}
\end{align}
for all $x \in D$.
In fact, the crucial point is to get the ratio of $Mf(x/g(x))$ correct so that when multiplied with $g(x)/M$, the target state follows (with $g(x)$ coming from the reference state, and $1/M$ coming from two $\USP$'s for creating and projecting on the sampling state $\frac{1}{\sqrt{M}} \ket{m}$).

Implementations of this idea can be seen in our examples in Sec.~\ref{subsec:ExamplesPowerLaw}, Sec.~\ref{subsec:ExamplesGaussian}, and Sec.~\ref{subsec:ExamplesTanh}.
For example, for $f(x)= 1/x$ in Sec.~\ref{subsec:ExamplesPowerLaw} with reference function given as in Eq.~\eqref{eq:powerlaw_ziggurat}, we could use a few clauses that would lead to the same result.
One option is $m/2^{(\mu - 1)} \leq M/x$, which would require computing inverses of $x$ and $2^{\mu(x)-1}$, and multiplying them with $M$ and $m$.
While on the other hand, the same result is obtained, i.e., the same $m$'s are flagged with $\ket{0}$ and otherwise with $\ket{1}$, when the clause is chosen to be $mx \leq M2^{(\mu -1)}$.
Note that $x$ is already registered as the system qubit, and computing $2^{(\mu -1)}$ is especially easy.
A similar observation can be made for the example given in Sec.~\ref{subsec:ExamplesTanh}.
Instead of computing $\tanh(x)$ and running the clause $m/2 \leq M \tanh(x)$, we compute $e^{2x}$ and run the clause $1 \leq e^{2x}(1-m/(2M)) - m/(2M)$.
The quantum circuit with general clauses for state preparation  is given in Fig.~\ref{fig:GeneralPurposeCircuitGeneralClause}, and similar generalization applies to matrix block-encoding circuit as well.

\begin{figure}
    \centering
    \includegraphics[scale=0.52]{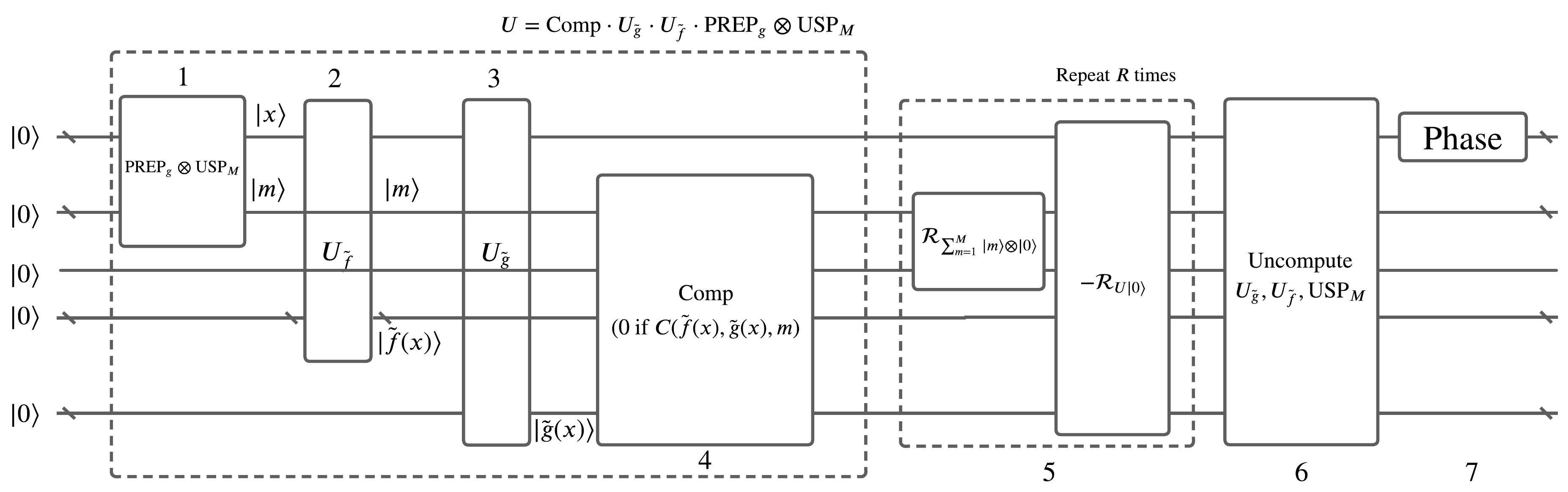}
    \caption{The general purpose state preparation algorithm that uses a clause $C$ which uses functions $\tilde{f}$ and $\tilde{g}$, such that the comparator $\Comp_C$ with clause $C$ satisfies $g(x)\sum^M_{m=1} \delta_{C(\tilde{f}(x), \tilde{g}(x), m),0}/M \approx f(x)$, where the approximation is up to the errors as analyzed in Lemma~\ref{lem:ChoosingM}.}
    \label{fig:GeneralPurposeCircuitGeneralClause}
\end{figure}

\subsection{Replacing comparators with coherent rotations}

We can replace the coherent sampling that includes use of a comparator and a sampling state with a coherent rotation.
Notice that the coherent sampling effectively multiplies the amplitude $g(x)$ with the ratio $|f(x)|/g(x)$.
Another way to realize this is to compute the angle $\alpha(x)= \arcsin(f(x)/g(x))$ coherently depending on $x$.
We then add an additional ancilla that is rotated along the $Y$-axis with an angle $\alpha(x)$.
Note that this does not use an additional $M$ dimensional sampling space, instead it uses $1$ qubit for the coherent rotation. 
However we expect this replacement to yield a higher cost due to a coherent implementation of $\arcsin(f(x)/g(x))$ rather than a simple comparator.
In fact, this is the subject of Ref.~\cite{sanders2019black} where they eliminated the costly subroutine of computing $\arcsin$ with a simpler inequality test.
This is due to comparator cost being only linear in the number of bits, while the cost of implementation of $\arcsin$ being a polynomial of number of bits, that is usually a degree of $d \gg 1$. 
Quantum circuits for the state preparation, and block-encoding problems with these modifications are given in Fig.~\ref{fig:GeneralPurposeCircuitCohRot} and Fig.~\ref{fig:GeneralPurposeBECircuitCohRot} respectively.

\begin{figure}
\centering
\begin{subfigure}[t]{1\textwidth}
    \centering
    \includegraphics[width=0.9\linewidth]{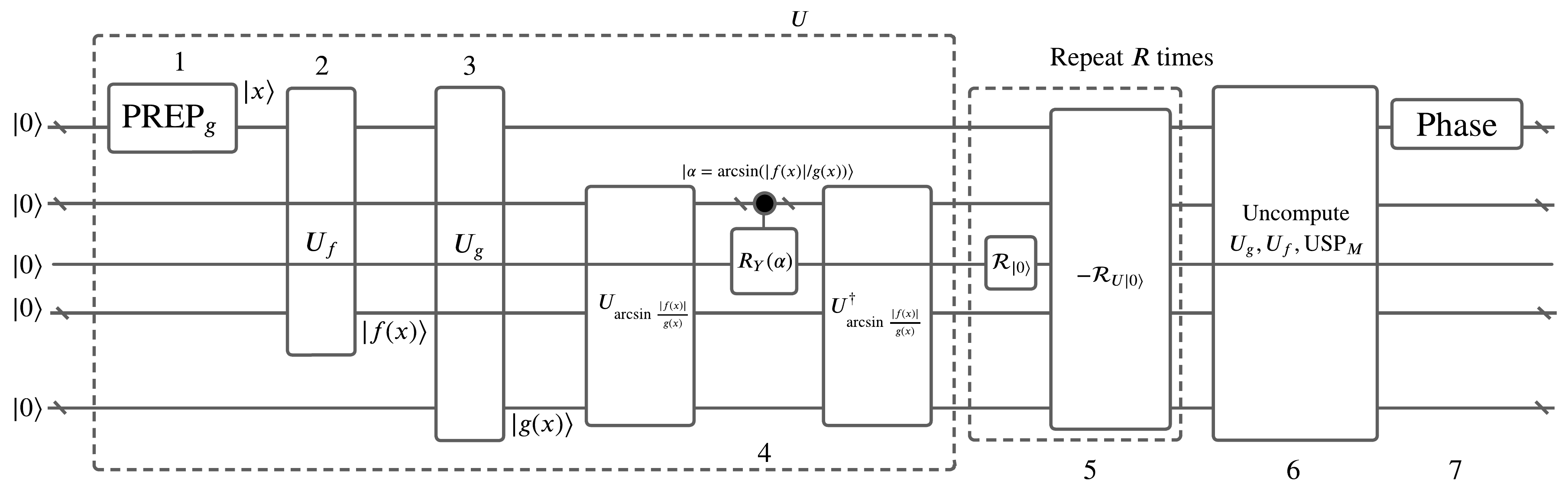}
    \caption{An alternative for a quantum circuit for the general purpose state preparation.
    }
    \label{fig:GeneralPurposeCircuitCohRot}
\end{subfigure}\\ \vspace{1em}%
\begin{subfigure}[t]{1\textwidth}
  \centering
  \includegraphics[width=0.9\linewidth]{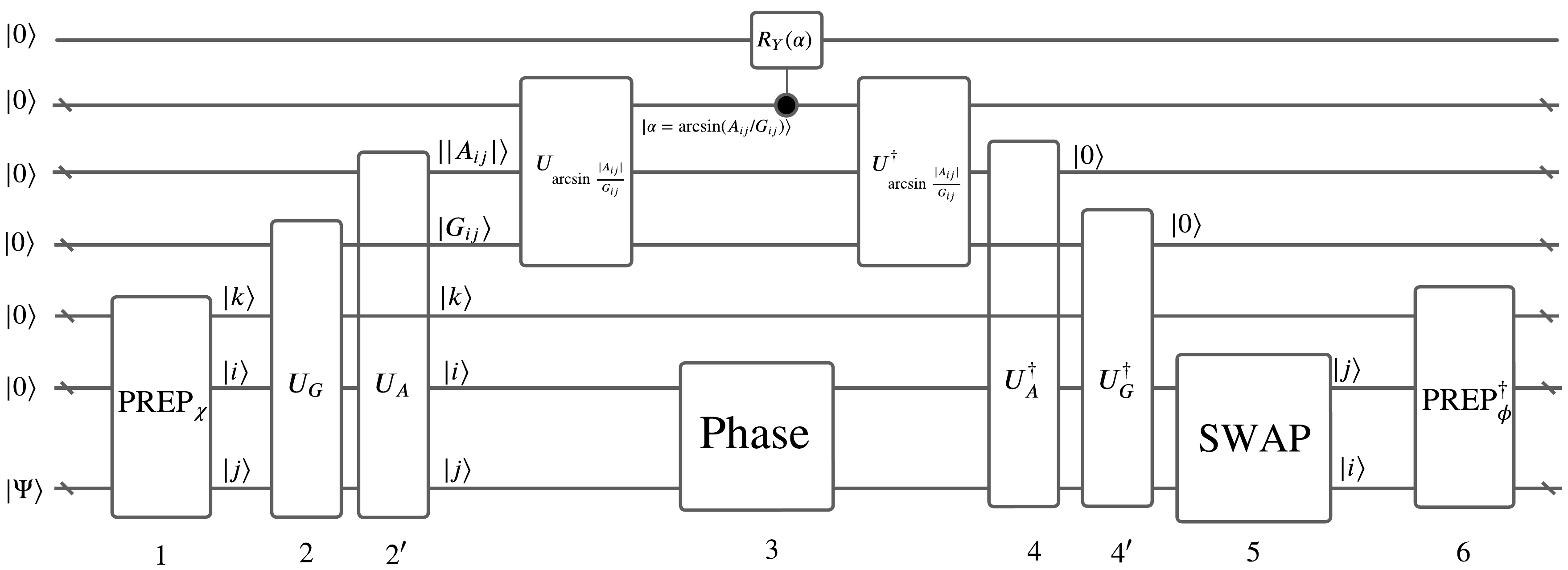}
    \caption{An alternative for a quantum circuit for the general purpose block-encoding.
    }
  \label{fig:GeneralPurposeBECircuitCohRot}
\end{subfigure}\\ \vspace{1em}%
\caption{Alternative circuit implementations of the general purpose (a) state preparation and (b) block-encoding, where the comparator has been replaced by coherent computation of angles via arcsin and controlled rotations.}
\label{fig:QRSCircuitCohRot}
\end{figure}

\section{Meshed Ziggurat Example}

We consider a two-dimensional function $f(x,y)$ and the mesh reference state $g(x,y)$ on the domain $x,y\in[0,1]\times[-4,4]$,

\begin{align}
f(x,y) & =|e^{-xy+\sin(y)}\sin(x(1-x)^{4})+\frac{1}{3^{2}}\sin(3^{2}(x^{2}+yx))+\frac{1}{4^{2}}\sin(4^{2}(x^{2}+yx))|\\
g(x,y) & =\max_{x,y\in i_{xy}}f(x,y),\forall i_{xy}\in I_{X}\times I_{Y}
\end{align}

\noindent
where have the set of intervals in each dimension $I_X={[x_i,x_i+1),i=0,...,n_X-1}$ and $I_Y={[y_i,y_i+1),i=0,...,n_Y-1}$, constructed from the points $x_i=x_0+\frac{L_X}{n_X},i=0,...,n_X$ and $y_j=y_0+\frac{L_Y}{n_Y},j=0,...,n_Y$, where $L_X=1,L_Y=8$ are the lengths of the domains in each direction and $x_0=0,y_0=-4$ are their left endpoints, respectively, and $n_X,n_Y$ are the number of domains desired in each direction. Increasing $n_X,n_Y$ makes the mesh resolution finer and therefore incrementally increases the resources required to prepare the ziggurat. For simplicity, we parameterize the mesh with a single positive integer $n$ as $n_X=\lceil\frac{5}{4}n\rceil,n_Y=\lceil2n\rceil$ yielding the total number of domains as $n_X n_Y$.

In Fig.~\ref{fig:meshExample}, on the leftmost panel we see that the function is highly oscillatory and irregular, making it difficult to see an apparent structure to exploit via other state preparation methods. In the rightmost panel, we see that the success probability of the mesh rapidly crosses the $0.25$ threshold, i.e. the threshold where only $1$ round of amplification is needed, by $n=3$. The middle panel shows the $n=3$ mesh which contains only $24$ domains. This is an unintuitively simple reference state given the complicated structure of the function. A detailed resource estimate for this example is beyond the scope of this work but it suffices to say that the cost will likely be bottlenecked by the coherent computation of the function $f(x,y)$. As long as it is possible to compute the function $f$ efficiently (e.g. if $f$ is fit to data, say in the example of initial conditions for partial differential equations), even if $f$ contains a lot of features that appear complex, constructing mesh reference states can be deceptively simple and cheap. This illustrative example demonstrates the conceptual simplicity and versatility of the rejection sampling approach.

\label{app:meshExample}
\begin{figure}
\centering
\includegraphics[width=.99\linewidth]{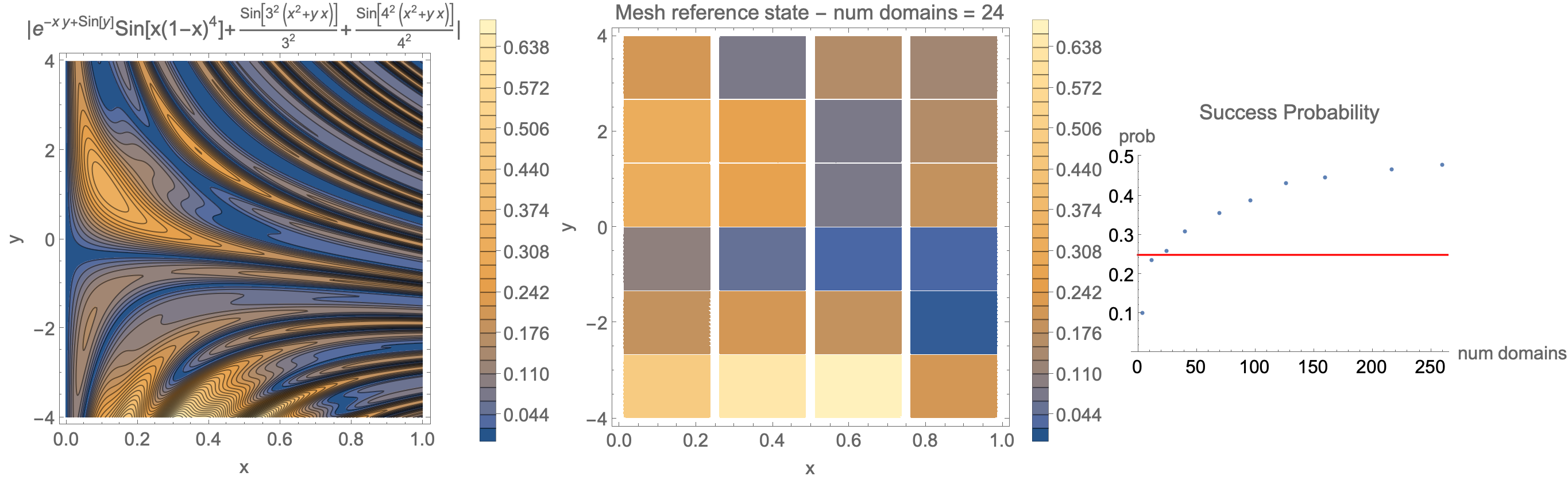}
  \caption{Left to right: Two-dimensional function $f(x,y)$, mesh reference state $g(x,y)$ with $24$ domains (corresponding to $n=3$ as discussed in the text), and the success probability as a function of the number of domains (mesh resolution). The red line indicates a success probability of $0.25$ meaning that mesh sizes (as per our meshing prescription in the text) with $24$ boxes or higher only require 1 round of amplitude amplification. This is an unintuitively simple reference state given the complicated structure of the function.}
  \label{fig:meshExample}
\end{figure}

\end{appendix}

\newpage
\bibliographystyle{unsrt}
\bibliography{main}
\end{document}